\documentclass[journal]{IEEEtran}
\usepackage{amsmath,amsfonts}
\usepackage{algorithmic}
\usepackage{algorithm}
\usepackage{array}
\usepackage[caption=false,font=normalsize,labelfont=sf,textfont=sf]{subfig}
\usepackage{textcomp}
\usepackage{stfloats}
\usepackage{url}
\usepackage{diagbox}
\usepackage{verbatim}
\usepackage{graphicx}
\usepackage{cite}
\usepackage{footnote}
\usepackage{pifont}
\usepackage{makecell}
\usepackage{multirow}
\newtheorem{lemma}{Lemma}
\newtheorem{definition}{Definition}
\newtheorem{theorem}{Theorem}
\newtheorem{proof}{Proof}

\begin{document}

\title{Efficient Quantum-resistant Delegable Data Analysis Scheme with Revocation and Keyword Search in Mobile Cloud Computing}

\author{Yue Han,~\IEEEmembership{}Jinguang Han,~\IEEEmembership{Senior Member,~IEEE}, and Jianying Zhou
        
\thanks{Yue Han is with the School of Cyber Science and Engineering, Southeast University, Nanjing 210096, China (e-mail: yuehan@seu.edu.cn).}
\thanks{Jinguang Han is with the School of Cyber Science and Engineering,
Southeast University, Nanjing 210096, China, and also with the Wuxi Campus,
Southeast University, Wuxi 214125, China (e-mail: jghan@seu.edu.cn).}
\thanks{Jianying Zhou is with the Singapore University of Technology and Design,
Singapore 487372 (e-mail: jianying\_zhou@sutd.edu.sg).}}



\maketitle

\begin{abstract}
With the rapid growth of smart devices and mobile internet, large-scale data processing is becoming increasingly important, 
while mobile devices remain resource-constrained. 
Mobile Cloud Computing (MCC) addresses this limitation by offloading tasks to the cloud. 
Nevertheless, the widespread adoption of MCC also raises challenges such as 
data privacy, selective computation, efficient revocation, and keyword search.
Additionally, the development of quantum computers also threatens data security in MCC. 
To address these challenges, we propose an efficient quantum-resistant delegable data analysis scheme with revocation and keyword search (EQDDA-RKS) for MCC. 
In the proposed scheme, an authorised mobile device can perform keyword searches and compute inner product values over encrypted data without disclosing any additional information. 
Meanwhile, if a user's function key is compromised, it can be revoked.
To alleviate the burden on mobile devices,  
most of the computation which should be executed by the mobile device is outsourced to a cloud server.
Furthermore, an authorised mobile device can temporarily delegate its keyword search and function computation rights 
to a delegatee in case the device becomes unavailable due to power depletion, going offline, etc.
Our scheme is formally proven secure in the standard model against quantum attacks, chosen plaintext attacks, chosen keyword attacks, and outside keyword guessing attacks. 
Furthermore, the analysis demonstrates that the number of interactions 
between a mobile device and the central authority is $O(1)$ in our scheme, rather than growing linearly with the number of functions, which is well-suited for MCC scenarios.
\end{abstract}


\section{Introduction}
\IEEEPARstart{I}{n} recent years, the widespread proliferation of smartphones, tablets, and wearable devices has positioned mobile terminals 
as one of the primary platforms for information access and task processing. 
Forecasts indicate that global mobile phone shipments will reach approximately 1.4378 billion units by 2028 \cite{wbpf2028}. 
Nevertheless, mobile devices remain inherently constrained by processor performance, storage capacity, and battery life, 
which results in pronounced performance bottlenecks when executing computation-intensive or data-intensive applications \cite{kasmi2018algorithms}. 
The rise of mobile cloud computing (MCC) represents a compelling paradigm for tackling these challenges.
By offloading part or all of the computation and storage workloads from mobile terminals to cloud servers, 
MCC effectively alleviates the limitations imposed by terminal hardware.
However, cloud servers cannot be fully trusted, and MCC applications may involve a substantial amount of sensitive user data, including location information, 
medical records, and behavioural patterns, etc \cite{othman2013survey}. 
Therefore, ensuring the confidentiality of data uploaded to cloud servers is essential for protecting user privacy.

Currently, various secure data sharing schemes \cite{weng2025efficient,zhao2024revocable,Guo2024abd,Xue2021tcs,li2025EPR} built on traditional public key encryption (PKE) have been proposed for MCC environments. 
Although these schemes preserve data confidentiality, their decryption mechanisms are all-or-nothing and do not support selective computation over encrypted data.
This all-or-nothing property is restrictive because many MCC applications, including medical data mining \cite{wu2020fedhome} and encrypted data filtering \cite{huang2024efficient}, require more fine-grained data access control. 

As an emerging public-key paradigm, functional encryption (FE) is extensively employed to enable selective computation over ciphertexts in data analysis applications \cite{cui2020pyd,chang2023privacy,qiu2025privacy}.
Compared to traditional PKE, FE enables authorised parties to compute specific functions on encrypted data without revealing extra information, supporting more fine-grained access control.
Inner product computation serves as a key primitive underpinning numerous MCC application scenarios, such as social networks \cite{marra2017blind,lu2012spoc}, 
health data mining \cite{kadota2008weighted,li2017indexing} and 
recommendation systems \cite{li2017fexipro,aouali2022reward}, etc.
To securely compute inner products, the inner-product functional encryption (IPFE) primitive was proposed.
IPFE restricts authorised users to learning only the inner product between the vector encoded in the function key and that encrypted in the ciphertext, revealing no additional information.
Therefore, IPFE can balance the confidentiality and the availability of data in MCC. 

Leveraging these benefits, some flexible data analysis schemes that support selective computation on encrypted data based on IPFE \cite{han2023privacy,li2025function,wang2025bidirectional,zhu2025revocable,luo2022generic} have been proposed.
However, the existing schemes are unsuitable for the MCC environment due to the following limitations: (1) In existing schemes, each data user must interact with the central authority (CA) to get authorisation. As the number of interactions required by a data user grows linearly with the number of functions, such methods are impractical for resource-constrained mobile devices.
(2) Existing schemes do not support keyword search.
A large amount of user privacy is usually involved in MCC applications, such as personal photos, behavioural data, health data, etc.
To preserve user privacy, sensitive data is typically encrypted before uploading to cloud servers, which prevents users from performing searches over it. 
(3) Existing delegable flexible data analysis schemes rely on traditional hardness assumptions, such as the discrete problem and subgroup decision problem, which are vulnerable to quantum attacks.

To tackle these issues, we propose an efficient quantum-resistant delegable data analysis scheme with revocation and keyword
search (EQDDA-RKS) for MCC. 
In EQDDA-RKS, mobile users can perform keyword searches over ciphertexts and compute inner product values of encrypted data without revealing the underlying data, thereby protecting sensitive information while enhancing data usability.
Since the entropy of keywords is limited in MCC, an outside adversary can launch offline keyword guessing attacks (OKGA) after obtaining the keyword trapdoor. In our scheme, to resist this attack, a mobile user can designate a cloud server to perform keyword searching. 
Meanwhile, our scheme supports fine-grained revocation, namely, revoking a user’s specific function computation rights, instead of all computation rights.
To reduce the overhead on mobile devices, most of the workloads for mobile users are outsourced to a cloud server. Specifically, the cloud server uses a short-term transformation key tied to a function to convert a ciphertext into a transformed ciphertext, which non-revoked users can then decrypt with their corresponding function keys. 
Notably, users only need to interact with the CA once and then can independently generate function keys.
Furthermore, when a mobile device is offline or out of power, it can temporarily delegate its function computation and keyword search rights to another device.
Our scheme is constructed on lattice-based cryptography and provides effective protection for sensitive data in MCC against the escalating threats of quantum computing. 

\subsection{Contributions}
We propose EQDDA-RKS, which provides the following interesting features: 
\begin{itemize}
    \item Keyword search and selective computation: data users can search over encrypted data and compute the inner product values of encrypted data, while ensuring that no additional information is disclosed.
    \item Fine-grained revocation: data users' specific function computation rights can be revoked.
    \item Outsourcing computation: Most of the data users' workloads are outsourced to the server, and a data user only needs to interact with the CA once.
    \item Delegation: data users can temporarily delegate their keyword search and function computation rights to a delegatee.
    
\end{itemize} 

Our contributions are summarised as follows: 
(1) we formalise the definitions and security models of the proposed EQDDA-RKS scheme;
(2) we present the first concrete instantiation of EQDDA-RKS based on lattice-based cryptography;
(3) we prove that our EQDDA-RKS is secure against quantum attacks, chosen plaintext attacks, chosen keyword attacks, and OKGA in the standard model;
(4) we analyse the efficiency of our proposed scheme through comprehensive theoretical analysis and implemented experiments, and compare it with related schemes.

\subsection{Challenges and Techniques}
$Challenges$: The following issues must be addressed in designing our scheme.

\begin{enumerate}
    \item Since the computational resources of mobile devices are limited, a challenging problem is how to reduce the user’s computational overhead and the number of interactions with the CA while achieving authorisation and revocation.
    
    \item Considering that mobile devices may go offline or run out of power, it is necessary to delegate users’ keyword search and function evaluation rights to other users. However, how to achieve temporary delegation remains a challenging problem.

    \item Existing post-quantum keyword search schemes that resist OKGA are secure only in the random oracle model. It is difficult to design a post-quantum keyword search scheme resistant to OKGA in the standard model.

\end{enumerate}

$Techniques$: To address the above issues, we employ the following techniques.

\begin{enumerate}
    \item To address the first issue, outsourcing computation is applied to enable mobile devices to outsource heavy computation to cloud servers. Meanwhile, since function authorisation is controlled through update keys issued by the CA to the server, each user only needs to interact with the central authority once.

    \item To realise temporary delegation, we embed timestamps into the keyword trapdoors, function keys, and ciphertexts. As a result, keyword search and function computation over encrypted data can only be performed when the timestamps match.

    \item To construct a post-quantum OKGA-resistant keyword search in the standard model, we build upon the scheme in \cite{agrawal2010efficient}. Our scheme employs a lattice composed of a left and right lattice: the left-lattice trapdoor acts as the master secret for generating all keyword trapdoors, while the right-lattice trapdoor is used in the security proof to simulate keyword trapdoors. Additionally, the server's public key is embedded in the keyword trapdoors, ensuring that only the designated server can execute keyword searches, thereby providing OKGA resistance.
\end{enumerate}

\subsection{Organization}
The rest of the paper is structured as follows. Section \ref{section:RW} surveys related work. Section \ref{section:Pre} presents the necessary preliminaries. Section \ref{section:SA2D} introduces the system framework, formal definitions, and security models. Section \ref{section:CC} describes the concrete construction of our scheme. Section \ref{section:PA} presents both theoretical and experimental performance analyses. Section \ref{section:SP} gives the security proofs. Section \ref{section:Con} concludes the paper.

\section{Related work} \label{section:RW}
     
\subsection{Inner-Product Functional Encryption (IPFE)}
IPFE allows an authorised user to compute the inner product between the vector in their secret key and that in a ciphertext, without disclosing any additional information about the plaintext.
Abdalla et al. \cite{abdalla2015simple} first proposed a generic IPFE construction with selective security. Later, Agrawal et al. \cite{agrawal2016fully} presented adaptively secure IPFE schemes based on standard assumptions. Considering function privacy, Agrawal et al. \cite{agrawal2015practical} introduced a function-hiding IPFE. 
To tackle concerns regarding master key privacy and encryption vector privacy in IPFE, Yang et al. \cite{yang2020privacy} proposed a privacy-preserving IPFE scheme enabling users to request an unrestricted number of function keys.
Additionally, the scheme supports outsourced decryption, with the user’s storage and computation overhead independent of the vector length. 
However, the schemes in \cite{abdalla2015simple,agrawal2016fully,agrawal2015practical,yang2020privacy} support inner product computations only on data from a single encryptor. To overcome this limitation, Abdalla et al. \cite{abdalla2017multi} introduced multi-input IPFE (MI-IPFE), allowing each encryptor to independently produce encrypted data.

To prevent key abuse by malicious users, Luo et al. \cite{luo2022generic} proposed a trace-and-revoke IPFE scheme, which identifies malicious users and allows the encryptor to revoke their computation rights.
However, the original construction \cite{luo2022generic} only supports bounded collusion, requiring a predetermined limit on the number of secret keys a pirate decoder may acquire. To overcome this, Luo et al. \cite{luo2024fully} proposed a fully collusion-resistant trace-and-revoke IPFE scheme.
In both \cite{luo2022generic,luo2024fully}, user revocation is implemented via a revocation list specified during encryption, referred to as direct revocation. 
To achieve a balance between security and efficiency for health data, Zhu et al. \cite{zhu2025revocable} introduced a revocable hierarchical-identity-based IPFE scheme, 
where patients send private information to a central authority, which generates ciphertexts for authenticated users and uploads them to cloud servers. When user revocation is required, the central authority updates the corresponding ciphertext to dismiss the user’s computation capability.
Considering that data users may need to verify the origin of data, Yao et al. \cite{yao2025revocable}  proposed a server-aided revocable attribute-based matchmaking IPFE, where both encryptors and users can define expressive policies, thereby realising bilateral access control to ensure data source authentication.
Meanwhile, the scheme supports user revocation and is resistant to key exposure attacks.
Although schemes \cite{luo2022generic,luo2024fully,zhu2025revocable,yao2025revocable} enable revocation, they can only revoke a user's entire function computation rights rather than specific function computing rights. To solve this issue, Han et al. \cite{han2025inner} proposed the IPFE scheme supporting fine-grained revocation, where users' specific function computing rights can be revoked. 

\subsection{Searchable Encryption (SE)}
Searchable encryption (SE) \cite{song2000practical} enables keyword searches over encrypted data and is primarily categorised into symmetric SE (SSE) and public key SE (PKSE). In SSE, the same secret key is used to create both ciphertexts and keyword trapdoors. By contrast, PKSE generates ciphertexts with the public key, while trapdoors are derived from the associated secret key \cite{boneh2004public}. 

By combining outsourced attribute-based encryption (ABE) with PKSE, Li et al. \cite{li2017out} proposed ABE with keyword search, supporting fine-grained access control and offloading computationally intensive tasks to the server, thereby reducing the user’s local workload.
Considering that cloud servers may be untrusted or prone to software errors, Chen et al. \cite{10123958} proposed a publicly verifiable PKSE scheme based on blockchain technology. 
Han et al. \cite{2024Blockchain} proposed a privacy-preserving PKSE with strong traceability, where users obtain trapdoors without revealing identities and keywords, and a tracer can disclose them if required.
In PKSE, since anyone can generate keyword ciphertexts, an outside adversary who obtains a keyword trapdoor can perform an exhaustive search on the ciphertexts to launch OKGA \cite{byun2006off}.
To resist OKGA, Rhee et al. \cite{rhee2010trapdoor} introduced trapdoor indistinguishability and demonstrated that it ensures resistance to OKGA. 
Byun et al. \cite{byun2006off} identified insider keyword-guess attacks (IKGA) in PKSE, where an insider adversary, such as a cloud server, can exhaustively search keyword ciphertexts to determine if a trapdoor matches a keyword. To mitigate IKGA, Huang et al. \cite{huang2017efficient} proposed public key authenticated encryption with keyword search (PAEKS), in which keyword ciphertexts are generated using the sender’s secret key, providing authentication.

As quantum computing develops, there is an increasing focus on quantum-resistant searchable encryption schemes, and several lattice-based PKSE schemes have been proposed \cite{behnia2018lattice,zhang2019fs}. 
Considering the issues of key leakage and IKGA attacks, Xu et al. \cite{xu2025lattice} proposed a lattice-based forward-secure PAEKS scheme that also supports multi-user scenarios. 
To address the OKGA problem, Fan et al. \cite{fan2023lattice} introduced a designated-server PAEKS scheme that allows the sender to specify which server performs the keyword test.

In Table \ref{tab:compare}, we compare the features of EQDDA-RKS with related schemes.
Schemes \cite{xu2025lattice,fan2023lattice} support keyword search but do not allow computation over encrypted data. Although scheme \cite{fan2023lattice} can resist OKGA, it is secure only in the ROM, while our construction is proven secure in the standard model. 
In schemes \cite{zhu2025revocable,luo2022generic,luo2024fully,yao2025revocable,han2025inner}, users can perform computations over encrypted data, but they cannot conduct keyword searches. Moreover, schemes \cite{zhu2025revocable,luo2022generic,luo2024fully,yao2025revocable} can only support user revocation, but cannot revoke a user’s specific function computation rights. 
However, our scheme supports fine-grained revocation, which enables the revocation of a user’s specific computing rights, 
instead of all computation rights.
Since the security of schemes \cite{zhu2025revocable,yao2025revocable}  relies on the discrete logarithm problem, they cannot resist quantum attacks. In contrast, our scheme achieves post-quantum security by leveraging lattice-based cryptography.
Furthermore, our scheme supports delegation, namely, both search rights and function computation rights can be delegated. This feature has not been considered in the existing schemes.

\section{Preliminaries} \label{section:Pre}

\subsection{Lattice}
$\mathbf{Lattice.}$ 
A lattice $\Lambda$ generated by $n$ linearly independent vectors $\mathbf{b}_1,\ldots ,\mathbf{b}_n \in \mathbb{R}^n$ is defined as
$\Lambda = \{ \sum_{i=1}^{n} x_i \mathbf{b}_i : x_i \in \mathbb{Z} \}$. 
Its dual lattice is
$\Lambda^{\bot} = \{ \mathbf{y} \in \mathrm{span}(\Lambda) : \forall \mathbf{x} \in \Lambda, \langle \mathbf{x}, \mathbf{y} \rangle \in \mathbb{Z}\}$.
For $\mathbf{A} \in \mathbb{Z}_q^{n \times m}$, define
$\Lambda_q^{\bot}(\mathbf{A}) = \{ \mathbf{u} \in \mathbb{Z}^m : \mathbf{A}\mathbf{u} = \mathbf{0} \ (\mathrm{mod\ } q) \}, \quad
\Lambda_q^{\mathbf{z}}(\mathbf{A}) = \{ \mathbf{u} \in \mathbb{Z}^m : \mathbf{A}\mathbf{u} = \mathbf{z} \ (\mathrm{mod\ } q) \}$, 
where $q$ is prime.

$\mathbf{Matrix\ Norms.}$ Let $\|\mathbf{u}\|$ and $\|\mathbf{u}\|_\infty$ denote the $\ell_2$ and chebyshev norms of a vector $\mathbf{u}$. For a matrix $\mathbf{A}$, let $\widetilde{\mathbf{A}}$ be the Gram-Schmidt orthogonalisation of its columns. Additionally, $\|\mathbf{A}\|$, $\|\mathbf{A}\|_2$, and $s_1(\mathbf{A})$ denote the length of the longest column, the operator norm, and the spectral norm of $\mathbf{A}$, respectively.

\begin{table*}[!ht] 
    \caption{Feature Comparison with Related Schemes}
    \label{tab:compare}
    \centering
    \normalsize
    \resizebox{0.8\textwidth}{!}{
    \begin{tabular}{|c|c|c|c|c|c|c|c|c|}
    \hline
    Scheme&\makecell{Selective\\ computing}&\makecell{Quantum\\ resistance}&Revocation type&\makecell{Keyword\\ search}&\makecell{OKGA\\ resistance} & \makecell{Without\\ ROM}& \makecell{Delegation} \\
    \hline
    \cite{xu2025lattice}&\ding{55}&\ding{51}&\ding{55}&\ding{51}&\ding{55}&\ding{55}&\ding{55}\\
     \hline
     \cite{fan2023lattice}&\ding{55}&\ding{51}&\ding{55}&\ding{51}&\ding{51}&\ding{55}&\ding{55}\\
    \hline
    \cite{zhu2025revocable}&\ding{51}&\ding{55}&User revocation&\ding{55}&\ding{55}&\ding{55}&\ding{55}\\
    
    \hline
    \cite{luo2022generic}&\ding{51}&\ding{51}&User revocation&\ding{55}&\ding{55}&\ding{51}&\ding{55}\\

    \hline
    \cite{luo2024fully}&\ding{51}&\ding{51}&User revocation&\ding{55}&\ding{55}&\ding{51}&\ding{55}\\
    \hline
    \cite{yao2025revocable}&\ding{51}&\ding{55}&User revocation&\ding{55}&\ding{55}&\ding{51}&\ding{55}\\
    \hline
    \cite{han2025inner}&\ding{51}&\ding{51}&Fine-grained revocation&\ding{55}&\ding{55}&\ding{51}&\ding{55}\\
    \hline
    EQDDA-RKS&\ding{51}&\ding{51}&Fine-grained revocation&\ding{51}&\ding{51}&\ding{51}&\ding{51}\\
    \hline
    \end{tabular}}
    \end{table*}

$\mathbf{Gaussian\ Distribution.}$
For any lattice $\Lambda \subseteq \mathbb{Z}^n$, the discrete gaussian distribution over $\Lambda$ with center $\mathbf{c} \in \mathbb{R}^n$ and width $\sigma > 0$ is defined as
\begin{equation*}
\mathcal{D}_{\Lambda, \sigma, \mathbf{c}}(\mathbf{y}) 
= \frac{e^{-\pi \frac{\|\mathbf{y}-\mathbf{c}\|^2}{\sigma^2}}}
       {\sum_{\mathbf{x} \in \Lambda} e^{-\pi \frac{\|\mathbf{x}-\mathbf{c}\|^2}{\sigma^2}}},  
\quad \mathbf{y} \in \Lambda.
\end{equation*}

$\mathbf{Learning\ with\ Errors\ (LWE)}$ \cite{regev2009lattices}. For a positive integer $n$, a prime $q$, and a real $\alpha \in (0,1)$, the $\mathrm{LWE}_{n,q,\alpha}$ problem is defined via access to a challenge oracle $\mathcal{O}$, which is either a sampler $\mathcal{O}_s$ or a  sampler $\mathcal{O}'_s$, with behaviors as follows: 

\begin{itemize}
    \item $\mathcal{O}_s$ outputs $(\mathbf{u}, \mathbf{u}^\top \mathbf{s} + \mathbf{e})$, where $\mathbf{u}, \mathbf{s} \in \mathbb{Z}_q^n$ are sampled uniformly, and $\mathbf{e} \in \mathbb{Z}$ is drawn from the discrete gaussian $\mathcal{D}_{\mathbb{Z}, \alpha q}$.
    \item $\mathcal{O}'_s$ outputs uniform samples from $\mathbb{Z}_q^n \times \mathbb{Z}_q$.
\end{itemize}

The LWE assumption states that no PPT adversary $\mathcal{A}$, even when allowed multiple queries to $\mathcal{O}$, can distinguish $\mathcal{O}_s$ from $\mathcal{O}'_s$ with non-negligible advantage, namely
$Adv^{LWE}_{\mathcal{A}}(\lambda) = 
\left\lvert 
\mathrm{Pr}\left[ \mathcal{A}^{O_s}=1\right]  -\mathrm{Pr}\left[ \mathcal{A}^{O_s'}=1\right]
\right\rvert
\leq \epsilon(\lambda).$

$\mathbf{Sample\ Algorithms}$\cite{gentry2008trapdoors,agrawal2010efficient}. 
 \begin{itemize}
     \item $\mathsf{TrapGen}(n,m,q) \to (\mathbf{A},\mathbf{T}_\mathbf{A})$:  
Given $n$, $q$, and $m \ge O(n \log q)$, the algorithm outputs a matrix $\mathbf{A} \in \mathbb{Z}_q^{n \times m}$ and a trapdoor $\mathbf{T}_\mathbf{A} \in \mathbb{Z}^{m \times m}$ for the lattice $\Lambda_q^\bot(\mathbf{A})$. The matrix $\mathbf{A}$ is statistically close to uniform over $\mathbb{Z}_q^{n \times m}$, and the Gram--Schmidt norm of the trapdoor satisfies $\|\widetilde{\mathbf{T}_\mathbf{A}}\| \le O(\sqrt{n \log q})$.

    \item For $m \ge O(n \log q)$, there exists a full-rank matrix $\mathbf{G} \in \mathbb{Z}_q^{n \times m}$ such that the lattice $\Lambda_q^\bot(\mathbf{G})$ has a publicly known basis $\mathbf{T}_{\mathbf{G}} \in \mathbb{Z}^{m \times m}$ with Gram--Schmidt norm $\|\widetilde{\mathbf{T}_{\mathbf{G}}}\| \le \sqrt{5}$.
    
    \item $\mathsf{SampleBasisLeft}(\mathbf{A},\mathbf{B},\mathbf{T}_\mathbf{A},\rho) \to \mathbf{T}_{[\mathbf{A} | \mathbf{B}]}$:  
Given $\mathbf{A}, \mathbf{B} \in \mathbb{Z}_q^{n \times m}$, a trapdoor $\mathbf{T}_\mathbf{A} \in \mathbb{Z}^{m \times m}$ for $\Lambda_q^\bot(\mathbf{A})$, and a parameter $\rho \ge \|\widetilde{\mathbf{T}_\mathbf{A}}\| \cdot \omega(\sqrt{\log m})$, the algorithm outputs a basis $\mathbf{T}_{[\mathbf{A} \mid \mathbf{B}]}$ that is statistically close to the discrete gaussian distribution $\mathcal{D}_{\Lambda_q^\bot([\mathbf{A} \mid \mathbf{B}]), \rho}$. 
\end{itemize}

\begin{itemize}
   \item $\mathsf{SampleBasisRight}(\mathbf{A},\mathbf{G},\mathbf{S},\mathbf{T}_\mathbf{G},\rho) \to \mathbf{T}_{[\mathbf{A} | \mathbf{AS}+\mathbf{G}]}$:  
Given $\mathbf{A}, \mathbf{G} \in \mathbb{Z}_q^{n \times m}$, a low-norm matrix $\mathbf{S} \in \mathbb{Z}_q^{m \times m}$, a trapdoor $\mathbf{T}_\mathbf{G} \in \mathbb{Z}^{m \times m}$ for $\Lambda_q^\bot(\mathbf{G})$, and a parameter  
$\rho \ge \sqrt{5} \cdot (1 + \|\mathbf{S}\|_2 \cdot \omega(\sqrt{\log m}))$,  
the algorithm outputs a basis $\mathbf{T}_{[\mathbf{A} \mid \mathbf{A}\mathbf{S} + \mathbf{G}]}$ that is statistically close to the discrete gaussian distribution  
$\mathcal{D}_{\Lambda_q^\bot([\mathbf{A} \mid \mathbf{A}\mathbf{S} + \mathbf{G}]), \rho}$. 

    \item $\mathsf{SampleLeft}(\mathbf{A},\mathbf{B},\mathbf{T}_\mathbf{A},\mathbf{U},\rho) \to \mathbf{Z}$:  
Given $\mathbf{A}, \mathbf{B} \in \mathbb{Z}_q^{n \times m}$, a trapdoor $\mathbf{T}_\mathbf{A} \in \mathbb{Z}^{m \times m}$ for $\Lambda_q^\bot(\mathbf{A})$, a target matrix $\mathbf{U} \in \mathbb{Z}_q^{n \times l}$, and a parameter  
$\rho \ge \|\widetilde{\mathbf{T}_\mathbf{A}}\| \cdot \omega(\sqrt{\log m})$,  
the algorithm outputs a matrix $\mathbf{Z} \in \mathbb{Z}^{2m \times l}$ that is statistically close to the discrete gaussian distribution  
$\mathcal{D}_{\Lambda_q^\mathbf{U}([\mathbf{A} \mid \mathbf{B}]), \rho}$. 

    \item $\mathsf{SampleRight}(\mathbf{A},\mathbf{G},\mathbf{S},\mathbf{T}_\mathbf{G},\mathbf{U},\rho) \to \mathbf{Z}$:  
Given $\mathbf{A}, \mathbf{G} \in \mathbb{Z}_q^{n \times m}$, a low-norm matrix $\mathbf{S} \in \mathbb{Z}_q^{m \times m}$, a trapdoor $\mathbf{T}_\mathbf{G} \in \mathbb{Z}^{m \times m}$ for $\Lambda_q^\bot(\mathbf{G})$, a target matrix $\mathbf{U} \in \mathbb{Z}_q^{n \times l}$, and a parameter  
$\rho \ge \sqrt{5}\cdot (1 + \|\mathbf{S}\|_2 \cdot \omega(\sqrt{\log m}))$,  
the algorithm outputs a matrix $\mathbf{Z} \in \mathbb{Z}^{2m \times l}$ that is statistically close to the discrete gaussian distribution  
$\mathcal{D}_{\Lambda_q^\mathbf{U}([\mathbf{A} \mid \mathbf{A}\mathbf{S} + \mathbf{G}]), \rho}$.
\end{itemize}

\begin{lemma}[Bounding Spectral Norm of a Gaussian Matrix \cite{ducas2014improved}] \label{BoundGauss}
Let $\mathbf{Z} \in \mathbb{R}^{n \times m}$ be a sub-Gaussian random matrix with parameter $\rho$. There exists a universal constant $C \approx 1/\sqrt{2\pi}$ such that, for any $t \ge 0$, the spectral norm of $\mathbf{Z}$ satisfies
$s_1(\mathbf{Z}) \le C \cdot \rho \, (\sqrt{n} + \sqrt{m} + t)$
with probability at least $1 - 2 e^{-\pi t^2}$.
\end{lemma}
    
\begin{lemma}[Noise distribution \cite{katsumata2016partitioning}] 
Let $\mathbf{R} \in \mathbb{Z}^{m \times t}$ and let $s \ge s_1(\mathbf{R})$. 
For a vector $\mathbf{e} \gets \mathcal{D}_{\mathbb{Z}^t, \sigma}$, 
the sum 
$\mathbf{R}\mathbf{e} + \mathsf{NoiseGen}(\mathbf{R}, s)$
is statistically close to the discrete Gaussian 
$\mathcal{D}_{\mathbb{Z}^m, 2s\sigma}$.

\end{lemma}

\begin{lemma}[Bounding Norm of a $\{1,-1\}^{k \times m}$ Matrix \cite{agrawal2010efficient}] 
Let $\mathbf{R} \in \{1, -1\}^{k \times m}$ be a matrix chosen uniformly at random. 
There exists a universal constant $C'$ such that
$\mathrm{Pr} \left[\|\mathbf{R}\| \geq C' \sqrt{k+m}\right] < \frac{1}{e^{k+m}}$.
\end{lemma}

\subsection{Full-rank Differences (FRD)  \cite{agrawal2010efficient}}
Let $q$ be a prime and $n$ a positive integer.  
A function $H: \mathbb{Z}_q^n \to \mathbb{Z}_q^{n \times n}$ is said to have the full-rank difference (FRD) property if, for any distinct $\mathbf{x}, \mathbf{y} \in \mathbb{Z}_q^n$, the matrix
$H(\mathbf{x}) - H(\mathbf{y})$
is full-rank.

\subsection{The Complete Subtree Method \cite{naor2001revocation}}
In the complete subtree method, a binary tree $BT$ with $N$ leaves is built, where each leaf $v_{ID_u}$ corresponds to a user identity $ID_u$.
For any leaf $v_{ID_u}$, let $\mathsf{Path}(v_{ID_u})$ denote the nodes on the path from the root to that leaf.
For each internal node $\theta$, denote its left and right children by $\theta_l$ and $\theta_r$.
A revocation list $RL$ records revocation events, where $(v_{ID_u}, t) \in RL$ means that $ID_u$ is revoked at time $t$.
The method employs a node-selection algorithm, $\mathsf{KUNodes}$, which takes as input $BT$, $RL$, and a time $t$, and returns a node set $Y$.
The procedure of $\mathsf{KUNodes}$ is given below.

\begin{align*}  
\mathsf{K}&\mathsf{UNodes}(BT,RL,t): \\
&X,Y \gets \emptyset.\\
&\forall (\theta_i,t_i) \in \mathsf{RL}: \ \mathrm{if}\  t_i \leq t, \  \mathrm{add}\ \mathsf{Path}(\theta_{i}) \ \mathrm{to} \ X.\\
&\forall \theta \in X: \\
&\ \ \ \ \mathrm{if} \ \theta_{r} \notin X, \ \mathrm{then}\ \mathrm{add} \ 
\theta_{r} \ \mathrm{to} \ Y; \\
&\ \ \ \ \mathrm{if} \ \theta_{l} \notin X, 
\ \mathrm{then}\ \mathrm{add} \ \theta_{l} \ \mathrm{to} \ Y.\\
&\mathrm{if} \ Y = \emptyset,\ \mathrm{add}\ \mathrm{the}\ \mathrm{root}\ \mathrm{node}\ \mathrm{to}\ Y.\\
&\mathrm{Return} \ Y.\\
\end{align*} 

\section{System Architecture and definitions}\label{section:SA2D}
\subsection{System Architecture}
As illustrated in Fig. \ref{figure:overview}, the system model of EQDDA-RKS involves the following entities: 
\begin{itemize}

 \item Central authority (CA): As a trusted entity, the CA undertakes system initialisation and user privilege management, while also generating private keys for data users and the cloud server, and distributing tokens and update keys to the cloud server.

 \item Cloud server (CS): As a semi-trusted entity, CS is tasked with the storage and management of ciphertexts. Specifically, upon receiving a keyword search request, the server performs the keyword search to obtain the corresponding ciphertext and transforms it into a partially decrypted ciphertext before returning it to the user.

 \item Data owner (DO): DO represents mobile users in charge of encrypting the data they want to share and uploading it to the server.

 \item Data User (DU): DU represents mobile users who intend to analyse and compute on the encrypted data. They submit keyword search requests to the server, obtain the transformed ciphertexts, and then decrypt them to obtain the inner product values. When a DU runs out of battery or goes offline, he can temporarily delegate his keyword search and function computation rights to other users by generating short-term keyword trapdoors and function keys.

 \item Data Delegatee (DD): DD is a delegated mobile user who can use the short-term keyword trapdoor obtained from the DU to send a search request to the CS and obtain the corresponding ciphertext. Then, DD computes the inner product of the encrypted data using the short-term function key.
\end{itemize}
\begin{figure*}[!ht]
    \centering
    \includegraphics[width=0.8\linewidth]{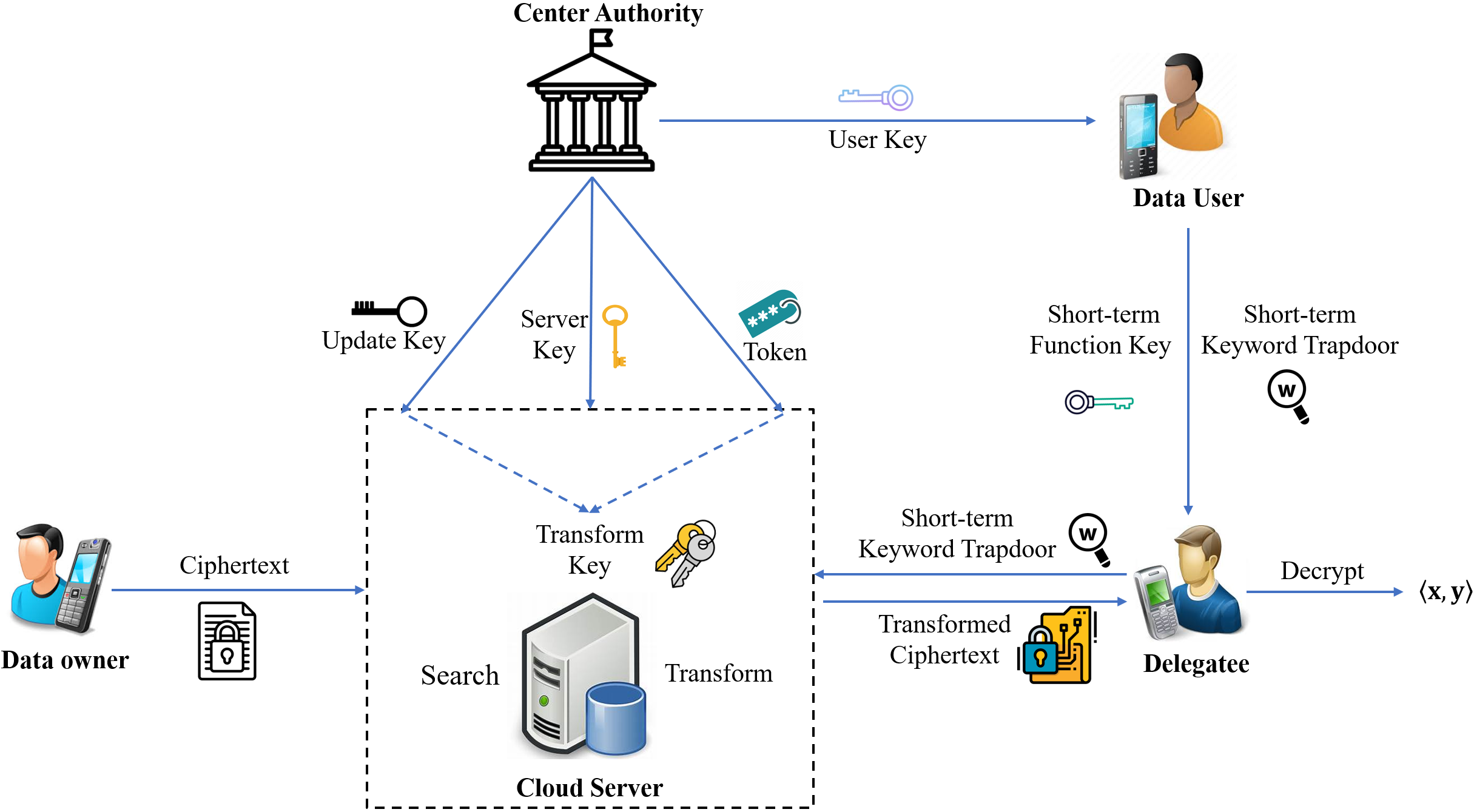}
    \caption{System model}\label{figure:overview}
\end{figure*}
\subsection{Formal Definition of EQDDA-RKS} 
Our scheme consists of the following algorithms:
\begin{itemize}
    \item $Setup(1^\lambda) \to (msk, pp, RL, st)$:  
CA executes this algorithm with the security parameter $\lambda$ as input, outputting the master secret key $msk$, public parameters $pp$, an empty revocation list $RL$, and the initial state $st$. For simplicity, subsequent algorithms omit $pp$ from their inputs.
    
    \item $SerKG(msk,ID_s) \to sk_{ID_s}$: CA executes this algorithm with the master secret key $msk$ and a cloud server identity $ID_s$ as input, outputting the corresponding private key $sk_{ID_s}$.
    
    \item $UserKG(msk,ID_u)\to sk_{ID_u}$: 
    CA executes this algorithm with the master secret key $msk$ and a user identity $ID_u$ as input, outputting the user’s private key $sk_{ID_u}$.

    \item $Token(msk,ID_u,st) \to (tok_{ID_u},st)$: 
     CA executes this algorithm with the master secret key $msk$, the user identity $ID_u$, and the current state $st$ as input, outputting a token $tok_{ID_u}$ along with an updated state $st$.
    
    \item $UpdKG(msk,\mathbf{x},RL_\mathbf{x},t,st) \to uk_{\mathbf{x},t}$: 
    CA executes this algorithm with the master secret key $msk$, a vector $\mathbf{x}$, its associated revocation list $RL_\mathbf{x}$, a time $t$, and the current state $st$ as input, outputting an update key $uk_{\mathbf{x},t}$.
    
    \item $TranKG(tok_{ID_u},uk_{\mathbf{x},t},\mathbf{x})\to tk_{ID_u,\mathbf{x},t}$: 
     CS executes this algorithm with the token $tok_{ID_u}$, the update key $uk_{\mathbf{x},t}$, and a vector $\mathbf{x}$ as input, outputting a transformation key $tk_{ID_u,\mathbf{x},t}$.
    
    \item $FunKG(sk_{ID_u},\mathbf{x},t) \to fk_{ID_u,\mathbf{x},t}$: 
    DU executes this algorithm with the private key $sk_{ID_u}$, the vector $\mathbf{x}$, and the time $t$ as input, outputting a short-term function key $fk_{ID_u,\mathbf{x},t}$.
    
    \item $Enc(ID_s,ID_u,\omega,t,\mathbf{y})\to CT_{ID_u,t}$:
    DO executes this algorithm with the server identity $ID_s$, the user identity $ID_u$, a keyword $\omega$, the time $t$, and a vector $\mathbf{y}$ as input, outputting a ciphertext $CT_{ID_u,t}$.
    
    \item $dTrapdoor(sk_{ID_u},ID_s,\omega,t) \to dt_{ID_u,\omega,t}$:
    DU executes this algorithm with the private key $sk_{ID_u}$, the cloud server identity $ID_s$, the keyword $\omega$, and the time $t$ as input, outputting a short-term keyword trapdoor $dt_{ID_u,\omega,t}$.
    
    \item $Test(dt_{ID_u,\omega,t},sk_{ID_s},CT_{ID_u,t}) \to 0/1$:
    CS executes this algorithm with the short-term keyword trapdoor $dt_{ID_u,\omega,t}$, the server’s private key $sk_{ID_s}$, and the ciphertext $CT_{ID_u,t}$ as input, outputting 1 for a successful match or 0 otherwise.

    \item $Transform(tk_{ID_u,\mathbf{x},t},CT_{ID_u,t},\mathbf{x}) \to TCT_{ID_u,t}^{\mathbf{x}}$: 
    CS executes this algorithm with the transformation key $tk_{ID_u,\mathbf{x},t}$, the ciphertext $CT_{ID_u,t}$, and the vector $\mathbf{x}$ as input, outputting a transformed ciphertext $TCT_{ID_u,t}^{\mathbf{x}}$.
    
    \item $Dec(TCT_{ID_u,t}^{\mathbf{x}},fk_{ID_u,\mathbf{x},t},\mathbf{x}) \to \langle \mathbf{x},\mathbf{y}\rangle$: 
    DO executes this algorithm with the transformed ciphertext $TCT_{ID_u,t}^{\mathbf{x}}$, the short-term function key $fk_{ID_u,\mathbf{x},t}$, and the vector $\mathbf{x}$ as input, outputting the inner product $\langle \mathbf{x}, \mathbf{y} \rangle$.

    \item $Revoke(ID_u,t,RL_\mathbf{x},st) \to RL_\mathbf{x}$: 
    CA executes this algorithm with a user identity $ID_u$, a time $t$, the revocation list $RL_{\mathbf{x}}$ associated with $\mathbf{x}$, and the current state $st$ as input, outputting the updated $RL_{\mathbf{x}}$. 
\end{itemize}

\subsection{Security Model of EQDDA-RKS} \label{SM}
We define the security models of our scheme as follows.

$\mathbf{sIND-CPA}.$ The selective indistinguishability under chosen-plaintext attacks (sIND-CPA) is defined via the following security game between an adversary $\mathcal{A}$ and a challenger $\mathcal{C}$.

$\mathbf{Init.}$ $\mathcal{A}$ selects a target identity $ID_u^*$, a vector $\mathbf{x}^*$, and a timestamp $t^*$, and submits $(ID_u^*, \mathbf{x}^*, t^*)$ to $\mathcal{C}$.

$\mathbf{Setup.}$ $\mathcal{C}$ sends $pp$ to $\mathcal{A}$
by running $Setup(1^\lambda) \to (msk,pp,RL,st)$. 

$\mathbf{Phase-1.}$ 

SerKG Query.
$\mathcal{A}$ submits $ID_s$ to $\mathcal{C}$, and $\mathcal{C}$ returns $sk_{ID_s}$ by running $SerKG(msk, ID_s) \to sk_{ID_s}$.

UserKG Query. $\mathcal{A}$ submits $ID_u$ to $\mathcal{C}$, and $\mathcal{C}$ returns $sk_{ID_u}$ by running $UserKG(msk, ID_u)\to sk_{ID_u}$.

Token Query. $\mathcal{A}$ submits $ID_u$ to $\mathcal{C}$, and $\mathcal{C}$ returns $(tok_{ID_u}, st)$ by running $Token(msk, ID_u, st) \to (tok_{ID_u}, st)$.

UpdKG Query. $\mathcal{A}$ submits $(\mathbf{x}, t) \neq (\mathbf{x}^*, t^*)$ to $\mathcal{C}$, and $\mathcal{C}$ returns $uk_{\mathbf{x}, t}$ by running  $UpdKG(msk, \mathbf{x}, RL_{\mathbf{x}}, t, st) \to uk_{\mathbf{x}, t}$. An initially empty table $Table_1$ records $\mathbf{x}$ whenever $t = t^*$ and $\mathbf{x} \notin Table_1$.
 
FunKG Query. $\mathcal{A}$ submits $(ID_u, \mathbf{x}, t)$ to $\mathcal{C}$, and $\mathcal{C}$ returns $fk_{ID_u,\mathbf{x}, t}$ by running  $FunKG(sk_{ID_u}, \mathbf{x}, t) \to fk_{ID_u,\mathbf{x}, t}$. An initially empty table $Table_2$ records $\mathbf{x}$ whenever $(ID_u, t) = (ID_u^*, t^*)$ and $\mathbf{x} \notin Table_2$.

dTrapdoor Query. $\mathcal{A}$ submits $(ID_u, ID_s, \omega, t)$ to $\mathcal{C}$, and $\mathcal{C}$ returns $dt_{ID_u,\omega,t}$ by running  $dTrapdoor(sk_{ID_u}, ID_s, \omega, t) \to dt_{ID_u,\omega,t}$.

Revoke Query. $\mathcal{A}$ submits $(\mathbf{x}, RL_{\mathbf{x}}, ID_u, t)$ to $\mathcal{C}$, and $\mathcal{C}$ returns $RL_{\mathbf{x}}$ by running $Revoke(ID_u, t, RL_{\mathbf{x}}, st) \to RL_{\mathbf{x}}$.

$\mathbf{Challenge.}$ $\mathcal{A}$ submits $(ID_s^*, \mathbf{y}^*_0, \mathbf{y}^*_1)$ to $\mathcal{C}$, and $\mathcal{C}$ randomly selects $b \in \{0,1\}$, and returns $CT^*$ by running $Enc(ID_s^*, ID_u^*, \omega, t^*, \mathbf{y}^*_b) \to CT^*$.

$\mathbf{Phase-2.}$
$\mathcal{A}$ can continue to make SerKG, UserKG, Token, UpdKG, FunKG, dTrapdoor, and Revoke queries, except the FunKG queries for $\mathbf{x}$ 
satisfying
$\left\langle \mathbf{x},\mathbf{y}^*_0 \right\rangle \neq \left\langle \mathbf{x},\mathbf{y}^*_1 \right\rangle$, 
and $\mathcal{C}$ answers as $\mathbf{Phase-1}$. 

$\textbf{Guess.}$ $\mathcal{A}$ outputs a guess $b' \in \{0,1\}$ for $b$ and succeeds if $b' = b$.

The following conditions must be met in the game above: 
 \begin{itemize}
    \item[$\bullet$] If a UserKG query has been issued for $ID_u^*$, then for all $\mathbf{x} \in Table_1 \setminus \{\mathbf{x}^*\}$, it must hold that $\langle \mathbf{x}, \mathbf{y}^*_0 \rangle = \langle \mathbf{x}, \mathbf{y}^*_1 \rangle$.
    \item[$\bullet$] If no UserKG query is issued for $ID_u^*$, then for all $\mathbf{x} \in Table_2 \setminus \{\mathbf{x}^*\}$, it must hold that $\langle \mathbf{x}, \mathbf{y}^*_0 \rangle = \langle \mathbf{x}, \mathbf{y}^*_1 \rangle$.
    \item[$\bullet$] UpdKG and Revoke queries are restricted to times not earlier than any prior query.
    \item[$\bullet$] If an UpdKG query has been issued at time $t$, no Revoke query can be executed at the same time.
    \item[$\bullet$] If a UserKG query has been issued for $ID_u^*$, the Revoke query must target $(ID_u^*, \mathbf{x}^*, t)$ for some $t \le t^*$. 
    \item[$\bullet$] If $(ID_u^*, \mathbf{x}^*)$ is unrevoked at $t^*$, no FunKG query can be issued on $(ID_u^*, \mathbf{x}^*, t^*)$.
 \end{itemize}

\begin{definition}
The EQDDA-RKS scheme is said to achieve sIND-CPA security if every PPT adversary $\mathcal{A}$ has at most a negligible advantage in the above game, namely,   
\begin{align*}
        Adv^{sIND-CPA}_{\mathcal{A}}(\lambda) = \left|Pr[b=b']-\frac{1}{2} \right| \leq \epsilon(\lambda).
    \end{align*}
\end{definition}

$\mathbf{KC-sIND-CKA}$. The keyword-ciphertext selective indistinguishability under chosen-keyword attacks (KC-sIND-CKA) is defined via two security games between a challenger $\mathcal{C}$ and an adversary $\mathcal{A}_s$ (or $\mathcal{A}_o$), where $\mathcal{A}_s$ models a malicious server and $\mathcal{A}_o$ represents an outside adversary, including authorised users.
 
The KC-sIND-CKA security game between the adversary $\mathcal{A}_s$ 
and the challenger~$\mathcal{C}$ is defined as follows.

$\mathbf{Init.}$ $\mathcal{A}_s$ chooses a challenged identity $ID_{u}^{*}$ and a timestamp $t^{*}$, and submits $(ID_u^*,t^*)$ to $\mathcal{C}$. 

$\mathbf{Setup.}$ $\mathcal{C}$ sends $pp$ to $\mathcal{A}_s$
by running $Setup(1^\lambda) \to (msk,pp,RL,st)$.  

$\mathbf{Phase-1.}$ 

SerKG Query.
$\mathcal{A}_s$ submits $ID_s$ to $\mathcal{C}$, 
and $\mathcal{C}$ returns $sk_{ID_s}$ by running $SerKG(msk,ID_s)\to sk_{ID_s}$. 

UserKG Query. $\mathcal{A}_s$ submits $ID_u \neq ID_u^*$ to $\mathcal{C}$, and $\mathcal{C}$ returns $sk_{ID_u}$ by running 
$UserKG(msk,ID_u)\to sk_{ID_u}$. 

Token Query. $\mathcal{A}_s$ submits $ID_u \neq ID_u^*$ to $\mathcal{C}$, and $\mathcal{C}$ returns 
$tok_{ID_u}$ by running $Token(msk,ID_u,st) \to tok_{ID_u}$. 

UpdKG Query. $\mathcal{A}_s$ submits $(RL_\mathbf{x},\mathbf{x},t)$ to $\mathcal{C}$, and $\mathcal{C}$ returns $uk_{\mathbf{x},t}$ by running  
$UpdKG(msk,\mathbf{x},RL_\mathbf{x},t,st) \to uk_{\mathbf{x},t}$. 

FunKG Query. $\mathcal{A}_s$ submits $(ID_u, \mathbf{x}, t)$ to $\mathcal{C}$, subject to $(ID_u, t) \neq (ID_u^*, t^*)$. $\mathcal{C}$ returns $fk_{ID_u, \mathbf{x}, t}$ by running $FunKG(sk_{ID_u}, \mathbf{x}, t) \to fk_{ID_u, \mathbf{x}, t}$.

dTrapdoor Query. $\mathcal{A}_s$ submits $(ID_u,ID_s,\omega,t)$ to $\mathcal{C}$, and $\mathcal{C}$ returns $dt_{ID_u,\omega,t}$ by running 
$dTrapdoor(sk_{ID_u},ID_s,\omega,t) \to dt_{ID_u,\omega,t}$.   
Let $Table_{\omega}$ be an initially empty table. If $(ID_u,t) = (ID_u^*,t^*)$ and 
$\omega \notin T_{\omega}$, 
$\mathcal{C}$ appends $\omega$ to $Table_{\omega}$. 

Revoke Query. $\mathcal{A}_s$ submits $RL_{\mathbf{x}},ID_u,t$ to $\mathcal{C}$, 
and $\mathcal{C}$ returns $RL_{\mathbf{x}}$ by running $Revoke(ID_u,t,RL_{\mathbf{x}},st) \to RL_{\mathbf{x}}$.

$\mathbf{Challenge.}$ $\mathcal{A}_s$ submits $(ID^*_s,\omega^*_0,\omega^*_1)$, with the restriction that 
$\omega^*_0,\omega^*_1 \notin Table_{\omega}$.  
Then, $\mathcal{C}$ randomly selects $b \in \{0,1\}$, and returns  $CT^*$ by running $Enc(ID^*_s,ID_u^*,\omega^*_b,t^*,\mathbf{y}) \to CT^*$. 

$\mathbf{Phase-2.}$
$\mathcal{A}_s$ can continue to make SerKG, UserKG, Token, UpdKG, FunKG, dTrapdoor and Revoke queries, except the dTrapdoor queries for $(ID_u^*, ID_s,\omega^*_0,t^*), (ID_u^*, ID_s,\omega^*_1,t^*)$ and 
UserKG query for $ID^*_u$, 
and $\mathcal{C}$ answers as $\mathbf{Phase-1}$.  

$\textbf{Guess.}$ $\mathcal{A}_s$ outputs a guess $b' \in \{0,1\}$ for $b$ and succeeds if $b' = b$.

The KC-sIND-CKA security game between the adversary $\mathcal{A}_o$ 
and the challenger $\mathcal{C}$ is defined as follows.
 
$\mathbf{Init.}$ $\mathcal{A}_o$ chooses a challenged identity $ID_{s}^{*}$ and a timestamp $t^{*}$, and submits $(ID_s^*,t^*)$ to $\mathcal{C}$. 

$\mathbf{Setup.}$  $\mathcal{C}$ sends $pp$ to $\mathcal{A}_o$ by executing $Setup(1^\lambda) \to (msk, pp, RL, st)$.
 
$\mathbf{Phase-1.}$ 

SerKG Query.
$\mathcal{A}_o$ submits $ID_s \neq ID_s^*$ to $\mathcal{C}$, and $\mathcal{C}$ returns $sk_{ID_s}$ by running $SerKG(msk, ID_s) \to sk_{ID_s}$.

UserKG Query. $\mathcal{A}_o$ submits $ID_u$ to $\mathcal{C}$, and $\mathcal{C}$ returns $sk_{ID_u}$ byu running $UserKG(msk, ID_u) \to sk_{ID_u}$.

Token Query. $\mathcal{A}_o$ submits $ID_u$ to $\mathcal{C}$, and $\mathcal{C}$ returns $(tok_{ID_u}, st)$ by running $Token(msk, ID_u, st) \to (tok_{ID_u}, st)$.

UpdKG Query. $\mathcal{A}_o$ submits $(RL_{\mathbf{x}}, \mathbf{x}, t)$ to $\mathcal{C}$, and $\mathcal{C}$ returns $uk_{\mathbf{x},t}$ by running $UpdKG(msk, \mathbf{x}, RL_{\mathbf{x}}, t, st) \to uk_{\mathbf{x},t}$.
 
Revoke Query. $\mathcal{A}_o$ submits $(RL_{\mathbf{x}}, ID_u, t)$ to $\mathcal{C}$, and $\mathcal{C}$ returns $RL_{\mathbf{x}}$ by running $Revoke(ID_u, t, RL_{\mathbf{x}}, st) \to RL_{\mathbf{x}}$.

$\mathbf{Challenge.}$ $\mathcal{A}_o$ submits $(ID_u^*, \omega^*_0, \omega^*_1)$ to $\mathcal{C}$, which randomly selects $b \in \{0,1\}$, and returns $CT^*$ by running $Enc(ID_s^*, ID_u^*, \omega^*_b, t^*, \mathbf{y})\to CT^*$.

$\mathbf{Phase-2.}$
$\mathcal{A}_o$ can continue issuing SerKG, UserKG, Token, UpdKG, and Revoke queries, with the restriction that SerKG queries for $ID_s^*$ are forbidden. All other queries are answered by $\mathcal{C}$ as in $\mathbf{Phase-1}$.

$\mathbf{Guess.}$ $\mathbf{Guess.}$ $\mathcal{A}_o$ outputs a guess $b' \in \{0,1\}$ for $b$ and wins the game if $b' = b$.

\begin{definition}
The EQDDA-RKS scheme is said to achieve KC-sIND-CKA security if any PPT adversary $\mathcal{A}_s$ or $\mathcal{A}_o$ has only a negligible advantage in winning the above game, namely, 
\begin{align*}
        Adv^{KC-sIND-CKA}_{\mathcal{A}_s/\mathcal{A}_o}(\lambda) = \left|Pr[b=b']-\frac{1}{2} \right| \leq \epsilon(\lambda).
    \end{align*}
\end{definition}

$\mathbf{KT-sIND-CKA}.$ 
The keyword trapdoor selective indistinguishability against chosen keyword attack (KT-sIND-CKA) security of keyword trapdoors is defined via a game played between an adversary $\mathcal{A}$ and a challenger $\mathcal{C}$.

$\mathbf{Init.}$ $\mathcal{A}$ selects a target identity $ID_s^*$ and a timestamp $t^*$, 
and submits the pair $(ID_s^*, t^*)$ to the challenger $\mathcal{C}$.

$\mathbf{Setup.}$ $\mathcal{C}$ sends $pp$ to $\mathcal{A}$ by running $Setup(1^\lambda) \to (msk, pp, RL, st)$.

$\mathbf{Phase-1.}$ 

SerKG Query.
$\mathcal{A}$ submits $ID_s \neq ID_s^*$ to $\mathcal{C}$, and $\mathcal{C}$ returns $sk_{ID_s}$ by running $SerKG(msk, ID_s) \to sk_{ID_s}$.

UserKG Query. $\mathcal{A}$ submits $ID_u$ to $\mathcal{C}$, and $\mathcal{C}$ returns $sk_{ID_u}$ by running $UserKG(msk, ID_u) \to sk_{ID_u}$.

Token Query. $\mathcal{A}$ submits $ID_u$ to $\mathcal{C}$, and $\mathcal{C}$ returns $(tok_{ID_u}, st)$ by running $Token(msk, ID_u, st) \to (tok_{ID_u}, st)$.

UpdKG Query. $\mathcal{A}$ submits $(RL_{\mathbf{x}}, \mathbf{x}, t)$ to $\mathcal{C}$, and $\mathcal{C}$ returns $uk_{\mathbf{x}, t}$ by running $ UpdKG(msk, \mathbf{x}, RL_{\mathbf{x}}, t, st) \to uk_{\mathbf{x}, t}$.



Revoke Query. $\mathcal{A}$ submits $(RL_{\mathbf{x}}, ID_u, t)$ to $\mathcal{C}$, and $\mathcal{C}$ returns $RL_{\mathbf{x}}$ by running $ Revoke(ID_u, t, RL_{\mathbf{x}}, st) \to RL_{\mathbf{x}}$.

$\mathbf{Challenge.}$ $\mathcal{A}$ submits $(ID_u^*, \omega^*_0, \omega^*_1)$ to $\mathcal{C}$, which randomly selects $b \in \{0,1\}$, and $\mathcal{C}$ returns $dt^*$ by running $dTrapdoor(sk_{ID_u^*}, ID_s^*, \omega^*_b, t^*)\to dt^*$.

$\mathbf{Phase-2.}$
$\mathcal{A}$ can continue to make SerKG, UserKG, Token, UpdKG, and Revoke queries, 
except the SerKG query for $ID_s^*$, 
and $\mathcal{C}$ answers as $\mathbf{Phase-1}$. 

$\mathbf{Guess.}$ 
$\mathcal{A}$ outputs a guess $b' \in \{0,1\}$ for the challenge bit $b$ and wins the game if $b' = b$.

\begin{definition}
The EQDDA-RKS scheme is said to achieve KT-sIND-CKA security if any PPT adversary $\mathcal{A}$ has only a negligible advantage in winning the above game, namely, 
\begin{align*}
        Adv^{KT-sIND-CKA}_{\mathcal{A}}(\lambda) = \left|Pr[b=b']-\frac{1}{2} \right| \leq \epsilon(\lambda).
    \end{align*}
\end{definition}

Rhee et al.~\cite{rhee2010trapdoor} demonstrate that a dPEKS scheme 
that achieves keyword trapdoor indistinguishability is secure against OKGA.

\section{Construction} \label{section:CC}
The detailed construction of our EQDDA-RKS scheme is presented below.

1) $Setup(1^\lambda) \to (msk,pp,RL,st)$: 
CA runs $\mathsf{TrapGen}(n,m,q) \to (\mathbf{A},\mathbf{T}_{\mathbf{A}})$ and 
chooses $N = poly(\lambda)$ as the maximal number of users that the system will support, where $\mathbf{A} \in \mathbb{Z}_q^{n\times m}$. 
Then, CA selects integers $X,Y,l,k>0$ and 
reals $\rho,\sigma,\tau>0$, $h_1 = \left\lceil \mathrm{log}\ q\right\rceil$,
$h_2 = 4m * \left\lceil \mathrm{log}\ \rho\right\rceil$. 
Let $K=lXY$, $\mathcal{X} = \{0,\dots X-1\}^l$, $\mathcal{Y}=\{0,\dots,Y-1\}^l$. 
CA randomly selects $\mathbf{B}_1,\mathbf{B}_2,\mathbf{G}, \mathbf{C}_1,\ldots, \mathbf{C}_k \gets \mathbb{Z}_q^{n\times m}$, 
$\mathbf{U} \gets \mathbb{Z}_q^{n\times l}, \mathbf{V} \gets \mathbb{Z}_q^{n\times h_1}, 
\mathbf{v} \gets \mathbb{Z}_q^{n}$.
Suppose
$H_1: \mathbb{Z}_q^{n} \to \mathbb{Z}_q^{n\times n}$ is a FRD function and  
$H_2: \mathbb{Z}_q \to \{0,1\}^{h_2}$  
be a cryptographic hash function. 
Let 
$ID_u \in (0,uid_1,\ldots,uid_{n-1})$, 
$\widetilde{ID_u} \in (1,uid_1,\ldots,uid_{n-1})$, 
$ID_s \in (2,sid_1,\ldots,sid_{n-1})$,
$t \in (3,t_1,\ldots,t_{n-1})$ where 
$(uid_1,\ldots,uid_{n-1})\in \mathbb{Z}_q^{n-1}$, 
$(sid_1,\ldots,sid_{n-1})\in \mathbb{Z}_q^{n-1}$, 
$(t_1,\ldots,t_{n-1})\in \mathbb{Z}_q^{n-1}$.
Thus, there is a one-to-one
correspondence between $ID_u$ and $\widetilde{ID_u}$.
Then, CA 
initializes  
a binary tree $BT$ with at least 
$N$ leaf nodes and 
sets 
\begin{gather*}
RL = \emptyset,
st = BT,
msk=\mathbf{T}_\mathbf{A},\\
pp = (\mathbf{A},\mathbf{B}_1,\mathbf{B}_2,\mathbf{G},\mathbf{C}_1,\ldots,\mathbf{C}_k,\mathbf{U},\mathbf{V},
\mathbf{v},H_1,H_2).
\end{gather*}
Let 
\begin{align*}
\mathbf{B}_{\widetilde{ID_u}} = \mathbf{B}_1 + H_1(\widetilde{ID_u})\mathbf{G} &,\ 
\mathbf{B}_{ID_u} = \mathbf{B}_1 +H_1(ID_u)\mathbf{G},\\
\mathbf{B}_{ID_s} = \mathbf{B}_1 +H_1(ID_s)\mathbf{G} &,\ 
\mathbf{B}_t = \mathbf{B}_2 + H_1(t) \mathbf{G}, \\
\mathbf{B}_{\omega} = \mathbf{G} + \sum_{i=1}^k b_i \mathbf{C}_i &,\ 
\mathbf{A}_{\widetilde{ID_u}} = \left[ \mathbf{A} | \mathbf{B}_{\widetilde{ID_u}} \right],\\
\mathbf{A}_{ID_u} = \left[ \mathbf{A} | \mathbf{B}_{ID_u} \right] &,\ 
\mathbf{A}_{ID_s} = \left[ \mathbf{A} | \mathbf{B}_{ID_s} \right],\\
\mathbf{A}_{t} = \left[ \mathbf{A} | \mathbf{B}_{t} \right] &,\ 
\mathbf{A}_{\widetilde{ID_u},t} = \left[\mathbf{A}_{\widetilde{ID_u}} | \mathbf{B}_t \right], \\
\mathbf{A}_{ID_u,t} = \left[\mathbf{A}_{ID_u} | \mathbf{B}_t \right] &,\ 
\mathbf{A}_{\widetilde{ID_u},\omega,t} = \left[ \mathbf{A}_{\widetilde{ID_u}}| \mathbf{B}_{\omega} | \mathbf{B}_t\right]. 
\end{align*}

2) $SerKG(msk,ID_s) \to sk_{ID_s}$: 
CA runs 
$\mathsf{SampleLeft}(
\mathbf{A},
\mathbf{B}_{ID_s}, 
\mathbf{T}_{\mathbf{A}}, \mathbf{v}, \rho) 
\to \mathbf{z}_{ID_s} \in \mathbb{Z}^{2m}$, 
$\mathsf{SampleLeft}(
\mathbf{A},
\mathbf{B}_{ID_s}, 
\mathbf{T}_{\mathbf{A}}, \mathbf{V}, \rho) 
\to \mathbf{Z}_{ID_s} \in \mathbb{Z}^{2m \times h_1}$,
and sends 
\begin{align*}
sk_{ID_s} = (\mathbf{z}_{ID_s},\mathbf{Z}_{ID_s})
\end{align*} 
to $\mathcal{CS}$. 

3) $UserKG(msk,ID_u)\to sk_{ID_u}$: 
CA runs 
$\mathsf{SampleBasisLeft}(\mathbf{A},\mathbf{B}_{\widetilde{ID_u}},  
\mathbf{T}_{\mathbf{A}},\rho) \to \mathbf{T}_{\widetilde{ID_u}}$, 
and sends 
\begin{align*}
    sk_{ID_u} = \mathbf{T}_{\widetilde{ID_u}}
\end{align*}
to $\mathcal{DU}$. 

4) $Token(msk,ID_u,st) \to (tok_{ID_u},st)$: 
$\mathcal{CA}$ picks an unassigned leaf note $v_{ID_u}$ from $BT$ and stores $ID_u$ in this node. 
For each $\theta\in \mathsf{path}(v_{ID_u})$, if $\mathbf{U}_{\theta,1},\mathbf{U}_{\theta,2}$ 
are undefined, $\mathcal{CA}$ randomly selects $\mathbf{U}_{\theta,1} \gets \mathbb{Z}_q^{n\times l}$, 
sets $\mathbf{U}_{\theta,2} = \mathbf{U} - \mathbf{U}_{\theta,1}\in \mathbb{Z}_q^{n\times l}$, and 
stores $(\mathbf{U}_{\theta,1},\mathbf{U}_{\theta,2})$ in node $\theta$. 
Then, CA runs $\mathsf{SampleLeft}(\mathbf{A},\mathbf{B}_{ID_u},
\mathbf{T}_{\mathbf{A}},\mathbf{U}_{\theta,1},\rho) \to \mathbf{Z}_{ID_u,\theta} \in \mathbb{Z}^{2m\times l}$, 
outputs updated $st$, 
and sends 
\begin{align*}
tok_{ID_u} = (\theta,\mathbf{Z}_{ID_u,\theta})_{\theta\in \mathsf{path}(v_{ID_u})} 
\end{align*}
to CS. 

5) $UpdKG(msk,\mathbf{x},RL_\mathbf{x},t,st) \to uk_{\mathbf{x},t}$: 
For each $\theta \in \mathsf{KUNodes}(BT,RL_\mathbf{x},t)$, 
CA runs $\mathsf{SampleLeft}(\mathbf{A},\mathbf{B}_t, \mathbf{T}_{\mathbf{A}}, 
\mathbf{U}_{\theta,2},\rho) \to \mathbf{Z}_{t,\theta} \in \mathbb{Z}^{2m\times l}$, and sends 
$uk_{\mathbf{x},t} = (\theta,\mathbf{Z}_{t,\theta}\cdot \mathbf{x})_
{\theta \in \mathsf{KUNodes}(BT,RL_\mathbf{x},t)}$ to $\mathcal{CS}$. 

6) $TranKG(tok_{ID_u},uk_{\mathbf{x},t},\mathbf{x})\to tk_{ID_u,\mathbf{x},t}$:
CS parses $tok_{ID_u} = (\theta,\mathbf{Z}_{ID_u,\theta})_{\theta\in I}$ and
$uk_{\mathbf{x},t} = (\theta,\mathbf{Z}_{t,\theta}\cdot \mathbf{x})_
{\theta \in J}$. 
If $I\cap J = \emptyset$, CS returns $\bot$; 
otherwise, CS chooses $\theta \in I\cap J$, parses
$\mathbf{Z}_{ID_u,\theta} \cdot \mathbf{x} = 
\left[\begin{matrix}
\mathbf{z}^0_{ID_u,\theta,\mathbf{x}} \\
\mathbf{z}^1_{ID_u,\theta,\mathbf{x}}
\end{matrix}\right] \in \mathbb{Z}^{2m}$, computes 
$\mathbf{Z}_{t,\theta} \cdot \mathbf{x} =
\left[\begin{matrix}
\mathbf{z}^0_{t,\theta,\mathbf{x}} \\
\mathbf{z}^1_{t,\theta,\mathbf{x}}
\end{matrix}\right]\in \mathbb{Z}^{2m}$, where 
$\mathbf{z}^0_{ID_u,\theta,\mathbf{x}},\mathbf{z}^0_{t,\theta,\mathbf{x}} \in \mathbb{Z}^m$, 
and outputs 
\begin{align*}
tk_{ID_u,t,\mathbf{x}}=
\left[\begin{matrix}
\mathbf{z}^0_{ID_u,\theta,\mathbf{x}}+\mathbf{z}^0_{t,\theta,\mathbf{x}} \\
\mathbf{z}^1_{ID_u,\theta,\mathbf{x}}\\
\mathbf{z}^1_{t,\theta,\mathbf{x}}
\end{matrix}\right] \in\mathbb{Z}^{3m}. 
\end{align*}

7) $FunKG(sk_{ID_u},\mathbf{x},t) \to fk_{ID_u,\mathbf{x},t}$:
DU runs $\mathsf{SampleLeft}(
\mathbf{A}_{\widetilde{ID_u}},
\mathbf{B}_t, 
\mathbf{T}_{\widetilde{ID_u}}, \mathbf{U}, \rho) \to \mathbf{Z}_{ID_u,t} \in \mathbb{Z}^{3m \times l}$, and sends $fk_{ID_u,\mathbf{x},t} = \mathbf{Z}_{ID_u,t} \cdot \mathbf{x}\in \mathbb{Z}^{3m}$ to DD.

8) $Enc(ID_s,ID_u,\omega,t,\mathbf{y})\to CT_{ID_u,t}$: 
DO randomly selects $\mathbf{s}_0,\mathbf{s}_1,\mathbf{s}_2,\mathbf{s}_3 \gets \mathbb{Z}_q^n$, 
$\mathbf{R}_1,\ldots,\mathbf{R}_7,\mathbf{F}_1,\ldots \mathbf{F}_k \gets \{1,-1\}^{m\times m}$, 
$\mathbf{e}_0 ,\mathbf{e}_1 \gets \mathcal{D}_{\mathbb{Z}^m,\sigma}$, 
$\mathbf{e}_2 \gets \mathcal{D}_{\mathbb{Z}^l,\sigma}$,
$\mathbf{e}_3 \gets \mathcal{D}_{\mathbb{Z}^l,\tau}$
$\mathbf{e}_4,\mathbf{e}_5  \gets \mathcal{D}_{\mathbb{Z}^m,\sigma}$, 
$\mathbf{e}_6 \gets \mathcal{D}_{\mathbb{Z},\sigma}$,
computes \begin{align*}
\mathbf{F}_{\omega} &= \sum_{i = 1}^{k}b_i \mathbf{F}_i, 
\mathbf{c}_0 = \mathbf{A}_{ID_u,t}^T 
\cdot \mathbf{s}_0 + 
\left[\mathbf{I}_m|\mathbf{R}_1|\mathbf{R}_2\right]^T \cdot  \mathbf{e}_0, \\
\mathbf{c}_1 &= \mathbf{A}_{\widetilde{ID_u},t}^T 
\cdot \mathbf{s}_1 + \left[\mathbf{I}_m|\mathbf{R}_3|\mathbf{R}_4\right]^T \cdot  \mathbf{e}_1, \\
\mathbf{c}_2 &= \mathbf{U}^T \cdot (\mathbf{s}_0+\mathbf{s}_1) + 
\mathbf{e}_2 + \mathbf{e}_3 +\left\lfloor \frac{q}{K} \right\rfloor \cdot \mathbf{y}, \\
\mathbf{c}_3 &= \mathbf{A}_{\widetilde{ID_u},\omega,t}^T 
\cdot \mathbf{s}_2 + (\mathbf{I}_m | \mathbf{R}_5 | \mathbf{F}_{\omega} | \mathbf{R}_6)^T \cdot \mathbf{e}_4, \\
\mathbf{c}_4 &= \mathbf{A}_{ID_s}^T \cdot \mathbf{s}_3 + 
(\mathbf{I}_m | \mathbf{R}_7)^T \cdot \mathbf{e}_5,\\
\mathbf{c}_5 &= \mathbf{v}^T \cdot (\mathbf{s}_2+\mathbf{s}_3) + \mathbf{e}_6.
\end{align*} 
DO uploads $CT_{ID_u,t} = (\mathbf{c}_0,\mathbf{c}_1,\mathbf{c}_2,\mathbf{c}_3,\mathbf{c}_4,\mathbf{c}_5)$ 
to CS.

9) $dTrapdoor(sk_{ID_u},ID_s,\omega,t) \to dt_{ID_u,\omega,t}$:
DU runs $\mathsf{SampleLeft}(\mathbf{A}_{\widetilde{ID_u}}, \left[ \mathbf{B}_{\omega} | \mathbf{B}_t \right],
\mathbf{T}_{\widetilde{ID_u}},\mathbf{v},\rho) \to  kt_{ID_u,\omega,t} \in \mathbb{Z}^{4m}$, 
randomly selects $\mathbf{s}_4 \gets \mathbb{Z}_q^n$, 
$\mathbf{R}_8 \gets \{1,-1\}^{m\times m}$,
$\mathbf{e}_7 \gets \mathcal{D}_{\mathbb{Z}^m,\sigma}$, 
$\mathbf{e}_8 \gets \mathcal{D}_{\mathbb{Z}^{h_2},\sigma}$, 
$kx \gets \mathbb{Z}_q$, 
computes 
\begin{align*}
\mathbf{kt}_1 &= \mathbf{A}_{ID_s}^T \cdot \mathbf{s}_4 + (\mathbf{I}_m | \mathbf{R}_8)^T \cdot \mathbf{e}_7, \\
\mathbf{kt}_2 &= \mathbf{V}^T \cdot \mathbf{s}_4 + \mathbf{e}_8 + \mathsf{bin}(kx) \cdot \left\lfloor \frac{q}{2} \right\rfloor, \\
\mathbf{kt}_3 &= H_2(kx) \oplus bin(kt_{ID_u,\omega,t}). 
\end{align*} 
DU sends 
$dt_{ID_u,\omega,t} = (\mathbf{kt}_1,\mathbf{kt}_2,\mathbf{kt}_3)$to DD. 

10) $Test(dt_{ID_u,\omega,t},sk_{ID_s},CT_{ID_u,t}) \to 0/1$: 
CS computes $\mathbf{w} = \mathbf{kt}_2 - \mathbf{Z}^T_{ID_s} \cdot \mathbf{kt}_1$,  
$kx = \mathbf{g} \cdot \left\lfloor \frac{2}{q} \mathbf{w}  \right\rceil$, 
$bin(kt_{ID_u,\omega,t}) = H_2(kx) \oplus \mathbf{kt}_3$, 
$\mu' = \mathbf{c}_5 - \mathbf{z}_{ID_s}^T \cdot \mathbf{c}_4 - kt^T_{ID_u,\omega,t} \cdot \mathbf{c}_3$, where $\mathbf{g} = (1,2,\ldots,2^{h_2-1}) \in \mathbb{Z}^{h_2}$. 
If $\left\lvert \mu' - \left\lfloor \frac{q}{2} \right\rfloor  \right\rvert <\left\lfloor \frac{q}{4} \right\rfloor$, 
CS returns 0; otherwise, it returns 1.

11) $Transform(tk_{ID_u,\mathbf{x},t},CT_{ID_u,t},\mathbf{x}) \to TCT_{ID_u,t}^{\mathbf{x}}$:  
CS computes $CT^{\mathbf{x}}_{ID_u,t} = \mathbf{x}^T \cdot \mathbf{c}_2 - tk_{ID_u,\mathbf{x},t}^T\cdot \mathbf{c}_0$ mod $q$, 
and sends 
\begin{align*}
TCT_{ID_u,t}^{\mathbf{x}} = (\mathbf{c}_1,CT^{\mathbf{x}}_{ID_u,t})
\end{align*}
to DD. 

12) $Dec(TCT_{ID_u,t}^{\mathbf{x}},fk_{ID_u,\mathbf{x},t},\mathbf{x}) \to \langle \mathbf{x},\mathbf{y}\rangle$:
DD computes $\varphi'  = CT^{\mathbf{x}}_{ID_u,t} - fk_{ID_u,\mathbf{x},t}^T \cdot \mathbf{c}_1$ mod $q$, 
and outputs a value $\varphi \in \{0,\ldots,K-1\}$ that minimizes 
$\left| \left\lfloor \tfrac{q}{K} \right\rfloor \cdot \varphi - \varphi' \right|$. 
Note that $\varphi = \langle \mathbf{x},\mathbf{y}\rangle$. 

13) $Revoke(ID_u,t,RL_\mathbf{x},st) \to RL_\mathbf{x}$: 
The CA inserts the pair $(ID_u, t)$ into each revocation list $RL_{\mathbf{x}}$ corresponding to the nodes associated with identity $ID_u$, and then outputs the updated revocation list $RL_{\mathbf{x}}$.

$\mathbf{Correctness.}$ 
The correctness of $Test$ can be shown as follows. 

In $Test$, we have
$\mathbf{w} = \mathbf{V}^T\cdot \mathbf{s}_4+\mathbf{e}_8+\mathsf{bin}(kx)
\cdot \left\lfloor \frac{q}{2} \right\rfloor 
-\mathbf{Z}^T_{ID_s} \cdot (\mathbf{A}_{ID_s}^T \cdot \mathbf{s}_4 + (\mathbf{I}_m | \mathbf{R}_8)^T \cdot \mathbf{e}_7)
= \mathsf{bin}(kx)
\cdot \left\lfloor \frac{q}{2} \right\rfloor + \mathbf{e}_8 - 
\mathbf{Z}^T_{ID_s} \cdot (\mathbf{I}_m | \mathbf{R}_8)^T \cdot \mathbf{e}_7$. 
If $\left\lvert \mathbf{e}_8 - \mathbf{Z}^T_{ID_s} \cdot (\mathbf{I}_m | \mathbf{R}_8)^T \cdot \mathbf{e}_7 \right\rvert_{\infty} \leq \sigma(1+2\rho m^2)
\leq \frac{q}{4}$, we have $\mathbf{g}\cdot \left\lfloor \frac{2}{q}\cdot \mathbf{w} \right\rceil
= kx$. 

Furthermore, 
$\mu' = \mathbf{v}^T \cdot (\mathbf{s}_2+\mathbf{s}_3) + \mathbf{e}_6 
-\mathbf{z}^T_{ID_s} \cdot (\mathbf{A}_{ID_s}^T \cdot \mathbf{s}_3 + 
(\mathbf{I}_m | \mathbf{R}_7)^T \cdot \mathbf{e}_5)
- kt^T_{ID_u,\omega,t} \cdot (\mathbf{A}_{\widetilde{ID_u},\omega,t}^T 
\cdot \mathbf{s}_2 + (\mathbf{I}_m | \mathbf{R}_5 | \mathbf{F}_{\omega} | \mathbf{R}_6)^T \cdot \mathbf{e}_4)
= \mathbf{e}_6 - \mathbf{z}^T_{ID_s} \cdot (\mathbf{I}_m | \mathbf{R}_7)^T \cdot \mathbf{e}_5
- kt^T_{ID_u,\omega,t} \cdot (\mathbf{I}_m | \mathbf{R}_5 | \mathbf{F}_{\omega} | \mathbf{R}_6)^T \cdot \mathbf{e}_4$.
To ensure the correctness of $Test$, we require that 
$\left\lvert \mathbf{e}_6 - \mathbf{z}^T_{ID_s} \cdot (\mathbf{I}_m | \mathbf{R}_7)^T \cdot \mathbf{e}_5
- kt^T_{ID_u,\omega,t} \cdot (\mathbf{I}_m | \mathbf{R}_5 | \mathbf{F}_{\omega} | \mathbf{R}_6)^T 
\cdot \right. \\ \left. \mathbf{e}_4 
\right\rvert \leq
\sigma + \rho \sigma \sqrt{m}(2+C'(k+3)\sqrt{2m})
\leq \frac{q}{4}$. 

The correctness of $Dec$ can be shown as follows. 

In $Dec$, we have 
$\varphi' 
= \mathbf{x}^T \cdot \mathbf{c}_2 - tk_{ID_u,\mathbf{x},t}^T\cdot \mathbf{c}_0 
- fk_{ID_u,\mathbf{x},t}^T \cdot \mathbf{c}_1
= \mathbf{x}^T \cdot (\mathbf{U}^T \cdot (\mathbf{s}_0+\mathbf{s}_1) + \mathbf{e}_2 
+ \mathbf{e}_3 + \left\lfloor \frac{q}{K} \right\rfloor \cdot \mathbf{y})
- \left[\begin{matrix}
\mathbf{z}^0_{ID_u,\theta,\mathbf{x}}+\mathbf{z}^0_{t,\theta,\mathbf{x}} \\
\mathbf{z}^1_{ID_u,\theta,\mathbf{x}}\\
\mathbf{z}^1_{t,\theta,\mathbf{x}}
\end{matrix}\right]^T \cdot 
(\mathbf{A}_{ID_u,t}^T 
\cdot \mathbf{s}_0 + \left[\mathbf{I}_m|\mathbf{R}_1|\mathbf{R}_2\right]^T \cdot  \mathbf{e}_0)
- (\mathbf{Z}_{ID_u,t} \cdot \mathbf{x})^T \cdot (\mathbf{A}_{\widetilde{ID_u},t}^T 
\cdot \mathbf{s}_1 + \left[\mathbf{I}_m|\mathbf{R}_3|\mathbf{R}_4\right]^T \cdot  \mathbf{e}_1) 
= \mathbf{x}^T \mathbf{U}^T \cdot (\mathbf{s}_0 + \mathbf{s}_1) 
+ \mathbf{x}^T \cdot (\mathbf{e}_2 + \mathbf{e}_3) + \left\lfloor \frac{q}{K} \right\rfloor \cdot \left\langle \mathbf{x},\mathbf{y} \right\rangle 
- \mathbf{x}^T (\mathbf{U}_{\theta,1}^T + \mathbf{U}^T_{\theta,2}) \cdot \mathbf{s}_0 
- \left[\begin{matrix}
\mathbf{z}^0_{ID_u,\theta,\mathbf{x}}+\mathbf{z}^0_{t,\theta,\mathbf{x}} \\
\mathbf{z}^1_{ID_u,\theta,\mathbf{x}}\\
\mathbf{z}^1_{t,\theta,\mathbf{x}}
\end{matrix}\right]^T \cdot \left[\mathbf{I}_m|\mathbf{R}_1|\mathbf{R}_2\right]^T 
\cdot \mathbf{e}_0 
- \mathbf{x}^T \mathbf{U}^T \cdot \mathbf{s}_1 
-\mathbf{x}^T \mathbf{Z}_{ID_u,t}^T \cdot \left[\mathbf{I}_m|\mathbf{R}_3|\mathbf{R}_4\right]^T \cdot \mathbf{e}_1
=\left\lfloor \frac{q}{K} \right\rfloor \cdot \left\langle \mathbf{x},\mathbf{y} \right\rangle 
+\mathbf{x}^T \cdot (\mathbf{e}_2+ \mathbf{e}_3)
-\left[\begin{matrix}
\mathbf{z}^0_{ID_u,\theta,\mathbf{x}}+\mathbf{z}^0_{t,\theta,\mathbf{x}} \\
\mathbf{z}^1_{ID_u,\theta,\mathbf{x}}\\
\mathbf{z}^1_{t,\theta,\mathbf{x}}
\end{matrix}\right]^T \cdot \left[\mathbf{I}_m|\mathbf{R}_1|\mathbf{R}_2\right]^T \cdot
\mathbf{e}_0
- \mathbf{x}^T \mathbf{Z}_{ID_u,t}^T \cdot \left[\mathbf{I}_m|\mathbf{R}_3|\mathbf{R}_4\right]^T \cdot \mathbf{e}_1$. 

To ensure the correctness of $Dec$, 
$\vert \mathbf{x}^T \cdot (\mathbf{e}_2+ \mathbf{e}_3)
-\left[\begin{matrix}
\mathbf{z}^0_{ID_u,\theta,\mathbf{x}}+\mathbf{z}^0_{t,\theta,\mathbf{x}} \\
\mathbf{z}^1_{ID_u,\theta,\mathbf{x}}\\
\mathbf{z}^1_{t,\theta,\mathbf{x}}
\end{matrix}\right]^T \cdot
\left[\mathbf{I}_m|\mathbf{R}_1|\mathbf{R}_2\right]^T \cdot
\mathbf{e}_0 -
\mathbf{x}^T \mathbf{Z}_{ID_u,t}^T \cdot \left[\mathbf{I}_m|\mathbf{R}_3|\mathbf{R}_4\right]^T 
\cdot \mathbf{e}_1 \vert 
\leq lX(\sigma + \tau) + \rho \sigma \sqrt{m}(1+C'\sqrt{3m})(\sqrt{6m}+lX)
\leq  \frac{q}{4K}$
should be satisfied. 

$\mathbf{Parameters\ Setting.}$ 
To ensure the correctness and security of our scheme, the parameters are set as follows: 
$\sigma(1+2\rho m^2) \leq \frac{q}{4}$, 
$\sigma + \rho \sigma \sqrt{m}(2+C'(k+3)\sqrt{2m}) \leq \frac{q}{4}$,
$lX(\sigma + \tau) + \rho \sigma \sqrt{m}(1+C'\sqrt{3m})(\sqrt{6m}+lX) \leq  \frac{q}{4K}$,
$\alpha q > 2\sqrt{n}$,
$m = O(n\mathrm{log\ }q)$,
$\rho > \omega(\sqrt{\mathrm{log\ }n})$,
$\sigma = 2C\alpha q (\sqrt{n}+\sqrt{m}+\sqrt{l}) $.

\section{PERFORMANCE Analysis} \label{section:PA}
In this section, we conduct both theoretical analysis and experimental evaluation of the scheme and compare it with related schemes. 
\subsection{Theoretical Analysis}
We analyse our scheme from four aspects: the number of interactions among the entities,  computation overhead, storage overhead, and communication overhead of our scheme. 
Tables \ref{tab:interaction_number}, \ref{tab:computation_cost}, 
\ref{tab:storge_cost}, and \ref{tab:communication_cost} show the results.  
In the above tables, $N$, $F$, $T$, $C$, $W$, $V$ and $r$ are the number of users, functions, timestamps, ciphertexts, keywords, system version and revoked users, respectively. 
$t_G$, $t_{SL}$, $t_{BL}$, and $t_m$ are the running time of 
$\mathsf{TrapGen}$, $\mathsf{SampleLeft}$, $\mathsf{SampleBasisLeft}$, and matrix multiplication, respectively.

\begin{table}[!ht] 
    \caption{The number of interactions} \label{tab:interaction_number}
    \centering
    \normalsize
    \resizebox{\linewidth}{!}{
    \begin{tabular}{|c|c|c|c|c|c|}
    \hline
    Scheme&$(CA,DU)$&$(CA,CS)$&$(DU,DD)$&$(DU,CS)$&$(DD,CS)$\\
    \hline
    \cite{han2025inner}&$O(F)$ & $O(V)$ &$-$&$O(C)$&$-$\\
    \hline
    EQDDA-RKS&$O(1)$ & $O(N+F\cdot T)$ & $O(T)$ &-&$O(W)$\\
    \hline
    \end{tabular}}
\end{table}

\begin{table*}[!ht] 
    \caption{Computation Overhead} \label{tab:computation_cost}
    \centering
    \normalsize
    \resizebox{0.7\textwidth}{!}{
    \begin{tabular}{|c|c|c|c|c|c|c|}
    \hline
    Scheme&$Setup$&$SerKG$&$UserKG$&$Token$&$UpdKG$&$TranKG$\\
    \hline
    EQDDA-RKS&$t_G$ & $2t_{SL}$ & $t_{BL}$ & $log\ N \cdot t_{SL}$ & $rlog\ \frac{N}{r} \cdot t_{SL}$& 
    $-$\\
    \hline
    &$FunKG$&$Enc$&$dTrapdoor$&$Test$&$Dec$&$Revoke$\\
    \hline
    EQDDA-RKS&$t_{SL}+t_M$ & $10t_M$ & $t_{SL} +t_M$ & $4t_M$ & $t_M$ & $-$\\
    \hline
    \end{tabular}}
\end{table*}

\begin{table*}[!ht] 
    \caption{Storage Overhead} \label{tab:storge_cost}
    \centering
    \normalsize
    \resizebox{0.7\textwidth}{!}{
    \begin{tabular}{|c|c|c|c|c|c|c|}
    \hline
    \multicolumn{2}{|c|}{Scheme}&public parameters&server secret key& user secret key & token&update key\\  
    \hline
    \multirow{2}{*}{EQDDA-RKS} & Server& \multirow{2}{*}{$O(\lambda^2)$}& $O(\lambda^2)$&$-$& $log\ N \cdot O(\lambda^2)$ & $rlog\ \frac{N}{r}\cdot O(\lambda^2)$\\
    \cline{2-2} \cline{4-7}
    &User & &$-$&$O(\lambda^2)$&$-$&$-$\\
    \hline
    \end{tabular}}
\end{table*}

\begin{table*}[!ht] 
    \caption{Communication Overhead} \label{tab:communication_cost}
    \centering
    \normalsize
    \resizebox{0.7\textwidth}{!}{
    \begin{tabular}{|c|c|c|c|c|c|c|}
    \hline
    Scheme&$Setup$&$SerKG$&$UserKG$&$Token$&$UpdKG$&$TranKG$\\
    \hline
    EQDDA-RKS&$O(\lambda^2)$ & $O(\lambda^2)$ & $O(\lambda^2)$ & $log\ N \cdot O(\lambda^2)$ & $rlog\ \frac{N}{r}\cdot O(\lambda^2)$& $-$\\
    \hline
    &$FunKG$&$Enc$&$dTrapdoor$&$Test$&$Dec$&$Revoke$\\
    \hline
    EQDDA-RKS&$-$ & $O(\lambda)$ & $O(\lambda)$ & $-$ & $-$ & $-$\\
    \hline
    \end{tabular}}
\end{table*}

\subsection{Implementation and Evaluation}
\begin{figure*}[!ht]
    \centering
        \subfloat[$Setup$]{\label{figure:test_Setup}\includegraphics[width = 0.25\textwidth]{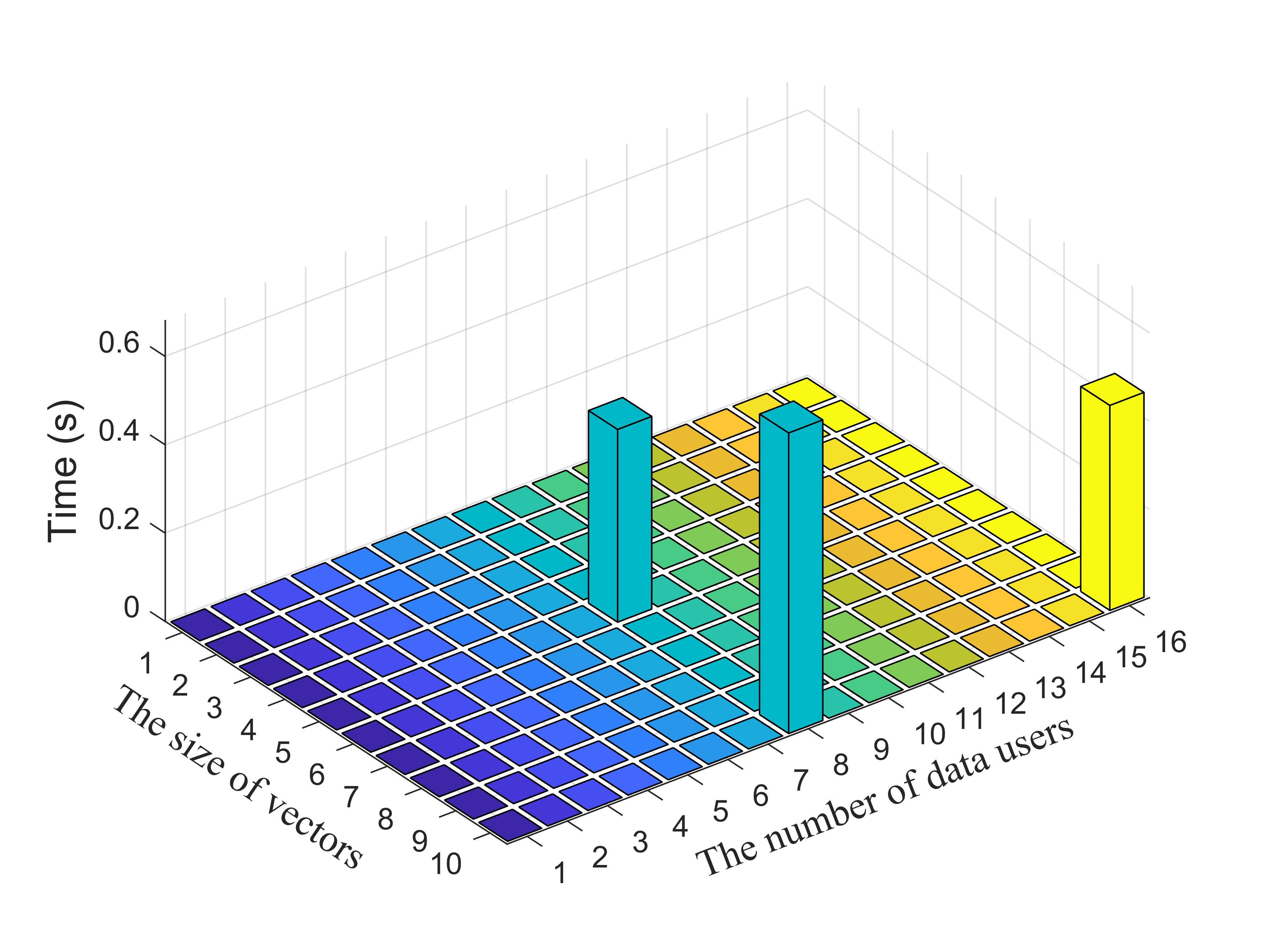}}
        \subfloat[$SerKG$]{\label{figure:test_SerKG}\includegraphics[width = 0.25\textwidth]{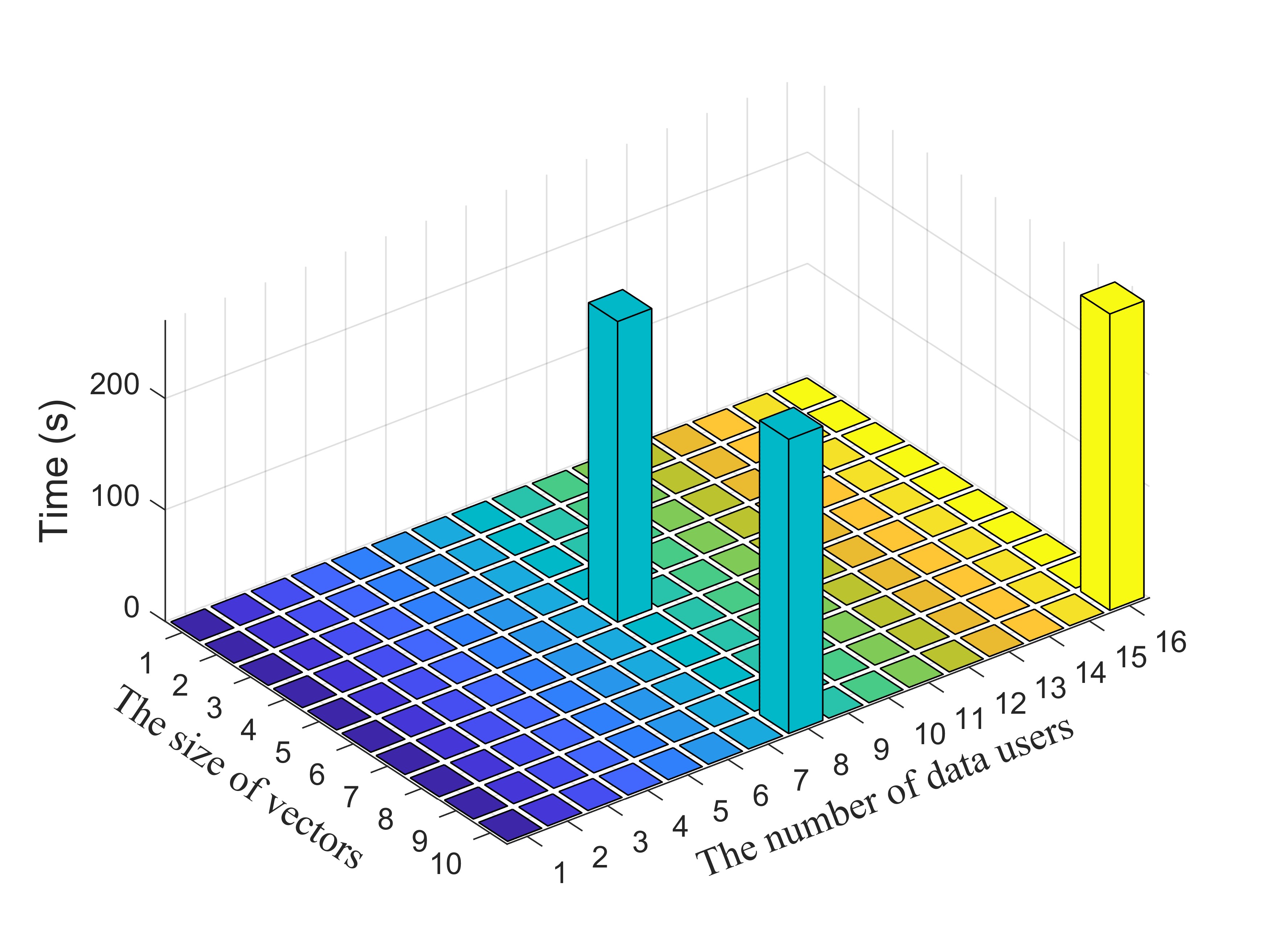}}
        \subfloat[$UserKG$]{\label{figure:test_UserKG}\includegraphics[width = 0.25\textwidth]{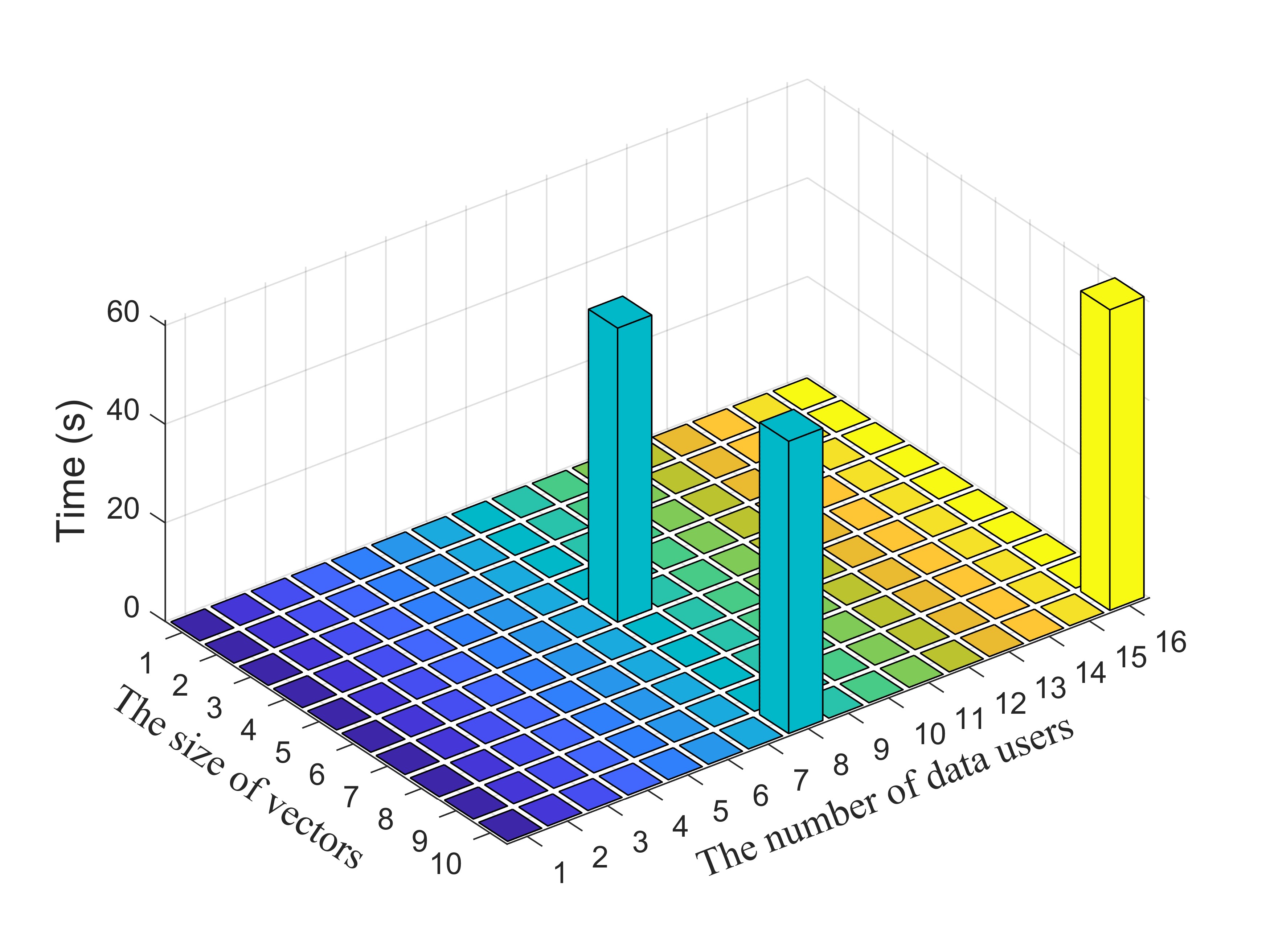}}
        \subfloat[$Token$]{\label{figure:test_Token}\includegraphics[width = 0.25\textwidth]{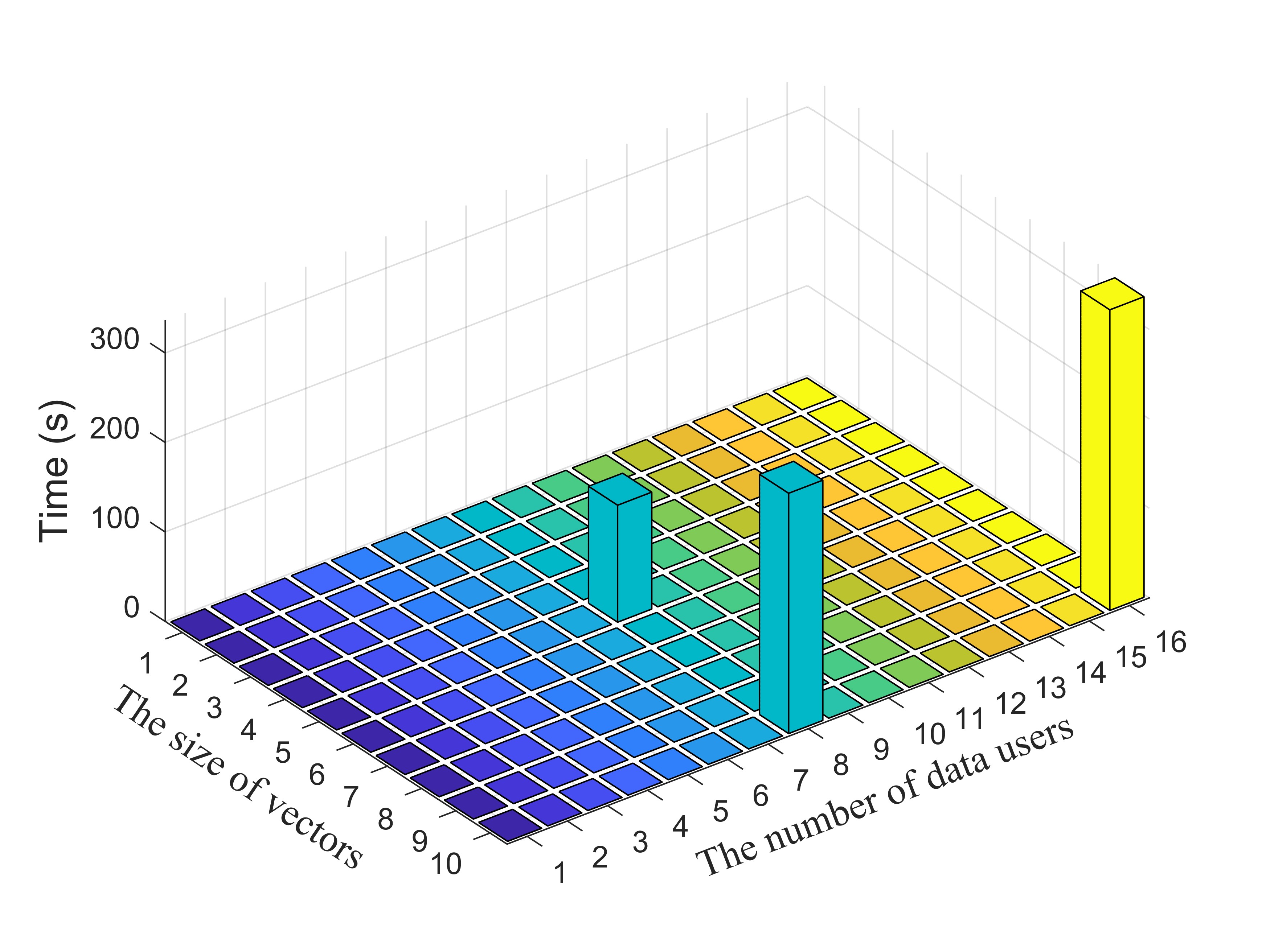}}
        
        \subfloat[$UpdKG$]{\label{figure:test_UpdKG}\includegraphics[width = 0.25\textwidth]{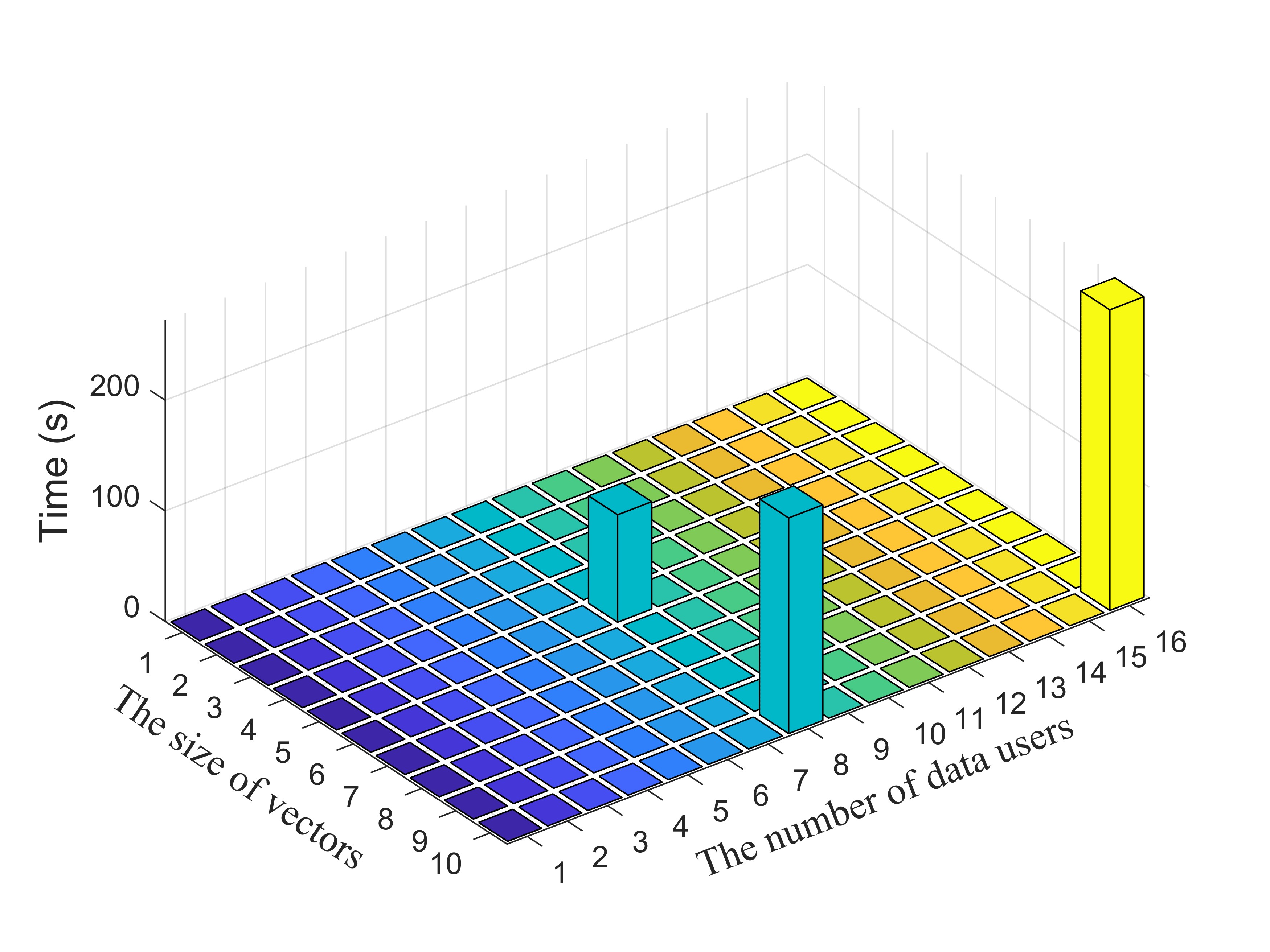}}
        \subfloat[$TranKG$]{\label{figure:test_TranKG}\includegraphics[width = 0.25\textwidth]{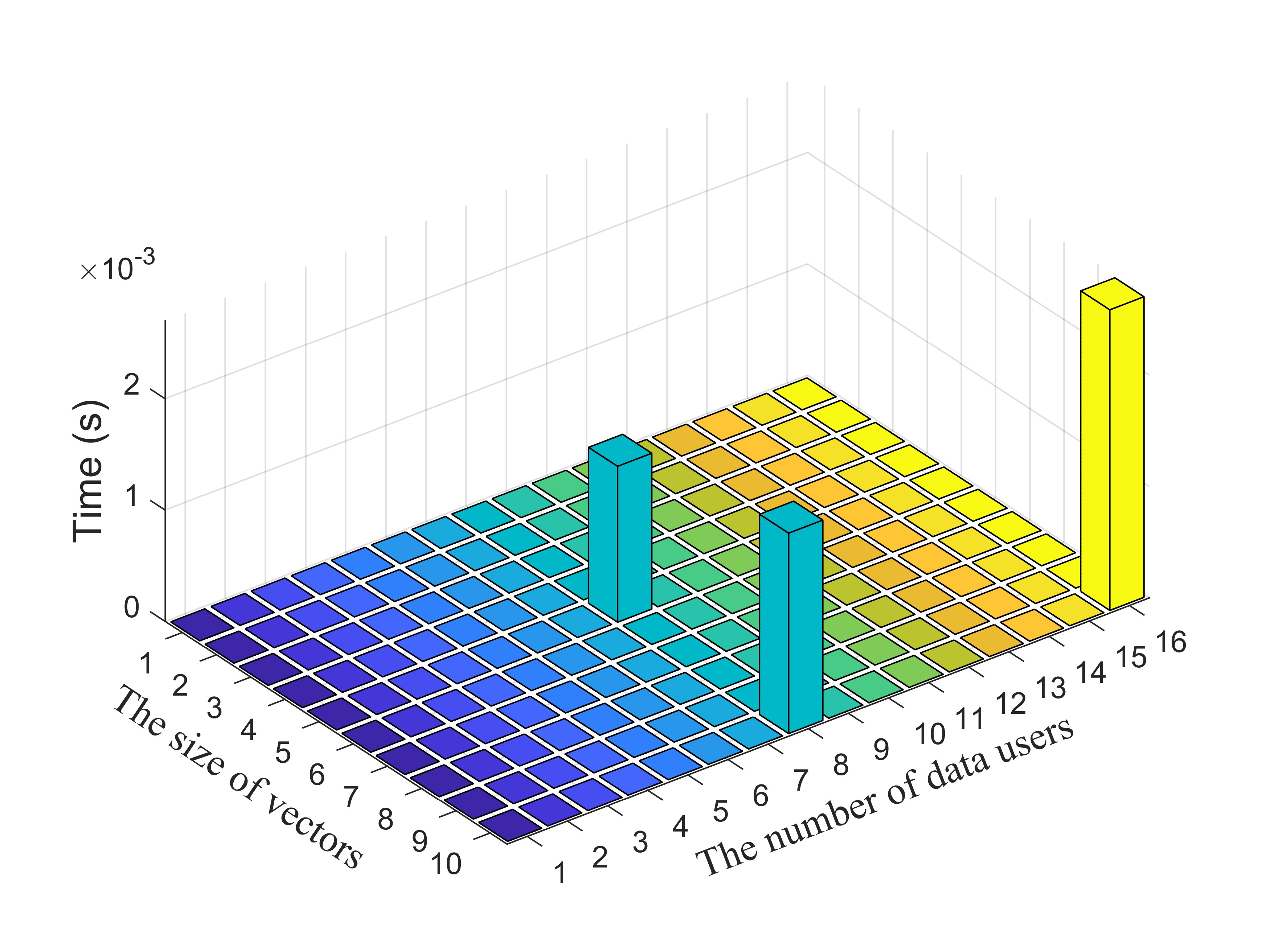}}
        \subfloat[$FunKG$]{\label{figure:test_FunKG}\includegraphics[width = 0.25\textwidth]{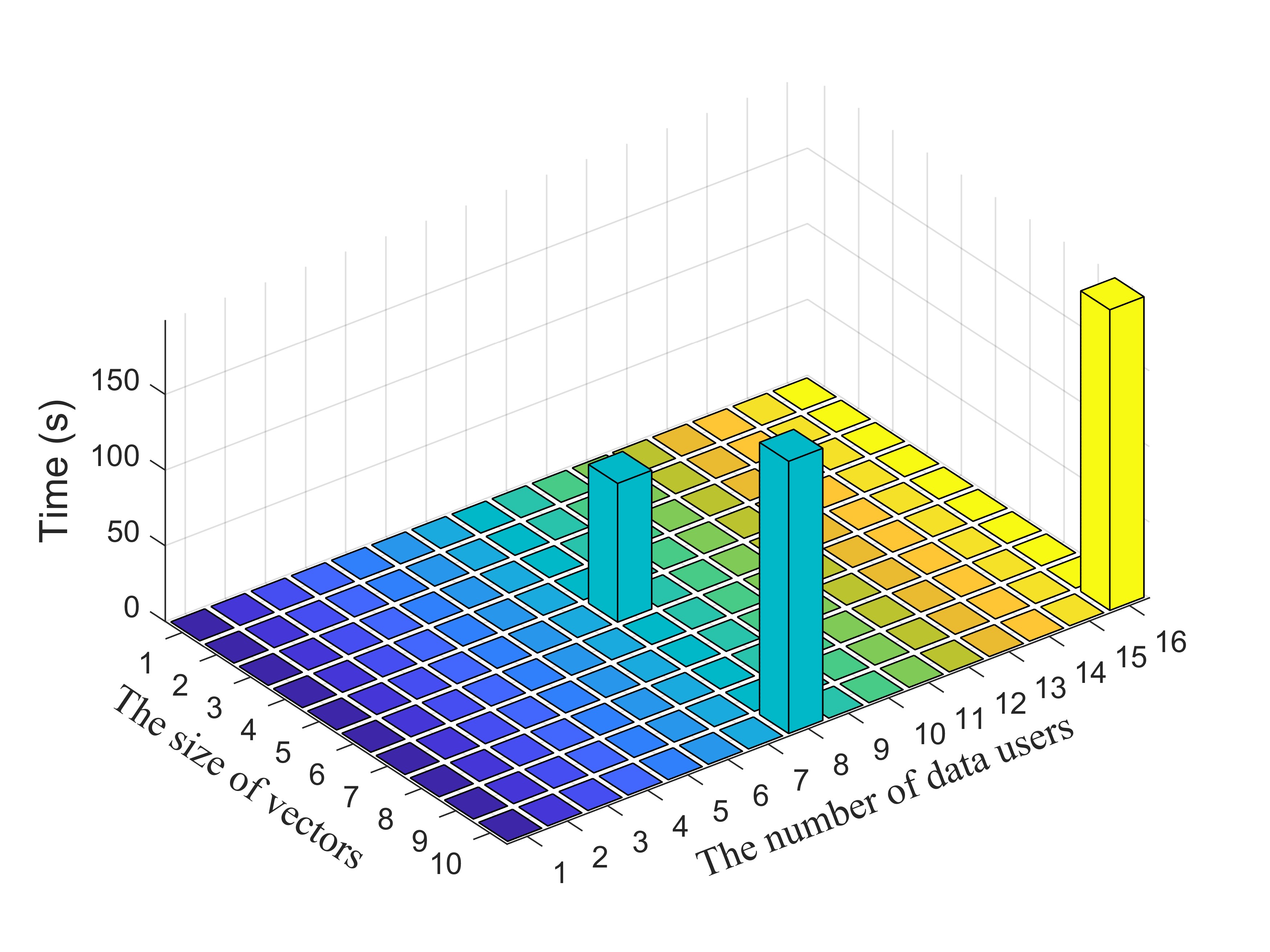}}
        \subfloat[$Enc$]{\label{figure:test_Enc}\includegraphics[width = 0.25\textwidth]{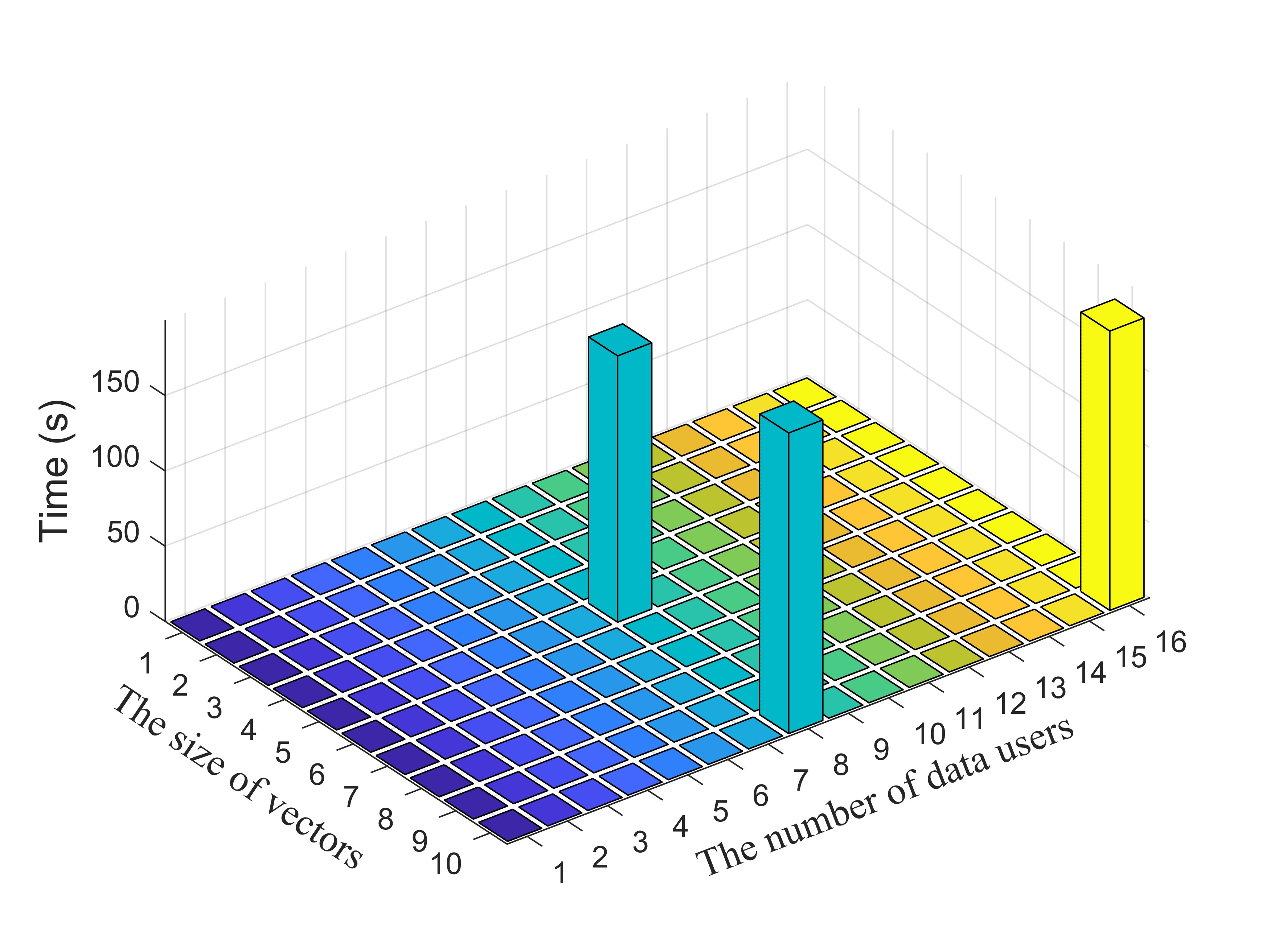}}
        
        \subfloat[$dTrapdoor$]{\label{figure:test_dTrapdoor}\includegraphics[width = 0.25\textwidth]{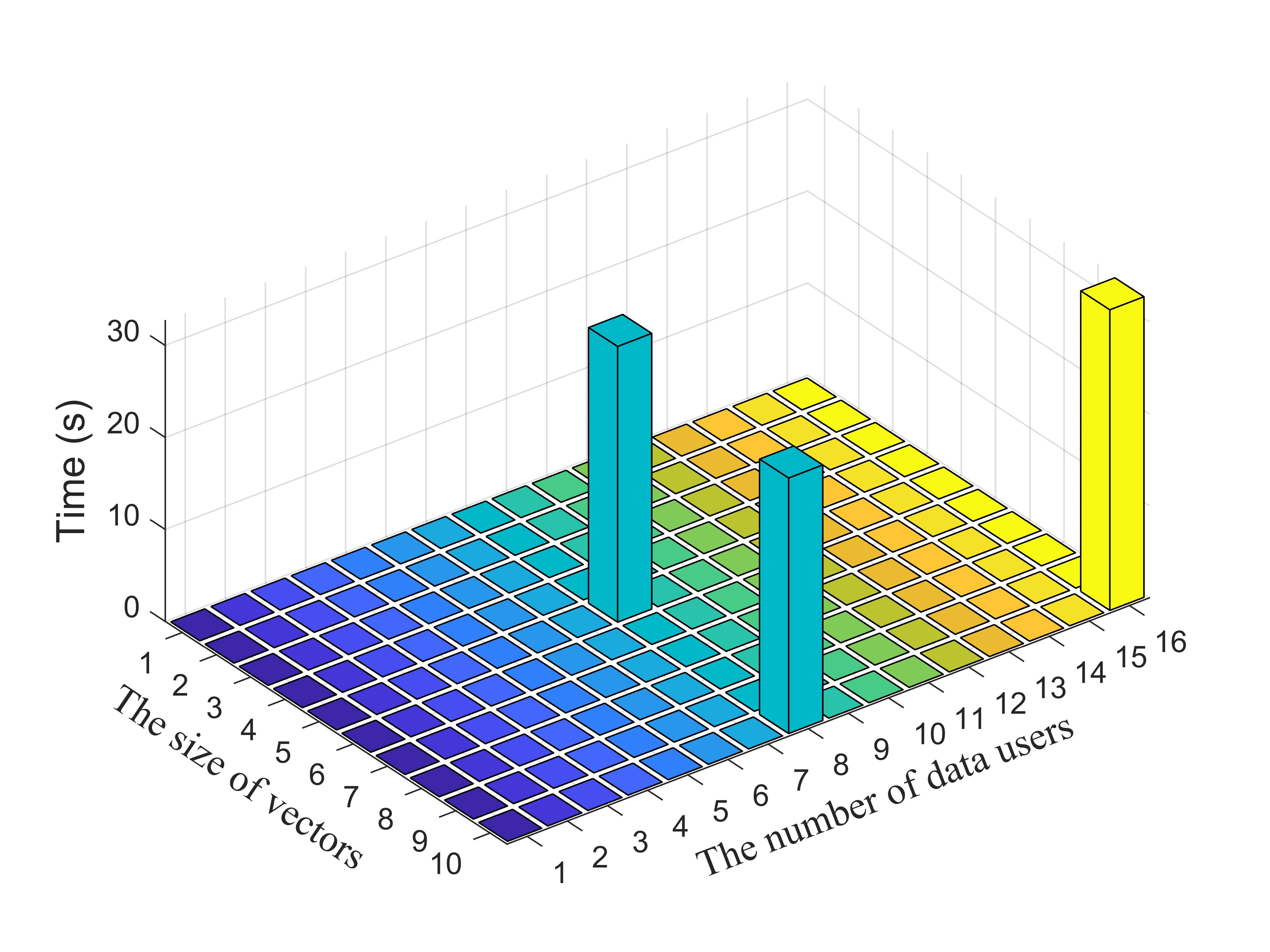}}
        \subfloat[$Test$]{\label{figure:test_Test}\includegraphics[width = 0.25\textwidth]{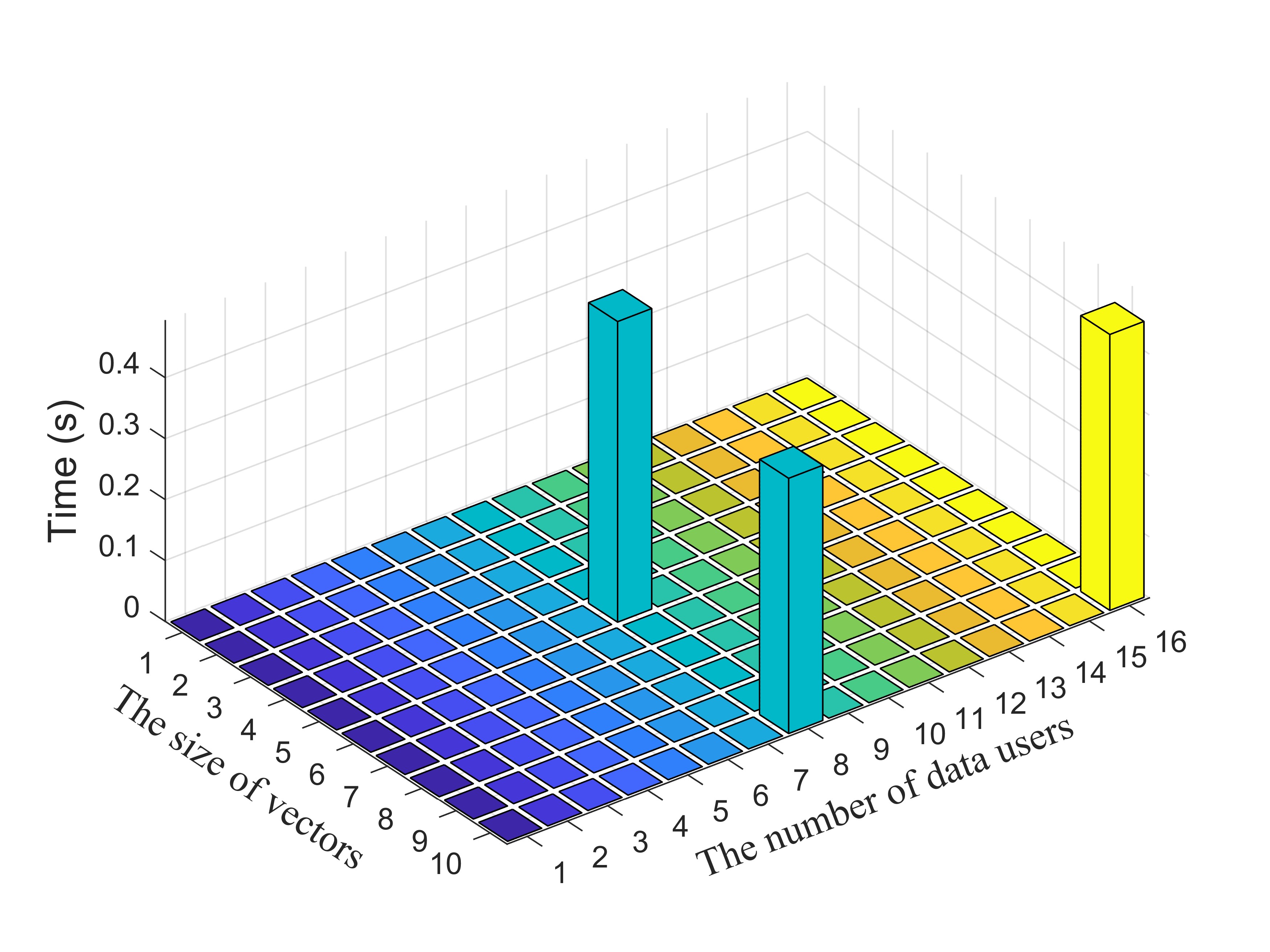}}
        \subfloat[$Dec$]{\label{figure:test_Dec}\includegraphics[width = 0.25\textwidth]{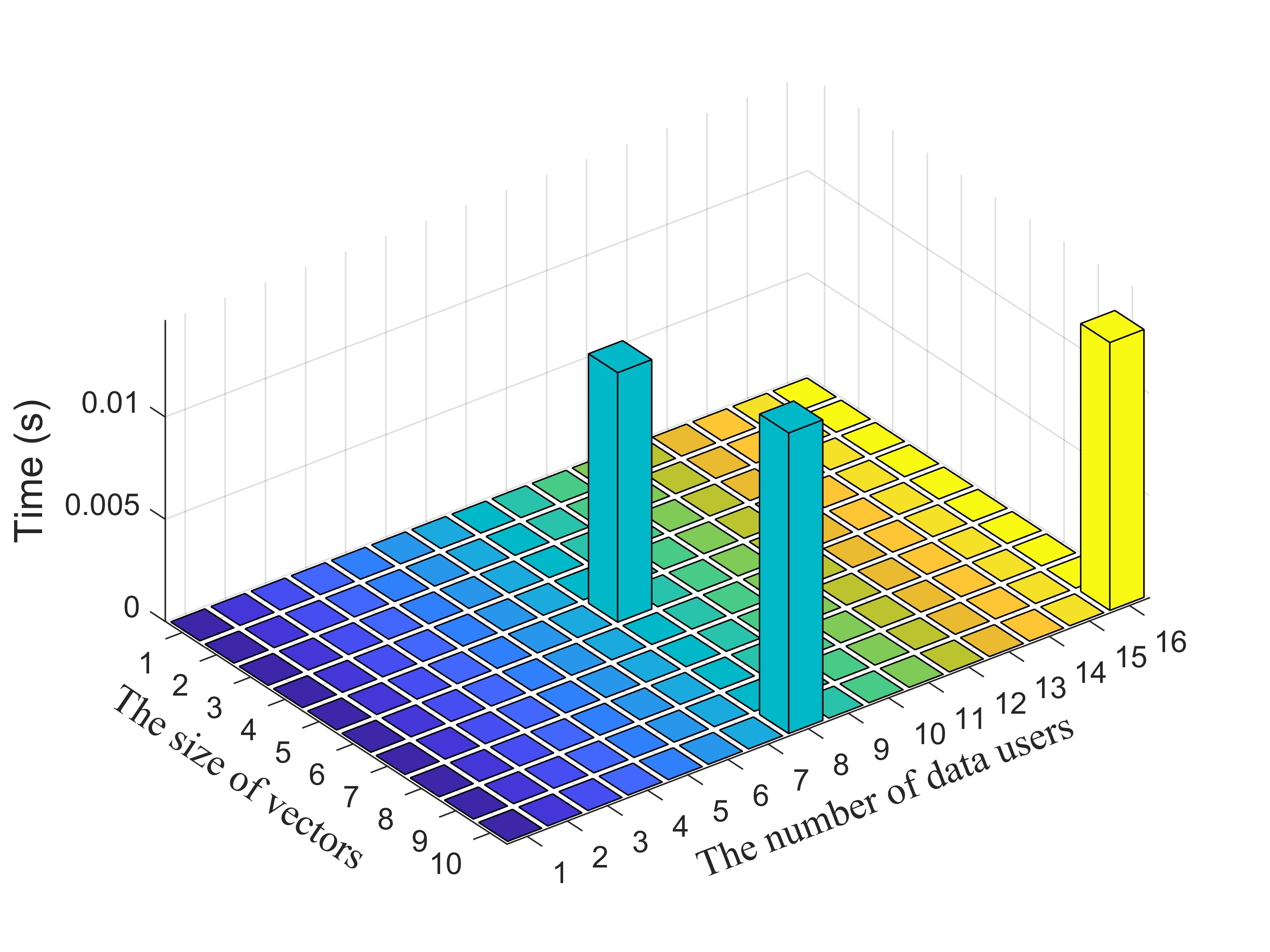}}
        \subfloat[$Revoke$]{\label{figure:test_Revoke}\includegraphics[width = 0.25\textwidth]{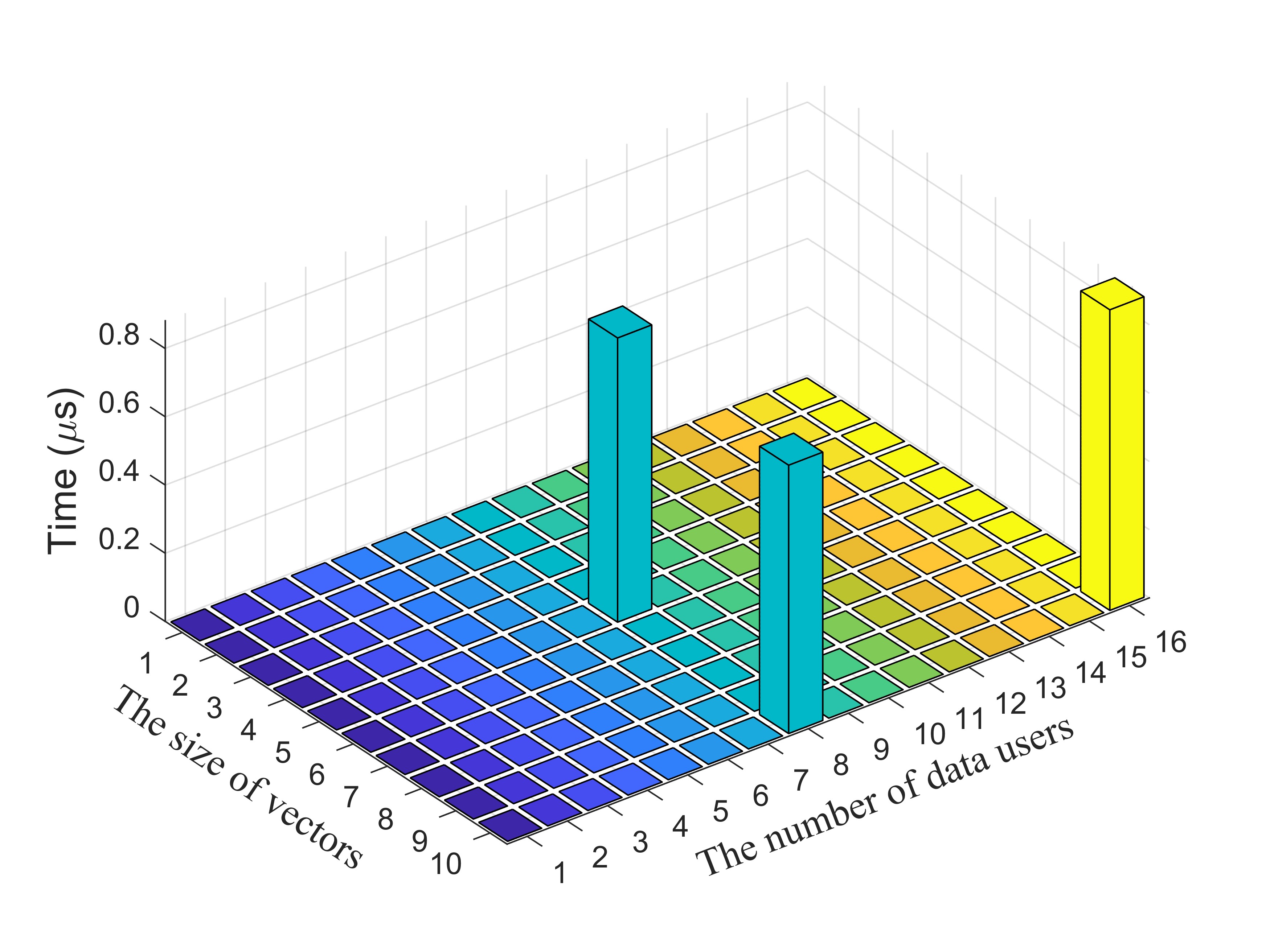}}
        \caption{The computational cost of algorithms}\label{figure:test_time}
\end{figure*}

In this section, we conduct experiments and evaluations on EQDDA-RKS using Sagemath \cite{sage} under the condition of n=64. All experiments were conducted on an Inspiron 13 5320 laptop from Dell, operating on Ubuntu 22.04.4 LTS and equipped with a 12th-generation Core i5-1240P processor, a 16-GB RAM, and a 512-GB SSD. 

Each algorithm was tested under three parameter configurations (1) $l=5$, $N=8$; (2) $l=10$, $N=8$; and (3) $l=10$, $N=16$, where $l$ denotes the vector length and $N$ represents the number of users. The computation overhead of our scheme is illustrated in Figure \ref{figure:test_time}.

Figure. \ref{figure:test_Setup} demonstrates that the $Setup$ algorithm requires 0.4363 s, 0.4406 s, and 0.4641 s in the three parameter settings, showing a linear increase with respect to $l$.
As shown in Figure. \ref{figure:test_SerKG}, the runtime of $SerKG$ remains essentially constant, with values of 269.6263 s, 263.9048 s, and 265.8749 s across the three cases.
Figure. \ref{figure:test_UserKG} similarly shows a constant cost for $UserKG$, namely 59.6338 s, 59.2706 s, and 60.9676 s.
Figure. \ref{figure:test_Token} demonstrates that the $Token$ procedure grows linearly in both $l$ and $N$, taking 130.7054 s, 268.5413 s, and 335.9397 s, respectively. 

Figure. \ref{figure:test_UpdKG} indicates that the $UpdKG$ algorithm takes 97.0945 s, 194.9672 s, and 271.7874 s in the three configurations, showing a linear dependency on both $l$ and $N$.
Figure. \ref{figure:test_TranKG} further shows that the cost of $TranKG$ also increases linearly with $l$ and $N$, yielding runtimes of 0.0014 s, 0.0018 s, and 0.0027 s, respectively.
As depicted in Figure. \ref{figure:test_FunKG}, the $FunKG$ procedure requires 91.666 s, 179.8688 s, and 198.6226 s, which reflects a linear relationship with $N$.
Finally, Figure. \ref{figure:test_Enc} demonstrates that the $Enc$ algorithm similarly grows linearly in $N$, with execution times of 176.856 s, 199.6392 s, and 185.5237 s under the three settings.

Figure. \ref{figure:test_dTrapdoor} shows that the $dTrapdoor$ procedure maintains a constant runtime, recorded as 29.9673 s, 27.7765 s, and 32.6723 s across the three settings.
In Figure. \ref{figure:test_Test}, the $Test$ algorithm likewise exhibits a constant cost, requiring 0.4933 s, 0.4189 s, and 0.4528 s in the respective cases.
Figure. \ref{figure:test_Dec} reports that $Dec$ runs in 0.0122 s, 0.0147 s, and 0.0131 s, showing a linear dependence on $l$.
Finally, Figure. \ref{figure:test_Revoke} indicates that the $Revoke$ algorithm remains constant as well, with execution times of 0.8331 $\mu$s, 0.7867 $\mu$s, and 0.8812 $\mu$s.

\subsection{Comparison}
In this section, we compare $Setup$, $Enc$, $FunKG$ and $Dec$ algorithms with \cite{han2025inner} in Table \ref{tab:comparison_cost}. 

In terms of the $Dec$ algorithm, our scheme is more efficient than \cite{han2025inner}. However, since our scheme supports keyword search, its $Setup$ and $Enc$ algorithms are less efficient than those of \cite{han2025inner}. 
Moreover, to enable temporary delegation, our scheme employs the $SampleLeft$ algorithm in the $FunKG$ algorithm, which leads to higher computational overhead compared with \cite{han2025inner}.

\begin{table}[!ht] 
    \small
    \caption{Comparison} \label{tab:comparison_cost}
    \centering
    \resizebox{0.8\linewidth}{!}{
    \begin{tabular}{|c|c|c|c|}
    \hline
    \multicolumn{4}{|c|}{$Setup$ (s)}\\  
    \hline
    \diagbox[]{Scheme}{$(l,N)$}& (5,8) &(10,8)& (10,16)\\
    \hline
    \cite{han2025inner}&0.0161&0.02 & 0.0189\\ 
    \hline
    EQDDA-RKS&0.4363&0.4406&0.4641\\
    \hline
    \multicolumn{4}{|c|}{$FunKG$ (s)}\\
    \hline
    \diagbox[]{Scheme}{$(l,N)$}& (5,8) &(10,8)& (10,16)\\
    \hline
    \cite{han2025inner}&0.1225&0.1621 &0.2006\\ 
    \hline
    EQDDA-RKS&91.666&179.8688&178.6226\\
    \hline
    \multicolumn{4}{|c|}{$Enc$ (s)}\\  
    \hline
    \diagbox[]{Scheme}{$(l,N)$}& (5,8) &(10,8)& (10,16)\\
    \hline
    \cite{han2025inner}&0.0021&0.0025 & 0.0023\\ 
    \hline
    EQDDA-RKS&176.856&199.6392&178.6226\\
    \hline
    \multicolumn{4}{|c|}{$Dec$ (s)}\\  
    \hline
    \diagbox[]{Scheme}{$(l,N)$}& (5,8) &(10,8)& (10,16)\\
    \hline
    \cite{han2025inner}&2.4528&2.7177 & 2.7607\\ 
    \hline
    EQDDA-RKS&0.0122&0.0147&0.0122\\
    \hline
    \end{tabular}}
\end{table}

\section{Security Proof} \label{section:SP}
This section presents the proofs of sIND-CPA, KC-sIND-CPA, and KT-sIND-CPA security for our scheme.

\begin{theorem} \label{proof1}
    Our scheme achieves $(t',\epsilon'(\lambda))$-sIND-CPA security provided that the ALS-IPFE construction in \cite{agrawal2016fully} satisfies $(t,\epsilon(\lambda))$-IND-CPA security,
where $t = O(t')$ and $\epsilon(\lambda) \ge \epsilon'(\lambda)$. 
\end{theorem}

Assume that a PPT adversary $\mathcal{A}$ attains a non-negligible advantage against the sIND-CPA security of our construction.
In that case, $\mathcal{A}$ can be transformed into an algorithm $\mathcal{B}$ that compromises the IND-CPA security of the ALS-IPFE scheme \cite{agrawal2016fully}.
Following the methodology of \cite{qin2015server}, we analyse the reduction by considering two separate adversarial types.
 
    $\mathcal{A}_1:$ The adversary issues a UserKG query for the target identity $ID_u^*$, then the function computing rights of $ID_u^*$ with respect to the challenge function $\mathbf{x}^*$ must be revoked before the challenge time $t^*$.

    $\mathcal{A}_2:$ The adversary refrains from requesting any UserKG query for the target identity $ID_u^*$. They are permitted to submit FunKG queries on tuples $(ID_u^*, \mathbf{x}, t)$ as long as $(\mathbf{x}, t) \neq (\mathbf{x}^*, t^*)$.

    Algorithm $\mathcal{B}$ starts by making a random guess about which type of adversary it is interacting with.
The remainder of the proof of Theorem~\ref{proof1} proceeds by invoking the strategy-dividing lemma presented in \cite{nguyen2016server,katsumata2020lattice}.

    \begin{lemma}\label{PL1}
    Assume there exists a PPT adversary $\mathcal{A}_1$ that compromises the sIND-CPA security of our construction with advantage $\epsilon'_1(\lambda)$.
From $\mathcal{A}_1$, one can build an algorithm $\mathcal{B}$ that violates the IND-CPA security of the ALS-IPFE scheme, obtaining an advantage $\epsilon_1(\lambda)$ such that $\epsilon_1(\lambda) \ge \epsilon'_1(\lambda)$.
    \end{lemma}

\begin{proof}
    $\mathbf{Game\ 0:}$ The sIND-CPA game is defined in its original form in Section \ref{SM}. 
    
    $\mathbf{Game\ 1:}$ This game is the same as the above one, except for the following differences. 
    Instead of running $\mathsf{TrapGen}$, we randomly select $\mathbf{A} \in \mathbb{Z}_q^{n\times m}$. 
    In addition, $\mathbf{G}$ is obtained by executing $\mathsf{TrapGen}(n,m,q) \to (\mathbf{G}, \mathbf{T}_{\mathbf{G}})$.
    We then sample $\mathbf{R}_1^*, \mathbf{R}_2^* \in {1,-1}^{m \times m}$ at random and define 
    $\mathbf{B}_1 = \mathbf{A}\mathbf{R}_1^*- H_1(ID_u^*)\mathbf{G}$, $\mathbf{B}_2 = \mathbf{A}\mathbf{R}_2^*- H_1(t^*)\mathbf{G}$. 

    $\mathbf{Game\ 2:}$ In this variant of the game, the procedure for generating $\mathbf{U}_{\theta,1}$ and $\mathbf{U}_{\theta,2}$ is altered. 
    The following two cases are considered. 

    For $\theta \in \mathsf{path}(ID_u^*)$, 
    we sample $\mathbf{Z}_{ID_u^*,\theta} \in \mathcal{D}_{\mathbb{Z}^{m\times l},\rho}$, and
    compute $\mathbf{U}_{\theta,1} = \mathbf{A}_{ID_u^*} \mathbf{Z}_{ID_u^*,\theta}$, 
    $\mathbf{U}_{\theta,2} = \mathbf{U}-\mathbf{U}_{\theta,1}$. 

    For $\theta \notin \mathsf{path}(ID_u^*)$, 
    we sample $\mathbf{Z}_{t^*,\theta} \in \mathcal{D}_{\mathbb{Z}^{m\times l},\rho}$, and 
    compute $\mathbf{U}_{\theta,2} = \mathbf{A}_{t^*} \mathbf{Z}_{t^*,\theta}$, 
    $\mathbf{U}_{\theta,1} = \mathbf{U}-\mathbf{U}_{\theta,2}$.

    $\mathbf{Game\ 3:}$ In this game, we establish a reduction from the sIND-CPA security of our construction to the IND-CPA security of the ALS-IPFE scheme.
    Let $\mathcal{S}$ be the challenger in the IND-CPA game associated with ALS-IPFE.

    $\mathbf{Init.}$ $\mathcal{A}_1$ chooses a challenged identity $ID_{u}^{*}$, a vector $\mathbf{x}^{*}$ and a timestamp $t^{*}$, and submits $(ID_u^*,\mathbf{x}^*,t^*)$ to $\mathcal{B}$. 

    $\mathbf{Setup.}$ $\mathcal{B}$ obtains $(\mathbf{A}_{ALS},\mathbf{U}_{ALS})$ 
    from the challenger $\mathcal{S}$, 
    runs $\mathsf{TrapGen}(n,m,q) \to (\mathbf{G},\mathbf{T}_{\mathbf{G}})$, 
    randomly selects $\mathbf{C}_1, \ldots \mathbf{C}_k \in \mathbb{Z}_1^{n\times m}$, 
    $\mathbf{R}_1^*, \mathbf{R}_2^* \in \{-1,1\}^{m\times m}$, 
    $\mathbf{V} \in \mathbb{Z}_q^{n\times h_1}$,
    $\mathbf{v} \in \mathbb{Z}_q^{n\times 1}$, 
    $\mathbf{Z}^* \gets \mathcal{D}_{\mathbb{Z}^{m\times l},\rho}$
    sets $\mathbf{A} = \mathbf{A}_{ALS}$, $\mathbf{U} = \mathbf{U}_{ALS} + \mathbf{A}\mathbf{R}_2^* \mathbf{Z}^*$, 
    $\mathbf{B}_1 = \mathbf{A} \mathbf{R}_1^* - H_1(ID_u^*) \mathbf{G}$, 
    $\mathbf{B}_2 = \mathbf{A} \mathbf{R}_2^* - H_1(t^*) \mathbf{G}$, 
    and sends $(\mathbf{A},\mathbf{B}_1,\mathbf{B}_2,\mathbf{G},
    \mathbf{C}_1,\ldots,\mathbf{C}_k,\mathbf{U},\mathbf{V},\mathbf{v})$ 
    to $\mathcal{A}$. 

    $\mathbf{Phase-1.}$ 

    SerKG Query.
    $\mathcal{A}_1$ submits $ID_s$ to $\mathcal{B}$.
    Then, $\mathcal{B}$ runs 
    $\mathsf{SampleRight}(\mathbf{A}, \mathbf{G}',\mathbf{R}_1^*,\mathbf{T}_{\mathbf{G}},\mathbf{v},\rho)
    \to \mathbf{z}_{ID_s}$, 
    $\mathsf{SampleRight}(\mathbf{A}, \mathbf{G}',\mathbf{R}_1^*,\mathbf{T}_{\mathbf{G}},\mathbf{V},\rho)
    \to \mathbf{Z}_{ID_s}$
    where $\mathbf{G}' = (H_1(ID_s)-H_1(ID_u^*))\cdot \mathbf{G}$,
    and sends $(\mathbf{z}_{ID_s},\mathbf{Z}_{ID_s})$ to $\mathcal{A}$. 
    Notably, we know that $H_1(ID_s)-H_1(ID_u^*)$ is full rank and 
    therefore $\mathbf{T}_{\mathbf{G}}$ is is also
    a trapdoor for the lattice $\Lambda_q^{\bot}(\mathbf{G}')$. 

    UserKG Query. 
    $\mathcal{A}_1$ submits $ID_u$ to $\mathcal{B}$.
    Then, 
    $\mathcal{B}$ runs $\mathsf{SampleBasisRight}(\mathbf{A},\mathbf{G}',\mathbf{R}_1^*,\mathbf{T}_{\mathbf{G}},\rho)
    \to \mathbf{T}_{\widetilde{ID_u}}$ 
    where $\mathbf{G}' = (H_1(\widetilde{ID_u})-H_1(ID_u^*))\cdot \mathbf{G}$
    , and 
    sends $\mathbf{T}_{\widetilde{ID_u}}$
    to $\mathcal{A}_1$. 

    Token Query. 
    $\mathcal{B}$ sets $\mathbf{U}_{\theta,1}$ and $\mathbf{U}_{\theta,2}$ 
    as $\mathbf{Game\ 2}$. 
    $\mathcal{A}_1$ submits $ID_u$ to $\mathcal{B}$. 
    If $ID_u=ID_u^*$, $\mathcal{B}$ returns $(\theta,\mathbf{Z}_{ID_u^*,\theta})_{\theta\in \mathsf{path}(ID_u^*)}$ to $\mathcal{A}_1$; 
    otherwise, for $\theta\in \mathsf{path}(ID_u)$, $\mathcal{B}$ runs 
    $\mathsf{SampleRight}(\mathbf{A},\mathbf{G}',\mathbf{R}^*_{1},\mathbf{T}_{\mathbf{G}},\mathbf{U}_{\theta,1},\rho)
    \to \mathbf{Z}_{ID_u,\theta}$ where 
    $\mathbf{G}' = (H_1(\widetilde{ID_u})-H_1(ID_u^*))\cdot \mathbf{G}$, and 
    returns $(\theta,\mathbf{Z}_{ID_u,\theta})_{\theta\in \mathsf{path}(ID_u)}$ to $\mathcal{A}_1$. 

    UpdKG Query. 
    $\mathcal{A}_1$ submits $(\mathbf{x},t)$ to $\mathcal{B}$.  
    
    If $t = t^*$ and $\mathbf{x}= \mathbf{x}^*$,  
    we have $\mathsf{path}(ID_u^*) \bigcap \mathsf{KUNodes}(BT,RL_{\mathbf{x}},t^*)
    = \emptyset$. 
    Then, $\mathcal{B}$ returns 
    $(\theta,\mathbf{Z}_{t^*,\theta} \cdot \mathbf{x}^*)_{\theta\in \mathsf{KUNodes}(BT,RL_{\mathbf{x}},t^*)}$ 
    to $\mathcal{A}_1$. 
    
    If $t=t^*$ and $\mathbf{x} \neq \mathbf{x}^*$,  $\mathcal{B}$ forwards $\mathbf{x}$ to $\mathcal{S}$, and 
    obtains $sk_\mathbf{x}^{ALS}$. 
    For each $\theta \in \mathsf{KUNodes}(BT,RL_{\mathbf{x}},t) 
    \bigcap \mathsf{path}(ID_u^*)$,  
    $\mathcal{B}$ computes and returns 
    $(\theta,\frac{sk^{ALS}_{\mathbf{x}}-\mathbf{z}^0_{ID^*_u,\theta}\cdot \mathbf{x} - \mathbf{R}_1^*\cdot\mathbf{z}^1_{ID^*_u,\theta} \cdot \mathbf{x} }{\mathbf{Z}^*\cdot \mathbf{x}})$ 
    to $\mathcal{A}_1$ where $\mathbf{Z}_{ID^*_u,\theta} \cdot \mathbf{x} = 
\left[\begin{matrix}
\mathbf{z}^0_{ID^*_u,\theta,\mathbf{x}} \\
\mathbf{z}^1_{ID^*_u,\theta,\mathbf{x}}
\end{matrix}\right] \in \mathbb{Z}^{2m}$ and $\mathbf{z}^0_{ID^*_u,\theta,\mathbf{x}} \in \mathbb{Z}^m$. 
    For each $\theta \in \mathsf{KUNodes}(BT,RL_{\mathbf{x}},t)$ and $\theta \notin \mathsf{path}(ID_u^*)$, 
    $\mathcal{B}$ returns $(\theta,\mathbf{Z}_{t^*,\theta}\cdot \mathbf{x})$ 
    to $\mathcal{A}_1$. 
    
    If $t\neq t^*$, for each $\theta \in \mathsf{KUNodes}(BT,RL_{\mathbf{x}},t)$, 
    $\mathcal{B}$ runs 
    $\mathsf{SampleRight}(\mathbf{A},\mathbf{G}',\mathbf{R}_{2}^*,\mathbf{T}_{\mathbf{G}},\mathbf{U}_{\theta,2},\rho) 
    \to \mathbf{Z}_{t,\theta}$, where $\mathbf{G}' = (H_1(t)-H_1(t^*))\cdot \mathbf{G}$, 
    and sends $(\theta,\mathbf{Z}_{t,\theta} \cdot \mathbf{x})_{\theta\in \mathsf{KUNodes}(BT,RL_{\mathbf{x}},t^*)}$ 
    to $\mathcal{A}_1$. 

    Let $Table_{1}$ be an initially empty table. If $\mathbf{x}\notin Table_{1}$ 
    and $t=t^*$, 
    $\mathcal{B}$ appends $\mathbf{x}$ to $Table_{1}$.

    FunKG Query. 
    $\mathcal{A}_1$ submits $(ID_u,\mathbf{x},t)$ to $\mathcal{B}$.
    $\mathcal{B}$ runs $\mathsf{SampleBasisRight}(\mathbf{A}, \mathbf{G}',\mathbf{R}_{1}^*,\mathbf{T}_{\mathbf{G}},\rho)
    \to \mathbf{T}_{\widetilde{ID_u}}$ where $\mathbf{G}' = (H_1(\widetilde{ID_u})-H_1(ID_u^*))\cdot \mathbf{G}$, 
    $\mathsf{SampleLeft}(
    \mathbf{A}_{\widetilde{ID_u}},
    \mathbf{B}_{t}, 
    \mathbf{T}_{\widetilde{ID_u}}, \mathbf{U}, \rho) \to \mathbf{Z}_{\widetilde{ID_u},t}$, 
    and sends $fk_{\widetilde{ID_u},\mathbf{x},t} = \mathbf{Z}_{\widetilde{ID_u},t} \cdot \mathbf{x}$ to $\mathcal{A}_1$. 

    dTrapdoor Query. $\mathcal{A}_1$ submits $(ID_u,ID_s,\omega,t)$ to $\mathcal{B}$. 
    $\mathcal{B}$ runs $\mathsf{SampleBasisRight}(\mathbf{A},\mathbf{G}',\mathbf{R}_1^*,\mathbf{T}_{\mathbf{G}},\rho)
    \to \mathbf{T}_{\widetilde{ID_u}}$ 
    where $\mathbf{G}' = (H_1(\widetilde{ID_u})-H_1(ID_u^*))\cdot \mathbf{G}$. 
    Then, $\mathcal{B}$ can use $\mathbf{T}_{\widetilde{ID_u}}$
    compute $dt_{ID_u,\omega,t}$ as $\mathbf{Game\ 0}$ and 
    returns $dt_{ID_u,\omega,t}$ to $\mathcal{A}_1$. 

    Revoke Query. $\mathcal{A}_1$ submits $RL_{\mathbf{x}},ID_u,t$ to $\mathcal{B}$. 
    Then, $\mathcal{B}$ runs $Revoke(ID_u,t,RL_{\mathbf{x}},st) \to RL_{\mathbf{x}}$, and sends $RL_{\mathbf{x}}$ to $\mathcal{A}_1$.

    $\mathbf{Challenge.}$ $\mathcal{A}_1$ submits $(\mathbf{y}^*_0,\mathbf{y}^*_1)$ to $\mathcal{B}$, with the restriction that 
    $\left\langle \mathbf{x},\mathbf{y}^*_0\right\rangle = \left\langle \mathbf{x},\mathbf{y}^*_1\right\rangle$ 
    for $\forall \mathbf{x} \in Table_{1}\setminus{\{ \mathbf{x}^*\}}$. 
    $\mathcal{B}$ forwards $(\mathbf{y}^*_0,\mathbf{y}^*_1)$ to $\mathcal{S}$, obtains $(\mathbf{c}_0^{ALS} = 
    \mathbf{A}_{ALS} \cdot \mathbf{s}_0 + \mathbf{e}_0,\mathbf{c}_1^{ALS}
    = \mathbf{U}_{ALS} \cdot \mathbf{s}_0 + \left\lfloor \frac{q}{K} \right\rfloor \cdot \mathbf{y}^*_b + \mathbf{e}_1)$, 
    selects $\mathbf{s}_1 \gets \mathbb{Z}_q^n$, $\mathbf{e}_2 \gets \mathcal{D}_{\mathbb{Z}^m,\sigma}$, 
    $\mathbf{R}_3,\mathbf{R}_4 \gets \{-1,1\}^{m\times m}$, and computes ciphertexts
    $\mathbf{c}_0 =  
    (\mathbf{I}_m|\mathbf{R}_1^*| \mathbf{R}_2^*)^T \cdot \mathbf{c}_0^{ALS} 
    = (\mathbf{A}|\mathbf{A}\mathbf{R}_1^*|\mathbf{A}\mathbf{R}_2^*)^T \cdot 
    \mathbf{s}_0 +  
    (\mathbf{I}_m|\mathbf{R}_1^*|\mathbf{R}_2^*)^T \cdot \mathbf{e}_0$, 
    $\mathbf{c}_1 = \mathbf{A}^T_{\widetilde{ID_u}^*,t^*} \cdot \mathbf{s}_1 + \left[\mathbf{I}_m|\mathbf{R}_3|\mathbf{R}_4\right]^T \cdot  \mathbf{e}_2$, 
    $\mathbf{c}_2 = \mathbf{c}_1^{ALS} +  \mathbf{c}_0^{ALS} \mathbf{R}_2 \mathbf{Z}^* + \mathbf{U}^T\cdot \mathbf{s}_1 + \mathsf{NoiseGen}(\mathbf{R}_2 \mathbf{Z}^*,s) = 
    \mathbf{U}^T \cdot (\mathbf{s}_0 + \mathbf{s}_1) +\left\lfloor \frac{q}{K} \right\rfloor \cdot \mathbf{y}^*_b + \mathbf{e}_1 
    + \mathbf{R}_2 \mathbf{Z}^* \mathbf{e}_0
    + \mathsf{NoiseGen}(\mathbf{R}_2 \mathbf{Z}^*,s)$. 

    Then, $\mathcal{B}$ computes $\mathbf{c}_3,\mathbf{c}_4,\mathbf{c}_5$ 
    as $\mathbf{Game\ 0}$ and returns $CT^* = (\mathbf{c}_0,\mathbf{c}_1,\mathbf{c}_2,\mathbf{c}_3,\mathbf{c}_4,\mathbf{c}_5)$ to $\mathcal{A}_1$.

    $\mathbf{Phase-2.}$
    $\mathcal{A}_1$ may continue issuing SerKG, UserKG, Token, UpdKG, FunKG, dTrapdoor, and Revoke queries,
    with the exception that FunKG queries on any $\mathbf{x}$ must satisfy
    $\left\langle \mathbf{x},\mathbf{y}^*_0 \right\rangle \neq \left\langle \mathbf{x},\mathbf{y}^*_1 \right\rangle$, 
    Under this restriction, the challenger $\mathcal{B}$ responds in the same manner as in $\mathbf{Phase-1}$.

    $\mathbf{Guess.}$ 
    $\mathcal{A}_1$ eventually outputs a bit $b' \in {0,1}$ as its guess for $b$,
and $\mathcal{B}$ relays this value to $\mathcal{S}$ as its own guess.

Whenever $b' = b$, $\mathcal{A}_1$ succeeds, and consequently $\mathcal{B}$ violates the IND-CPA security of the ALS-IPFE scheme.
Let $\epsilon_1(\lambda)$ be the maximal advantage achievable by any PPT adversary against the IND-CPA security of ALS-IPFE.
It follows that the advantage $\epsilon'_1(\lambda)$ of any PPT adversary $\mathcal{A}_1$ in the sIND-CPA game of our scheme satisfies $\epsilon'_1(\lambda) \leq \epsilon_1(\lambda)$.

\end{proof}

\begin{lemma}\label{PL2}
    Assume there exists a PPT adversary $\mathcal{A}_2$ that achieves an advantage of $\epsilon'_2(\lambda)$ against the sIND-CPA security of our construction.
From $\mathcal{A}_2$, we can build an algorithm $\mathcal{B}$ that compromises the IND-CPA security of the ALS-IPFE scheme \cite{agrawal2015practical} with advantage $\epsilon_2(\lambda)$,
where $\epsilon_2(\lambda) \ge \epsilon'_2(\lambda)$.
\end{lemma}
\begin{proof}
    $\mathbf{Game\ 0:}$ The sIND-CPA game is presented in its original form in Section \ref{SM}.

    $\mathbf{Game\ 1:}$ This game is the same as the above one, except for the following differences. 
    Instead of running $\mathsf{TrapGen}$, we randomly select $\mathbf{A} \in \mathbb{Z}_q^{n\times m}$. 
    In addition, $\mathbf{G}$ is produced by executing $\mathsf{TrapGen}(n,m,q) \to (\mathbf{G}, \mathbf{T}_{\mathbf{G}})$.
    We then sample $\mathbf{R}_1^*$ and $\mathbf{R}_2^*$ uniformly from $\{1,-1\}^{m \times m}$ and define 
    $\mathbf{B}_1 = \mathbf{A}\mathbf{R}_1^*- H_1(\widetilde{ID_u^*})\mathbf{G}$, $\mathbf{B}_2 = \mathbf{A}\mathbf{R}_2^*- H_1(t^*)\mathbf{G}$.  

    $\mathbf{Game\ 2:}$ In this game, the generation method of $\mathbf{U}_{\theta,1}$ and $\mathbf{U}_{\theta,2}$ is modified.

    For each node $\theta$, 
    we sample $\mathbf{Z}_{t^*,\theta} \in \mathcal{D}_{\mathbb{Z}^{m\times l},\rho}$, and 
    compute $\mathbf{U}_{\theta,2} = \mathbf{A}_{t^*} \mathbf{Z}_{t^*,\theta}$, 
    $\mathbf{U}_{\theta,1} = \mathbf{U}-\mathbf{U}_{\theta,2}$.

    $\mathbf{Game\ 3:}$ In this game, we reduce the sIND-CPA security of our scheme to the IND-CPA security of the ALS-IPFE scheme.
    Let $\mathcal{S}$ represent the challenger in the IND-CPA game for ALS-IPFE.

    $\mathbf{Init.}$ $\mathcal{A}_2$ chooses a challenged identity $ID_{u}^{*}$, a vector $x^{x}$ and a timestamp $t^{*}$, and submits $(ID_u^*,\mathbf{x}^*,t^*)$ to $\mathcal{B}$. 
    
    $\mathbf{Setup.}$ $\mathcal{B}$ obtains $(\mathbf{A}_{ALS},\mathbf{U}_{ALS})$ 
    from the challenger $\mathcal{S}$, 
    runs $\mathsf{TrapGen}(n,m,q) \to (\mathbf{G},\mathbf{T}_{\mathbf{G}})$, 
    randomly selects $\mathbf{C}_1, \ldots \mathbf{C}_k \in \mathbb{Z}_1^{n\times m}$, 
    $\mathbf{R}_1^*, \mathbf{R}_2^* \in \{-1,1\}^{m\times m}$, 
    $\mathbf{V} \in \mathbb{Z}_q^{n\times h_1}$,
    $\mathbf{v} \in \mathbb{Z}_q^{n\times 1}$, 
    sets $\mathbf{A} = \mathbf{A}_{ALS}$, 
    $\mathbf{U} = \mathbf{U}_{ALS}$, 
    $\mathbf{B}_1 = \mathbf{A} \mathbf{R}_1^* - H_1(\widetilde{ID_u^*}) \mathbf{G}$, 
    $\mathbf{B}_2 = \mathbf{A} \mathbf{R}_2^* - H_1(t^*) \mathbf{G}$, 
    and sends $(\mathbf{A},\mathbf{B}_1,\mathbf{B}_2,\mathbf{G},\mathbf{C}_1,\ldots,\mathbf{C}_k,\mathbf{G},\mathbf{U},\mathbf{V},\mathbf{v})$ 
    to $\mathcal{A}_2$. 

    $\mathbf{Phase-1.}$ 
    
    SerKG Query.
    $\mathcal{A}_2$ submits $ID_s$ to $\mathcal{B}$.
    Then, $\mathcal{B}$ runs 
    $\mathsf{SampleRight}(\mathbf{A}, \mathbf{G}',\mathbf{R}_1^*,\mathbf{T}_{\mathbf{G}},\mathbf{v},\rho)
    \to \mathbf{z}_{ID_s}$, 
    $\mathsf{SampleRight}(\mathbf{A}, \mathbf{G}',\mathbf{R}_1^*,\mathbf{T}_{\mathbf{G}},\mathbf{V},\rho)
    \to \mathbf{Z}_{ID_s}$
    where $\mathbf{G}' = (H_1(ID_s)-H_1(\widetilde{ID_u^*}))\cdot \mathbf{G}$,
    and sends $(\mathbf{z}_{ID_s},\mathbf{Z}_{ID_s})$ to $\mathcal{A}_2$. 
    Notably, we know that $H_1(ID_s)-H_1(\widetilde{ID_u^*})$ is full rank and 
    therefore $\mathbf{T}_{\mathbf{G}}$ is is also
    a trapdoor for the lattice $\Lambda_q^{\bot}(\mathbf{G}')$.   

    UserKG Query. 
    $\mathcal{A}_2$ submits $ID_u \neq ID_u^*$ to $\mathcal{B}$. 
    Then, 
    $\mathcal{B}$ runs $\mathsf{SampleBasisRight}(\mathbf{A},\mathbf{G}',\mathbf{R}_1^*,\mathbf{T}_{\mathbf{G}},\rho)
    \to \mathbf{T}_{\widetilde{ID_u}}$ 
    where $\mathbf{G}' = (H_1(\widetilde{ID_u})-H_1(\widetilde{ID^*_u}))\cdot \mathbf{G}$
    , and 
    sends $\mathbf{T}_{\widetilde{ID_u}}$
    to $\mathcal{A}_2$. 

    Token Query. 
    $\mathcal{A}_2$ submits $ID_u \neq ID_u^*$ to $\mathcal{B}$. 
    $\mathcal{B}$ sets $\mathbf{U}_{\theta,1}$, $\mathbf{U}_{\theta,2}$ 
    as $\mathbf{Game\ 2}$, runs 
    $\mathsf{SampleRight}(\mathbf{A},\mathbf{G}',\mathbf{R}^*_{1},\mathbf{T}_{\mathbf{G}},\mathbf{U}_{\theta,1},\rho)
    \to \mathbf{Z}_{ID_u,\theta}$ where 
    $\mathbf{G}' = (H_1(\widetilde{ID_u})-H_1(\widetilde{ID^*_u}))\cdot \mathbf{G}$, and 
    returns $(\theta,\mathbf{Z}_{ID_u,\theta})_{\theta\in \mathsf{path}(ID_u)}$ to $\mathcal{A}_2$.     

    UpdKG Query. 
    $\mathcal{A}_2$ submits $(\mathbf{x},t)$ to $\mathcal{B}$. 
    If $t = t^*$, 
    $\mathcal{B}$ returns 
    $(\theta,\mathbf{Z}_{t^*,\theta} \cdot \mathbf{x})_{\theta\in \mathsf{KUNodes}(BT,RL_{\mathbf{x}},t^*)}$ 
    to $\mathcal{A}_2$; 
    otherwise, 
    for each $\theta \in \mathsf{KUNodes}(BT,RL_{\mathbf{x}},t)$, 
     $\mathcal{B}$ runs 
    $\mathsf{SampleRight}(\mathbf{A},\mathbf{G}',\mathbf{R}_{2}^*,\mathbf{T}_{\mathbf{G}},\mathbf{U}_{\theta,2},\rho) 
    \to \mathbf{Z}_{t,\theta}$, where $\mathbf{G}' = (H_1(t)-H_1(t^*))\cdot \mathbf{G}$, 
    and sends $(\theta,\mathbf{Z}_{t,\theta} \cdot \mathbf{x})_{\theta\in \mathsf{KUNodes}(BT,RL_{\mathbf{x}},t)}$ 
    to $\mathcal{A}_2$. 

    FunKG Query. 
    $\mathcal{A}_2$ submits $(ID_u,\mathbf{x},t)$ to $\mathcal{B}$. 
    If $ID_u=ID_u^*$, $\mathcal{B}$ forwards 
    $\mathbf{x}$ to $\mathcal{S}$, and
    obtains $sk_{\mathbf{x}}^{ALS}$. 
    Then, $\mathcal{B}$ returns $sk_{\mathbf{x}}^{ALS}$ 
    to $\mathcal{A}_2$. 
    If $ID_u \neq ID_u^*$, 
    $\mathcal{B}$ runs $\mathsf{SampleBasisRight}(\mathbf{A}, \mathbf{G}',\mathbf{R}_{1}^*,\mathbf{T}_{\mathbf{G}},\rho)
    \to \mathbf{T}_{\widetilde{ID_u}}$ where $\mathbf{G}' = (H_1(\widetilde{ID_u})-H_1(\widetilde{ID_u^*}))\cdot \mathbf{G}$, 
    $\mathsf{SampleLeft}(
    \mathbf{A} _{\widetilde{ID_u}},
    \mathbf{B}_{t}, 
    \mathbf{T}_{\widetilde{ID_u}}, \mathbf{U}, \rho) \to \mathbf{Z}_{\widetilde{ID_u},t}$, 
    and sends $fk_{ID_u,\mathbf{x},t} = \mathbf{Z}_{\widetilde{ID_u},t} \cdot \mathbf{x}$ to $\mathcal{A}_2$. 

    Let $Table_2$ be an initially empty table. If $(ID_u,t) = (ID_u^*,t^*)$ and $\mathbf{x}\notin Table_2$, 
    $\mathcal{B}$ appends $\mathbf{x}$ to $Table_2$. 
    
    dTrapdoor Query. $\mathcal{A}$ submits $(ID_u,ID_s,\omega,t)$ to $\mathcal{B}$. 
    Let 
    $\mathbf{A}_{\widetilde{ID_u},\omega,t}  
    = \left[\mathbf{A} | \mathbf{A}\overline{\mathbf{R}} + \overline{\mathbf{G}} \right]$, 
    where $\overline{\mathbf{R}} = \left[ \mathbf{R}_1^* | \sum_{i = 1}^{k}(b_i\mathbf{F}_i^*)
    |\mathbf{R}_2^*\right]$ and 
    $\overline{\mathbf{G}} = (H_1(\widetilde{ID_u})-H_1(\widetilde{ID_u^*}) +1 + \sum_{i = 1}^{k}(b_ih_i) + H_1(t) -H_1(t^*)) \mathbf{G}$. 
    Then, $\mathcal{B}$ runs $\mathsf{SampleRight}(\mathbf{A},\overline{\mathbf{G}},\overline{\mathbf{R}},
    \mathbf{T}_{\mathbf{G}},\mathbf{v},\rho) \to kt_{ID_u,\omega,t}$
    and 
    computes $dt_{ID_u,\omega,t}$ as $\mathbf{Game\ 0}$.

    Revoke Query. $\mathcal{A}_2$ submits $RL_{\mathbf{x}},ID_u,t$ to $\mathcal{B}$. 
    Then, $\mathcal{B}$ runs $Revoke(ID_u,t,RL_{\mathbf{x}},st) \to RL_{\mathbf{x}}$, and sends $RL_{\mathbf{x}}$ to $\mathcal{A}_2$.

    $\mathbf{Challenge.}$ $\mathcal{A}_2$ submits $(\mathbf{y}^*_0,\mathbf{y}^*_1)$ to $\mathcal{B}$, with the restriction that 
    $\left\langle \mathbf{x},\mathbf{y}^*_0\right\rangle = \left\langle \mathbf{x},\mathbf{y}^*_1\right\rangle$ 
    for $\forall \mathbf{x} \in Table_{2}\setminus{\{ \mathbf{x}^*\}}$. 
    $\mathcal{B}$ forwards $(\mathbf{y}^*_0,\mathbf{y}^*_1)$ to $\mathcal{S}$, and obtains $\mathbf{c}_0^{ALS} = 
    \mathbf{A}_{ALS} \cdot \mathbf{s}_1 + \mathbf{e}_1,\mathbf{c}_1^{ALS}
    = \mathbf{U}_{ALS} \cdot \mathbf{s}_1 + \left\lfloor \frac{q}{K} \right\rfloor \cdot \mathbf{y}^*_b + \mathbf{e}_2.$
    $\mathcal{B}$ randomly selects $\mathbf{s}_0 \gets \mathbb{Z}_q^n$, $\mathbf{e}_0 \gets \mathcal{D}_{\mathbb{Z}^m,\sigma}$, 
    $\mathbf{R}_1,\mathbf{R}_2 \gets \{-1,1\}^{n\times n}$, $\mathbf{e}_3 \gets \mathcal{D}_{\mathbb{Z}^l,\tau}$, 
    and computes  
    $\mathbf{c}_0 = \mathbf{A}_{ID_u^*,t^*} \cdot \mathbf{s}_0 + \left[\mathbf{I}_m|\mathbf{R}_1|\mathbf{R}_2\right]^T \cdot  \mathbf{e}_0, \mathbf{c}_1 = (\mathbf{I}_m|\mathbf{R}_1^*| \mathbf{R}_2^*)^T \cdot \mathbf{c}_0^{ALS} 
    =(\mathbf{A}|\mathbf{A}\mathbf{R}_1^*|\mathbf{A}\mathbf{R}_2^*)^T \cdot 
    \mathbf{s}_1 +  
    (\mathbf{I}_m|\mathbf{R}_1^*|\mathbf{R}_2^*) \cdot \mathbf{e}_1, \mathbf{c}_2 = \mathbf{c}_1^{ALS} + \mathbf{U}^T\cdot \mathbf{s}_1 + \mathbf{e}_3 = 
    \mathbf{U}^T \cdot (\mathbf{s}_0 + \mathbf{s}_1) +\left\lfloor \frac{q}{K} \right\rfloor \cdot \mathbf{y}_b^* + \mathbf{e}_2 +\mathbf{e}_3$. 
    Then, $\mathcal{B}$ computes $\mathbf{c}_3,\mathbf{c}_4,\mathbf{c}_5$ 
    as $\mathbf{Game\ 0}$ and returns $CT^* =(\mathbf{c}_0,\mathbf{c}_1,\mathbf{c}_2,\mathbf{c}_3,\mathbf{c}_4,\mathbf{c}_5)$ to $\mathcal{A}_2$. 
    
    $\mathbf{Phase-2.}$
    $\mathcal{A}_2$ can continue to make SerKG, UserKG, Token, UpdKG, FunKG and dTrapdoor and Revoke queries, 
    expect the FunKG queries for $\mathbf{x}$ satisfying 
    with the restriction that 
    $\left\langle \mathbf{x},\mathbf{y}^*_0 \right\rangle \neq \left\langle \mathbf{x},\mathbf{y}^*_1 \right\rangle$, 
    and $\mathcal{B}$ answers as $\mathbf{Phase-1}$. 

    $\mathbf{Guess.}$ 
    $\mathcal{A}_2$ outputs a bit $b' \in {0,1}$ as its guess of $b$, and $\mathcal{B}$ forwards this value to $\mathcal{S}$.

Whenever $b' = b$, $\mathcal{A}_2$ succeeds, and consequently $\mathcal{B}$ breaches the IND-CPA security of the ALS-IPFE scheme.
Let $\epsilon_2(\lambda)$ denote the maximum advantage attainable by any PPT adversary in the IND-CPA game of ALS-IPFE.
Hence, the advantage $\epsilon'_2(\lambda)$ of any PPT adversary $\mathcal{A}_2$ in attacking the sIND-CPA security of our construction satisfies $\epsilon'_2(\lambda) \leq \epsilon_2(\lambda)$. 
\end{proof}

Finally, because $\mathcal{B}$ randomly guesses the type of 
the adversary, we have 
\begin{align*}
Adv_{\mathcal{A}}^{sIND-CPA} 
= \frac{1}{2}Adv_{\mathcal{A}_1}^{sIND-CPA} + 
\frac{1}{2}Adv_{\mathcal{A}_2}^{sIND-CPA}\\
= \frac{1}{2}(\epsilon'_1(\lambda) 
+ \epsilon'_2(\lambda)) 
\leq \frac{1}{2}(\epsilon_1(\lambda) 
+ \epsilon_2(\lambda)) 
\end{align*}
by Lemma \ref{PL1} and Lemma \ref{PL2}. 

Furthermore, if a PPT adversary $\mathcal{A}$ attains advantage $\epsilon'(\lambda)$ against the sIND-CPA security of our scheme, then we can construct an algorithm $\mathcal{B}$ that breaks the IND-CPA security of ALS-IPFE with advantage $\epsilon(\lambda)$, where $\epsilon(\lambda) \ge \epsilon'(\lambda)$, $\epsilon(\lambda)=\epsilon_1(\lambda)+\epsilon_2(\lambda)$, and $\epsilon'(\lambda)=\epsilon'_1(\lambda)+\epsilon'_2(\lambda)$.

Before proving the KC-sIND-CPA security of our scheme, we first introduce 
a hash family \cite{agrawal2010efficient}. 
\begin{definition}
    Let $\mathcal{H} = { H : X \to Y }$ be a family of hash functions, where $0 \in Y$.
For any $(Q+1)$-tuple $\omega = (\omega_0, \omega_1, \ldots, \omega_Q) \in X^{Q+1}$, define the non-abort probability 
$p(\omega) = \Pr\big[ H(\omega_0)=0 \wedge H(\omega_1)\neq 0 \wedge \dots \wedge H(\omega_Q)\neq 0 \big]$. 
We say that $\mathcal{H}$ is $(Q, p_{\min}, p_{\max})$-abort-resistant if for every tuple $\omega$ satisfying
$\omega_0 \notin {\omega_1, \ldots, \omega_Q}$, the probability $p(\omega)$ lies within the range 
$p(\omega) \in [p_{\min}, p_{\max}]$. 
\end{definition}

In our security analysis, we employ the following abort-resistant hash family, as used in \cite{bellare2009simulation,agrawal2010efficient}. 

For a prime $q$, the hash family $\mathcal{H}: 
\{H_{c}: (\mathbb{Z}_q^k)^* \to \mathbb{Z}_q \}$ is defined 
as $H_{c}(\omega) = 1+ \sum_{i = 1}^{k} \omega_i c_i \in \mathbb{Z}_q$, 
where $\omega =(\omega_1,\ldots \omega_k) \in \{1,-1\}^k$ 
and $c = (c_1,\ldots,c_k) \in \mathbb{Z}_q^k$. 

The hash family $\mathcal{H}$ introduced above satisfies the following properties \cite{agrawal2010efficient}.
\begin{lemma}
    Let $q$ be a prime and $0 < Q < q$. Then, the hash family $\mathcal{H}$ is $(Q, \frac{1}{q}(1-\frac{Q}{q}), \frac{1}{q})$-abort resistant.
\end{lemma}

\begin{theorem} 
    Our scheme is $(t',\epsilon'(\lambda))$-KC-sIND-CPA security against $\mathcal{A}_s$ if the $(t,\epsilon(\lambda))$-LWE assumption holds, 
    where $t= O(t')$ and $\epsilon(\lambda) \geq \frac{\epsilon'(\lambda)}{8q}$. 
\end{theorem}

\begin{proof}  Suppose that a PPT adversary $\mathcal{A}_s$ (a malicious server) attains a non-negligible advantage in breaking the KC-sIND-CPA security of our scheme. Then we can use $\mathcal{A}_s$ to construct an algorithm $\mathcal{B}$ that violates the LWE assumption. As outlined in Section~\ref{section:Pre}, the LWE instance is provided through an oracle $\mathcal{O}$, which is instantiated as either $\mathcal{O}_s$ or $\mathcal{O}_s'$.

    $\mathbf{Game\ 0:}$ The real game between $\mathcal{A}_s$ and $\mathcal{B}$ as shown 
    in section \ref{SM}. 

    $\mathbf{Game\ 1:}$ This game is the same as the above one, except for the following modifications. 
    Instead of running $\mathsf{TrapGen}$, we randomly select $\mathbf{A} \in \mathbb{Z}_q^{n\times m}$, 
    $\mathbf{R}_1^*,\mathbf{R}_2^*,\mathbf{F}_1^*,\ldots \mathbf{F}^*_k \in \{1,-1\}^{m\times m}$, 
    $h_1,\ldots h_k \in \mathbb{Z}_q$.
    Moreover, we  
    generate $\mathbf{G}$ by running $\mathsf{TrapGen}(n,m,q) \to (\mathbf{G},\mathbf{T}_{\mathbf{G}})$, 
    and sets $\mathbf{B}_1 = \mathbf{A} \mathbf{R}_1^* - H_1(\widetilde{ID_u^*})\mathbf{G}$, 
    $\mathbf{B}_2 = \mathbf{A} \mathbf{R}_2^* - H_1(t^*)\mathbf{G}$, 
    $\mathbf{C}_i  = \mathbf{A}\cdot \mathbf{F}^*_i + h_i \mathbf{G}$ for $i=1,\ldots k$. 


    $\mathbf{Game\ 2:}$ In this game, the generation of $\mathbf{U}_{\theta,1}$ and $\mathbf{U}_{\theta,2}$ is modified.

    For each node $\theta$, 
    we sample $\mathbf{Z}_{t^*,\theta} \in \mathcal{D}_{\mathbb{Z}^{m\times l},\rho}$, 
    compute $\mathbf{U}_{\theta,2} = \mathbf{A}_{t^*} \mathbf{Z}_{t^*,\theta}$, 
    $\mathbf{U}_{\theta,1} = \mathbf{U}-\mathbf{U}_{\theta,2}$. 


    $\mathbf{Game\ 3:}$ $\mathbf{Game\ 3}$ is identical to $\mathbf{Game\ 2}$ except that it introduces an artificial abort.
    In this game, $\mathcal{B}$ samples a hash function $H_{kw} \in \mathcal{H}$ uniformly at random and keeps it hidden.
    When $\mathcal{A}_s$ issues SerKG, UserKG, Token, UpdKG, FunKG, or dTrapdoor queries corresponding to $\omega_1,\ldots,\omega_Q$,
    $\mathcal{B}$ answers them exactly as in $\mathbf{Game\ 2}$. 
    After the above queries end, $\mathcal{A}_s$ submits challenge keywords $(\omega_0^*,\omega_1^*)$, which do not belong to
    $\{\omega_1, \ldots \omega_Q \}$. 
    $\mathcal{B}$ forwards challenge keyword ciphertexts 
    related to $\omega_{b}^*$, $ID_u^*$ and $t^*$, where $b\in \{0,1\}$ is randomly selected.
    Then, $\mathcal{A}_s$ returns $b'$ as his guess on $b$. After that, $\mathcal{B}$ validates the guess 
    by checking whether $H_{kw}(\omega_{i}) \neq 0$ for 
    $i=1,2,\ldots,Q$ and $H_{kw}(\omega_{b}^*) = 0$ hold. 
    If these conditions are not satisfied, 
    refreshes $b'$ as a random bit, and aborts the game. 
    Because adversaries do not see $H_{kw}$, they do not know if an abort event happened.
    Referring to Lemma 28 of \cite{agrawal2010efficient}, 
    we have 
    $Pr\left[\neg \mathrm{abort}\right] =  \frac{1}{4q}, 
    Pr\left[\mathbf{Game\ 2}\right] \leq \frac{1}{4q}  Pr\left[\mathbf{Game\ 3}\right].$
    
    $\mathbf{Game\ 4:}$ In this game, we alter the procedure used by $\mathcal{B}$ to generate the challenge ciphertexts as described below.
    The partial challenge ciphertexts related to keywords $(\mathbf{c}_3^*,\mathbf{c}_5^*)$ 
    are randomly chosen from $(\mathbb{Z}_q^{4m},\mathbb{Z}_q)$. 
    Since the partial challenge ciphertexts are  
    random elements, the advantage of $\mathcal{A}_s$ is zero 
    in this game.
    
    $\mathbf{Game\ 5:}$  
    In this game, we reduce the KC-sIND-CPA security against $\mathcal{A}_s$ 
    of our scheme to the LWE assumption.

    $\mathbf{Init.}$ $\mathcal{A}_s$
    chooses a challenged identity $ID_{u}^{*}$ and a timestamp $t^{*}$, and 
    submits $(ID_u^*,t^*)$ to $\mathcal{B}$. 
    
    $\mathbf{Setup.}$ $\mathcal{B}$ queries the LWE oracle $\mathcal{O}$, and obtains pairs $(\mathbf{u}_i,v_i)\in \mathbb{Z}_q^n\times \mathbb{Z}_q$ for  
    $i=0,\ldots m$.
    Let $\mathbf{v} = \mathbf{u}_0 \in \mathbb{Z}_q^n$ and $\mathbf{A} = \left[ \mathbf{u}_1,\ldots \mathbf{u}_m \right]\in \mathbb{Z}_q^{n\times m}$, and set $(\mathbf{B}_1,\mathbf{B}_2,\mathbf{G},
    \mathbf{C}_1,\ldots,\mathbf{C}_k,\mathbf{U},\mathbf{V})$ as $\mathbf{Game\ 1}$. 
    Then, $\mathcal{B}$ sends 
    $(\mathbf{A},\mathbf{B}_1,\mathbf{B}_2,\mathbf{G},
    \mathbf{C}_1,\ldots,\mathbf{C}_k,\mathbf{U},\mathbf{V},\mathbf{v})$ to 
    $\mathcal{A}_s$. 

    $\mathbf{Phase-1}$ 

   SerKG Query.
    $\mathcal{A}_s$ submits $ID_s$ to $\mathcal{B}$.
    Then, $\mathcal{B}$ runs 
    $\mathsf{SampleRight}(\mathbf{A}, \mathbf{G}',\mathbf{R}_1^*,\mathbf{T}_{\mathbf{G}},\mathbf{v},\rho)
    \to \mathbf{z}_{ID_s}$, 
    $\mathsf{SampleRight}(\mathbf{A}, \mathbf{G}',\mathbf{R}_1^*,\mathbf{T}_{\mathbf{G}},\mathbf{V},\rho)
    \to \mathbf{Z}_{ID_s}$
    where $\mathbf{G}' = (H_1(ID_s)-H_1(\widetilde{ID_u^*}))\cdot \mathbf{G}$,
    and sends $(\mathbf{z}_{ID_s},\mathbf{Z}_{ID_s})$ to $\mathcal{A}_s$. 
    Notably, we know that $H_1(ID_s)-H_1(ID_u^*)$ is full rank and 
    therefore $\mathbf{T}_{\mathbf{G}}$ is is also
    a trapdoor for the lattice $\Lambda_q^{\bot}(\mathbf{G}')$.   
    UserKG Query. 
    $\mathcal{A}_s$ submits $ID_u \neq ID_u^*$ to $\mathcal{B}$. 
    Then, 
    $\mathcal{B}$ runs $\mathsf{SampleBasisRight}(\mathbf{A},\mathbf{G}',\mathbf{R}_1^*,\mathbf{T}_{\mathbf{G}},\rho)
    \to \mathbf{T}_{\widetilde{ID_u}}$ 
    where $\mathbf{G}' = (H_1(\widetilde{ID_u})-H_1(\widetilde{ID^*_u}))\cdot \mathbf{G}$, and 
    sends $\mathbf{T}_{\widetilde{ID_u}}$
    to $\mathcal{A}_s$. 
    
    Token Query. 
    $\mathcal{A}_s$ submits $ID_u \neq ID_u^*$ to $\mathcal{B}$. 
    $\mathcal{B}$ sets $\mathbf{U}_{\theta,1}$, $\mathbf{U}_{\theta,2}$ 
    as $\mathbf{Game\ 2}$, runs 
    $\mathsf{SampleRight}(\mathbf{A},\mathbf{G}',\mathbf{R}^*_{1},\mathbf{T}_{\mathbf{G}},\mathbf{U}_{\theta,1},\rho)
    \to \mathbf{Z}_{ID_u,\theta}$ where 
    $\mathbf{G}' = (H_1(\widetilde{ID_u})-H_1(\widetilde{ID^*_u}))\cdot \mathbf{G}$, and 
    returns $(\theta,\mathbf{Z}_{ID_u,\theta})_{\theta\in \mathsf{path}(ID_u)}$ to $\mathcal{A}_s$.

    UpdKG Query. 
    $\mathcal{A}_s$ submits $(\mathbf{x},t)$ to $\mathcal{B}$. 
    If $t = t^*$, 
    $\mathcal{B}$ returns 
    $(\theta,\mathbf{Z}_{t^*,\theta} \cdot \mathbf{x})_{\theta\in \mathsf{KUNodes}(BT,RL_{\mathbf{x}},t^*)}$ 
    to $\mathcal{A}_s$; 
    otherwise, 
    for each $\theta \in \mathsf{KUNodes}(BT,RL_{\mathbf{x}},t)$, 
     $\mathcal{B}$ runs 
    $\mathsf{SampleRight}(\mathbf{A},\mathbf{G}',\mathbf{R}_{2}^*,\mathbf{T}_{\mathbf{G}},\mathbf{U}_{\theta,2},\rho) 
    \to \mathbf{Z}_{t,\theta}$, where $\mathbf{G}' = (H_1(t)-H_1(t^*))\cdot \mathbf{G}$, 
    and sends $(\theta,\mathbf{Z}_{t,\theta} \cdot \mathbf{x})_{\theta\in \mathsf{KUNodes}(BT,RL_{\mathbf{x}},t)}$ 
    to $\mathcal{A}_s$. 
    
    FunKG Query. 
    $\mathcal{A}_s$ submits $(ID_u,\mathbf{x},t)$ to $\mathcal{B}$, 
    with the restriction that $(ID_u,t)\neq (ID_u^*,t^*)$. 
    $\mathcal{B}$ runs $\mathsf{SampleRight}(\mathbf{A}, \mathbf{G}',\mathbf{R}',\mathbf{T}_{\mathbf{G}}, \mathbf{U},\rho)
    \to \mathbf{Z}_{\widetilde{ID_u},t}$ where 
    $\mathbf{R}' = \left[ \mathbf{R}_1^* | \mathbf{R}_2^* \right]$, 
    $\mathbf{G}' = (H_1(\widetilde{ID_u})-H_1(\widetilde{ID_u^*}) +H_1(t)-H_1(t^*))\cdot \mathbf{G}$, 
    and sends $fk_{ID_u,\mathbf{x},t} = \mathbf{Z}_{\widetilde{ID_u},t} \cdot \mathbf{x}$ to $\mathcal{A}_s$.  

    dTrapdoor Query. $\mathcal{A}_s$ submits $(ID_u,ID_s,\omega,t)$ to $\mathcal{B}$. 
    For $(ID_u,t,\omega)$, 
    $\mathbf{A}_{\widetilde{ID_u},\omega,t}  
    = \left[\mathbf{A} | \mathbf{A}\overline{\mathbf{R}} + \overline{\mathbf{G}} \right]$, 
    where $\overline{\mathbf{R}} = \left[ \mathbf{R}_1^* | \sum_{i = 1}^{k}(b_i\mathbf{F}_i^*)
    |\mathbf{R}_2^*\right]$ and 
    $\overline{\mathbf{G}} = (H_1(\widetilde{ID_u})-H_1(\widetilde{ID_u^*}) +1 + \sum_{i = 1}^{k}(b_ih_i) + H_1(t) -H_1(t^*)) \mathbf{G}$. 
    Then, $\mathcal{B}$ runs $\mathsf{SampleRight}(\mathbf{A},\mathbf{A}\overline{\mathbf{R}} + \overline{\mathbf{G}},
    \mathbf{T}_{\mathbf{G}},\mathbf{v},\rho) \to kt_{ID_u,\omega,t}$, 
    and 
    computes $dt_{ID_u,\omega,t}$ as $\mathbf{Game\ 0}$.
    Then, $\mathcal{B}$ sends 
    sends $dt_{ID_u,\omega,t}$ 
    to $\mathcal{A}_s$. 
    If $(ID_u,t) =(ID_u^*,t^*)$ and $\omega \notin Table_{\omega}$, $\mathcal{B}$ appends 
    $\omega$ to an initially empty table $Table_{\omega}$. 
    
    Revoke Query. $\mathcal{A}_s$ submits $RL_{\mathbf{x}},ID_u,t$ to $\mathcal{B}$. 
    Then, $\mathcal{B}$ runs $Revoke(ID_u,t,RL_{\mathbf{x}},st) \to RL_{\mathbf{x}}$, and sends $RL_{\mathbf{x}}$ to $\mathcal{A}_s$.

    $\mathbf{Challenge.}$ $\mathcal{A}_s$ submits $(\omega_0,\omega_1)$, with the restriction that $\omega_0,\omega_1 \notin Table_{\omega}$.  
    $\mathcal{B}$ flips a coin, gets a bit $b \in \{0,1\}$, 
    computes $\mathbf{c}_0,\mathbf{c}_1,\mathbf{c}_2$ as $\mathbf{Game\ 0}$, 
    and modifies the ciphertexts 
    $\mathbf{c}_3,\mathbf{c}_4,\mathbf{c}_5$ as follows. 
    Recall that $(\mathbf{u}_i,v_i)\in \mathbb{Z}_q^n\times \mathbb{Z}_q$ are entries from the LWE instance and let 
    $\mathbf{v}^* = \left[v_1,\ldots,v_m \right] \in \mathbb{Z}_q^{m}$.  
    Then, $\mathcal{B}$ selects $\mathbf{s}_3 \gets \mathbf{Z}_q^n$, $\mathbf{e}_5 \gets \mathcal{D}_{\mathbb{Z}^m,\sigma}$, 
    computes $\mathbf{c}_3 = (\mathbf{I}_m | \mathbf{R}_1^* | \sum_{i = 1}^{k} b_i \mathbf{F}_i^* | \mathbf{R}_2^*)^T 
    \cdot \mathbf{v}^*$, 
    $\mathbf{c}_4 = \mathbf{A}_{ID_s}^T \cdot \mathbf{s}_3 + 
    (\mathbf{I}_m | \mathbf{R}_7)^T \cdot\mathbf{e}_5$, 
    $\mathbf{c}_5 = v_0 + \mathbf{v}^T \cdot \mathbf{s}_3$, 
    and sends 
    $CT^* = (\mathbf{c}_0,\mathbf{c}_1,\mathbf{c}_2,\mathbf{c}_3,\mathbf{c}_4,\mathbf{c}_5)$ 
    to $\mathcal{A}_s$.  

    If $\mathcal{O} = \mathcal{O}_s$, 
    we have $\mathbf{v}^* = \mathbf{A}^T \cdot \mathbf{s}_2 + \mathbf{e}_4$ 
    and $v_0 = \mathbf{u}_0^T \cdot \mathbf{s}_2 +\mathbf{e}_6$, 
    where $\mathbf{e}_4 \gets \mathcal{D}_{\mathbb{Z}^{m},\sigma}$,
    $\mathbf{e}_6 \gets \mathcal{D}_{\mathbb{Z},\sigma}$. 
    When no abort happens, since $H_{kw}(\omega_b^*) = 0$, 
    we have 
    $\mathbf{A}_{\widetilde{ID_u^*},\omega_b^*,t^*} = 
    \left[\mathbf{A}|\mathbf{A}\mathbf{R}_1^*|\mathbf{A} 
    \cdot \sum^k_{i=1}b_i\mathbf{F}_i^*|\mathbf{A}\mathbf{R}_2^*\right]$. 
    Therefore, we selects $\mathbf{s}_3\gets \mathbb{Z}_q^n$, $\mathbf{e}_5 \gets \mathcal{D}_{\mathbb{Z}^m,\sigma}$, 
    $\mathbf{R}_7 \gets \{1,-1\}^{m \times m}$, and computes  
    $\mathbf{c}_3 = (\mathbf{I}_m | \mathbf{R}_1^* | \sum_{i = 1}^{k} b_i \mathbf{F}_i^* 
    | \mathbf{R}_2^*)^T \cdot \mathbf{v}^* 
    = \mathbf{A}_{\widetilde{ID_u^*},\omega_b^*,t^*}^T \cdot \mathbf{s}_2 
    + (\mathbf{I}_m | \mathbf{R}_1^* | \sum_{i = 1}^{k} b_i \mathbf{F}_i^* 
    | \mathbf{R}_2^*)^T \mathbf{e}_4, 
    \mathbf{c}_4 = \mathbf{A}_{ID_s}^T \cdot \mathbf{s}_3 + 
    (\mathbf{I}_m | \mathbf{R}_7)^T \cdot\mathbf{e}_5, 
    \mathbf{c}_5 = v_0 + \mathbf{v}^T \cdot \mathbf{s}_3 
    = \mathbf{v}^T \cdot \mathbf{s}_2 +\mathbf{e}_6 + \mathbf{v}^T \cdot \mathbf{s}_3
    = \mathbf{v}^T \cdot (\mathbf{s}_2+\mathbf{s}_3) +\mathbf{e}_6$.   
    
    Furthermore, the challenge ciphertexts are valid. 
    As a result, the challenge ciphertexts are the same as the challenge ciphertexts in $\mathbf{Game\ 3}$. 

    If $\mathcal{O} = \mathcal{O}'_s$, we have  
    $\mathbf{c}^*_3 = (\mathbf{I}_m | \mathbf{R}_1^* | \sum_{i = 1}^{k} b_i \mathbf{F}_i^* | \mathbf{R}_2^*)^T 
    \cdot \mathbf{v}^*$ is uniform 
    in $\mathbb{Z}_q^{4m}$ and  
    $\mathbf{c}^*_5 = v_0 + \mathbf{v}^T \cdot \mathbf{s}_3$
    is also uniform in $\mathbb{Z}_q$. 
    As a result, the challenge ciphertexts are the same as the challenge ciphertexts in $\mathbf{Game\ 4}$. 

    $\mathbf{Phase-2.}$
    $\mathcal{A}_s$ can continue to make SerKG, UserKG, Token, UpdKG, FunKG, dTrapdoor and Revoke queries, except 
    the UserKG query for $ID_u^*$ and 
    the dtrapdoor queries for $(ID_u^*,ID_s,\omega_0,t^*), (ID_u^*,ID_s,\omega_1,t^*)$.  
    $\mathcal{B}$ answers as $\mathbf{Phase-1}$.  
    
    $\mathbf{Guess.}$ $\mathcal{A}_s$ outputs a guess $b' \in \{0,1\}$ on $b$. 
    Then, $\mathcal{B}$ validates the guess 
    by checking whether $H_{kw}(\omega_{i}) \neq 0$ for 
    $i=1,2,\ldots,Q$ and $H_{kw}(\omega_{b}^*) = 0$ hold. 
    If these conditions are not satisfied, 
    refreshes $b'$ as a random bit, and aborts the game. 
    If $b=b'$, $\mathcal{B}$ returns guess: $\mathcal{O} = \mathcal{O}_s$;  
    otherwise, $\mathcal{B}$ returns guess: $\mathcal{O} = \mathcal{O}'_{s}$. 
    The $\mathcal{B}$'s advantage in breaking LWE assumption is 
    $\epsilon(\lambda) \geq
    Pr\left[ \neg \mathrm{abort} \right] \times 
    Pr\left[ \mathcal{B}\ \mathrm{wins} \right]
     =
    Pr\left[ \neg \mathrm{abort} \right] \times 
    (\frac{1}{2} \times Pr\left[ \mathcal{B}\ \mathrm{wins}|\mathcal{O} =\mathcal{O}_{s}\right] + 
    \frac{1}{2} \times Pr\left[\mathcal{B}\ \mathrm{wins}|\mathcal{O} =\mathcal{O}'_{s}\right]- 
    \frac{1}{2})
    =\frac{1}{4q } \times (\frac{1}{2}\times (\frac{1}{2} + \epsilon'(\lambda)) 
    + \frac{1}{2} \times \frac{1}{2} -\frac{1}{2})
    = \frac{\epsilon'(\lambda)}{8q}$. 
\end{proof}

\begin{theorem}
    Our scheme is $(t',\epsilon'(\lambda))$-KC-sIND-CPA security against $\mathcal{A}_o$ if the $(t,\epsilon(\lambda))$-LWE assumption holds, 
    where $t= O(t')$ and $\epsilon(\lambda) \geq \frac{\epsilon'(\lambda)}{2}$. 
\end{theorem}

\begin{proof}
Suppose a PPT adversary $\mathcal{A}_o$ (outside attacker, possibly including authorised users) has a non-negligible advantage against the KC-sIND-CPA security of our scheme. Then, $\mathcal{A}_o$ can be leveraged to construct an algorithm $\mathcal{B}$ that breaks the LWE assumption. As in Section \ref{section:Pre}, the LWE instance is provided via an oracle $\mathcal{O}$, which is either $\mathcal{O}_s$ or $\mathcal{O}_s'$.
    
    $\mathbf{Game\ 0:}$ The real game between $\mathcal{A}_o$ and $\mathcal{B}$ as shown 
    in section \ref{SM}. 

    $\mathbf{Game\ 1:}$ This game is the same as the above one, except for the following modifications. 
    Instead of running $\mathsf{TrapGen}$, we randomly select $\mathbf{A} \in \mathbb{Z}_q^{n\times m}$, 
    $\mathbf{R}_1^*,\mathbf{R}_2^*,\mathbf{F}_1^*,\ldots \mathbf{F}^*_k \in \{1,-1\}^{m\times m}$, 
    $h_1,\ldots h_k \in \mathbb{Z}_q$.
    Moreover, we  
    generate $\mathbf{G}$ by running $\mathsf{TrapGen}(n,m,q) \to (\mathbf{G},\mathbf{T}_{\mathbf{G}})$, 
    and sets $\mathbf{B}_1 = \mathbf{A} \mathbf{R}_1^* - H_1(ID_s^*)\mathbf{G}$, 
    $\mathbf{B}_2 = \mathbf{A} \mathbf{R}_2^* - H_1(t^*)\mathbf{G}$, 
    $\mathbf{C}_i  = \mathbf{A}\cdot \mathbf{F}^*_i + h_i \mathbf{G}$ for $i=1,\ldots k$. 

    $\mathbf{Game\ 2:}$ In this game, we change the way  
    $\mathbf{U}_{\theta,1}$ and $\mathbf{U}_{\theta,2}$ are generated. 

    For each node $\theta$, 
    we sample $\mathbf{Z}_{t^*,\theta} \in \mathcal{D}_{\mathbb{Z}^{m\times l},\rho}$, 
    compute $\mathbf{U}_{\theta,2} = \mathbf{A}_{t^*} \mathbf{Z}_{t^*,\theta}$, 
    $\mathbf{U}_{\theta,1} = \mathbf{U}-\mathbf{U}_{\theta,2}$. 

    $\mathbf{Game\ 3:}$ In this game, we reduce the KC-sIND-CPA security against $\mathcal{A}_o$ of 
    our scheme to the LWE assumption. 
    
    $\mathbf{Init.}$ $\mathcal{A}_o$ submits 
    chooses a challenged identity $ID_s^*$ and a timestamp $t^{*}$, and 
    submits $(ID_s^*,t^*)$ to $\mathcal{B}$. 
    
    $\mathbf{Setup.}$ $\mathcal{B}$ queries the LWE oracle $\mathcal{O}$, and obtains pairs $(\mathbf{u}_i,v_i)\in \mathbb{Z}_q^n\times \mathbb{Z}_q$ for  
    $i=0,\ldots m$.
    Let $\mathbf{v} = \mathbf{u}_0 \in \mathbb{Z}_q^n$ and $\mathbf{A} = \left[ \mathbf{u}_1,\ldots \mathbf{u}_m \right]\in \mathbb{Z}_q^{n\times m}$, and set  $(\mathbf{B}_1,\mathbf{B}_2,\mathbf{G},
    \mathbf{C}_1,\ldots, \mathbf{C}_k,\mathbf{U},\mathbf{V})$ as $\mathbf{Game\ 1}$.  
    Then, $\mathcal{B}$ sends $(\mathbf{A},\mathbf{B}_1,\mathbf{B}_2,\mathbf{G},\mathbf{C}_1,\ldots,
    \mathbf{C}_k,\mathbf{U},\mathbf{V},\mathbf{v})$ 
    to $\mathcal{A}_o$. 

    $\mathbf{Phase-1}$ 

    SerKG Query.
    $\mathcal{A}_o$ submits $ID_s \neq ID_s^*$ to $\mathcal{B}$.
    Then, $\mathcal{B}$ runs 
    $\mathsf{SampleRight}(\mathbf{A}, \mathbf{G}',\mathbf{R}_1^*,\mathbf{T}_{\mathbf{G}},\mathbf{v},\rho)
    \to \mathbf{z}_{ID_s}$, 
    $\mathsf{SampleRight}(\mathbf{A}, \mathbf{G}',\mathbf{R}_1^*,\mathbf{T}_{\mathbf{G}},\mathbf{V},\rho)
    \to \mathbf{Z}_{ID_s}$
    where $\mathbf{G}' = (H_1(ID_s)-H_1(ID_s^*))\cdot \mathbf{G}$,
    and sends $(\mathbf{z}_{ID_s},\mathbf{Z}_{ID_s})$ to $\mathcal{A}_o$. 
    Notably, we know that $H_1(ID_s)-H_1(ID_s^*)$ is full rank and 
    therefore $\mathbf{T}_{\mathbf{G}}$ is is also
    a trapdoor for the lattice $\Lambda_q^{\bot}(\mathbf{G}')$. 

    UserKG Query.  
    $\mathcal{A}_o$ submits $ID_u$ to $\mathcal{B}$. 
    Then, 
    $\mathcal{B}$ runs $\mathsf{SampleBasisRight}(\mathbf{A},\mathbf{G}',\mathbf{R}_1^*,\mathbf{T}_{\mathbf{G}},\rho)
    \to \mathbf{T}_{\widetilde{ID_u}}$ 
    where $\mathbf{G}' = (H_1(\widetilde{ID_u})-H_1(ID^*_s))\cdot \mathbf{G}$, and 
    sends $\mathbf{T}_{\widetilde{ID_u}}$
    to $\mathcal{A}_o$. 

    $\mathsf{Token\ Query}.$ 
    $\mathcal{A}_o$ submits $ID_u$ to $\mathcal{B}$. 
    $\mathcal{B}$ sets $\mathbf{U}_{\theta,1}$, $\mathbf{U}_{\theta,2}$ as
    $\mathbf{Game\ 2}$, runs 
    $\mathsf{SampleRight}(\mathbf{A},\mathbf{G}',\mathbf{R}^*_{1},\mathbf{T}_{\mathbf{G}},\mathbf{U}_{\theta,1},\rho)
    \to \mathbf{Z}_{ID_u,\theta}$ where 
    $\mathbf{G}' = (H_1(\widetilde{ID_u})-H_1(ID^*_s))\cdot \mathbf{G}$, and 
    returns $(\theta,\mathbf{Z}_{ID_u,\theta})_{\theta\in \mathsf{path}(ID_u)}$ to $\mathcal{A}_o$.

    UpdKG Query. 
    $\mathcal{A}_o$ submits $(\mathbf{x},t)$ to $\mathcal{B}$. 
    If $t = t^*$, $\mathcal{B}$ returns 
    $(\theta,\mathbf{Z}_{t^*,\theta} \cdot \mathbf{x})_{\theta\in \mathsf{KUNodes}(BT,RL_{\mathbf{x}},t^*)}$ 
    to $\mathcal{A}_o$; 
    otherwise, 
    for each $\theta \in \mathsf{KUNodes}(BT,RL_{\mathbf{x}},t)$, 
    $\mathcal{B}$ runs 
    $\mathsf{SampleRight}(\mathbf{A},\mathbf{G}',\mathbf{R}_{2}^*,\mathbf{T}_{\mathbf{G}},\mathbf{U}_{\theta,2},\rho) 
    \to \mathbf{Z}_{t,\theta}$, where $\mathbf{G}' = (H_1(t)-H_1(t^*))\cdot \mathbf{G}$, 
    and sends $(\theta,\mathbf{Z}_{t,\theta} \cdot \mathbf{x})_{\theta\in \mathsf{KUNodes}(BT,RL_{\mathbf{x}},t)}$ 
    to $\mathcal{A}_o$.  

    Revoke Query. $\mathcal{A}_o$ submits $RL_{\mathbf{x}},ID_u,t$ to $\mathcal{B}$. 
    Then, $\mathcal{B}$ runs $Revoke(ID_u,t,RL_{\mathbf{x}},st) \to RL_{\mathbf{x}}$, and sends $RL_{\mathbf{x}}$ to $\mathcal{A}_o$.

    $\mathbf{Challenge.}$ $\mathcal{A}_o$ submits two keywords $\omega_0$ and $\omega_1$ to $\mathcal{B}$.  
    $\mathcal{B}$ flips a coin, gets a bit $b \in \{0,1\}$, 
    computes $\mathbf{c}_0,\mathbf{c}_1,\mathbf{c}_2$ as $\mathbf{Game\ 0}$, 
    and modifies the ciphertexts 
    $\mathbf{c}_3,\mathbf{c}_4,\mathbf{c}_5$ as follows. 
    Recall that $(\mathbf{u}_i,v_i)$ are entries from the LWE instance 
    and let 
    $\mathbf{v}^* = \left[v_1,\ldots,v_m \right] \in \mathbb{Z}_q^{m}$. 
    Then, $\mathcal{B}$ selects $\mathbf{s}_2 \in \mathbf{Z}_q^n$, $\mathbf{R}_5 \gets \{1,-1\}^{m \times m}$ and $\mathbf{e}_4 \gets \mathcal{D}_{\mathbb{Z}^m,\sigma}$, and 
     computes $\mathbf{c}_3 = \mathbf{A}^T_{\widetilde{ID_u},\omega_b,t^*} \cdot \mathbf{s}_2 + (\mathbf{I}_m | \mathbf{R}_5 | \sum_{i = 1}^{k} b_i \mathbf{F}_i | \mathbf{R}^*_2)^T \cdot \mathbf{e}_4$, 
    $\mathbf{c}_4 = \left[ \mathbf{I}_m | \mathbf{R}_1^*  \right]^T \cdot \mathbf{v}^*$ 
    and $\mathbf{c}_5 = v_0 + \mathbf{v}^T \cdot \mathbf{s}_2$, and sends 
    $CT^* = (\mathbf{c}_0,\mathbf{c}_1,\mathbf{c}_2,\mathbf{c}_3,\mathbf{c}_4,\mathbf{c}_5)$ 
    to $\mathcal{A}_o$.

    $\mathbf{Phase-2.}$
    $\mathcal{A}_o$ can continue to make user SerKG, UserKG, Token, UpdKG and Revoke queries,  except for the SerKG
 query for $ID_s^*$, 
    and $\mathcal{B}$ answers as $\mathbf{Phase-1}$.  

    $\mathbf{Guess.}$ $\mathcal{A}_o$ outputs a guess $b' \in \{0,1\}$ on $b$. 
    If $b=b'$, $\mathcal{B}$ returns guess: $\mathcal{O} = \mathcal{O}_s$; 
    otherwise, $\mathcal{B}$ returns guess: $\mathcal{O} = \mathcal{O}'_{s}$.  
    The $\mathcal{B}$'s advantage in breaking LWE assumption is 
    $\epsilon(\lambda) \geq 
    Pr\left[ \mathcal{B}\ \mathrm{wins} \right]
     = 
    \frac{1}{2} \times Pr\left[ \mathcal{B}\ \mathrm{wins}|\mathcal{O} =\mathcal{O}_{s}\right] + 
    \frac{1}{2} \times Pr\left[\mathcal{B}\ \mathrm{wins}|\mathcal{O} =\mathcal{O}'_{s}\right]- 
    \frac{1}{2}
    = \frac{1}{2}\times (\frac{1}{2} + \epsilon'(\lambda)) 
    + \frac{1}{2} \times \frac{1}{2} -\frac{1}{2}
    = \frac{\epsilon'(\lambda)}{2}$. 
\end{proof}

\begin{theorem} \label{proof3}
    Our scheme is $(t',\epsilon'(\lambda))$-KT-sIND-CPA security if the $(t,\epsilon(\lambda))$-LWE assumption holds, 
    where $t= O(t')$ and $\epsilon(\lambda) \geq \frac{\epsilon'(\lambda)}{2}$. 
\end{theorem}

\begin{proof}
Assume there exists a PPT adversary $\mathcal{A}$ that wins the KT-sIND-CPA game with non-negligible advantage. We show that such an adversary can be transformed into an algorithm $\mathcal{B}$ that violates the LWE assumption. As in Section \ref{section:Pre}, $\mathcal{B}$ receives its LWE instance via an oracle $\mathcal{O}$, which is instantiated as either $\mathcal{O}_s$ or $\mathcal{O}_s'$.

$\mathbf{Game\ 0:}$ The real game between $\mathcal{A}$ and $\mathcal{B}$ as shown 
    in section \ref{SM}. 

    $\mathbf{Game\ 1:}$ This game is the same as the above one, except for the following modifications. 
    Instead of running $\mathsf{TrapGen}$, we randomly select $\mathbf{v} \in \mathbb{Z}_q^{n}$, 
    $\mathbf{R}_1^*,\mathbf{R}_2^*,\mathbf{F}_1^*,\ldots \mathbf{F}^*_k \in \{1,-1\}^{m\times m}$, 
    $h_1,\ldots h_k \in \mathbb{Z}_q$.
    Moreover, we  
    generate $\mathbf{G}$ by running $\mathsf{TrapGen}(n,m,q) \to (\mathbf{G},\mathbf{T}_{\mathbf{G}})$, 
    and sets $\mathbf{B}_1 = \mathbf{A} \mathbf{R}_1^* - H_1(ID_s^*)\mathbf{G}$, 
    $\mathbf{B}_2 = \mathbf{A} \mathbf{R}_2^* - H_1(t^*)\mathbf{G}$, 
    $\mathbf{C}_i  = \mathbf{A}\cdot \mathbf{F}_i + h_i \mathbf{G}$ for $i=1,\ldots k$. 

    $\mathbf{Game\ 2:}$ In this game, we change the way  
    $\mathbf{U}_{\theta,1}$ and $\mathbf{U}_{\theta,2}$ are generated. 

    For each node $\theta$, 
    we sample $\mathbf{Z}_{t^*,\theta} \in \mathcal{D}_{\mathbb{Z}^{m\times l},\rho}$, 
    compute $\mathbf{U}_{\theta,2} = \mathbf{A}_{t^*} \mathbf{Z}_{t^*,\theta}$, 
    $\mathbf{U}_{\theta,1} = \mathbf{U}-\mathbf{U}_{\theta,2}$. 

    $\mathbf{Game\ 3:}$ In this game, we reduce the KT-sIND-CPA security of 
    our scheme against $\mathcal{A}$ to the LWE assumption. 

    $\mathbf{Init.}$ $\mathcal{A}$ submits $(ID_s^*,t^*)$ to $\mathcal{B}$. 
    
    $\mathbf{Setup.}$ 
    $\mathcal{B}$ queries the LWE oracle $\mathcal{O}$, and obtains pairs $(\mathbf{u}_i,v_i)\in \mathbb{Z}_q^n\times \mathbb{Z}_q$ for  
    $i=1,\ldots m+h_1$.
    Let $\mathbf{A} = \left[ \mathbf{u}_1,\ldots \mathbf{u}_m \right]\in \mathbb{Z}_q^{n\times m}$ 
    and $\mathbf{V} = \left[ \mathbf{u}_{m+1},\ldots \mathbf{u}_{m+h_1} \right]\in \mathbb{Z}_q^{n \times h_1}$, and set 
    $(\mathbf{B}_1,\mathbf{B}_2,\mathbf{G},\mathbf{C}_1,\ldots, \mathbf{C}_k,\mathbf{U},\mathbf{v})$ as $\mathbf{Game\ 1}$.  
    Then, $\mathcal{B}$ sends $(\mathbf{A},\mathbf{B}_1,\mathbf{B}_2,\mathbf{G},\mathbf{C}_1,\ldots,\mathbf{C}_k,\mathbf{U},\mathbf{V},\mathbf{v})$ 
    to $\mathcal{A}$. 

    $\mathbf{Phase-1}$ 

    SerKG Query.
    $\mathcal{A}$ submits $ID_s \neq ID_s^*$ to $\mathcal{B}$.
    Then, $\mathcal{B}$ runs 
    $\mathsf{SampleRight}(\mathbf{A}, \mathbf{G}',\mathbf{R}_1^*,\mathbf{T}_{\mathbf{G}},\mathbf{v},\rho)
    \to \mathbf{z}_{ID_s}$, 
    $\mathsf{SampleRight}(\mathbf{A}, \mathbf{G}',\mathbf{R}_1^*,\mathbf{T}_{\mathbf{G}},\mathbf{V},\rho)
    \to \mathbf{Z}_{ID_s}$
    where $\mathbf{G}' = (H_1(ID_s)-H_1(ID_s^*))\cdot \mathbf{G}$,
    and sends $(\mathbf{z}_{ID_s},\mathbf{Z}_{ID_s})$ to $\mathcal{A}_o$. 
    Notably, we know that $H_1(ID_s)-H_1(ID_s^*)$ is full rank and 
    therefore $\mathbf{T}_{\mathbf{G}}$ is is also
    a trapdoor for the lattice $\Lambda_q^{\bot}(\mathbf{G}')$.  

    UserKG Query.  
    $\mathcal{A}$ submits $ID_u$ to $\mathcal{B}$. 
    Then, 
    $\mathcal{B}$ runs $\mathsf{SampleBasisRight}(\mathbf{A},\mathbf{G}',\mathbf{R}_1^*,\mathbf{T}_{\mathbf{G}},\rho)
    \to \mathbf{T}_{\widetilde{ID_u}}$ 
    where $\mathbf{G}' = (H_1(\widetilde{ID_u})-H_1(ID^*_s))\cdot \mathbf{G}$, and 
    sends $\mathbf{T}_{\widetilde{ID_u}}$
    to $\mathcal{A}$. 

    $\mathsf{Token\ Query}.$ 
    $\mathcal{A}$ submits $ID_u$ to $\mathcal{B}$. 
    $\mathcal{B}$ sets $\mathbf{U}_{\theta,1}$, $\mathbf{U}_{\theta,2}$ as
    $\mathbf{Game\ 2}$, runs 
    $\mathsf{SampleRight}(\mathbf{A},\mathbf{G}',\mathbf{R}^*_{1},\mathbf{T}_{\mathbf{G}},\mathbf{U}_{\theta,1},\rho)
    \to \mathbf{Z}_{ID_u,\theta}$ where 
    $\mathbf{G}' = (H_1(\widetilde{ID_u})-H_1(ID^*_s))\cdot \mathbf{G}$, and 
    returns $(\theta,\mathbf{Z}_{ID_u,\theta})_{\theta\in \mathsf{path}(ID_u)}$ to $\mathcal{A}$. 

    UpdKG Query. 
    $\mathcal{A}$ submits $(\mathbf{x},t)$ to $\mathcal{B}$. 
    If $t = t^*$, $\mathcal{B}$ returns 
    $(\theta,\mathbf{Z}_{t^*,\theta} \cdot \mathbf{x})_{\theta\in \mathsf{KUNodes}(BT,RL_{\mathbf{x}},t^*)}$ 
    to $\mathcal{A}$; 
    otherwise, 
    for each $\theta \in \mathsf{KUNodes}(BT,RL_{\mathbf{x}},t)$, 
    $\mathcal{B}$ runs 
    $\mathsf{SampleRight}(\mathbf{A},\mathbf{G}',\mathbf{R}_{2}^*,\mathbf{T}_{\mathbf{G}},\mathbf{U}_{\theta,2},\rho) 
    \to \mathbf{Z}_{t,\theta}$, where $\mathbf{G}' = (H_1(t)-H_1(t^*))\cdot \mathbf{G}$, 
    and sends $(\theta,\mathbf{Z}_{t,\theta} \cdot \mathbf{x})_{\theta\in \mathsf{KUNodes}(BT,RL_{\mathbf{x}},t)}$ 
    to $\mathcal{A}$.  

    Revoke Query. $\mathcal{A}$ submits $RL_{\mathbf{x}},ID_u,t$ to $\mathcal{B}$. 
    Then, $\mathcal{B}$ runs $Revoke(ID_u,t,RL_{\mathbf{x}},st) \to RL_{\mathbf{x}}$, and sends $RL_{\mathbf{x}}$ to $\mathcal{A}$.



    $\mathbf{Challenge.}$ $\mathcal{A}$ submits two keywords $\omega^*_0$ and $\omega^*_1$ to $\mathcal{B}$.  
    $\mathcal{B}$ flips a coin, gets a bit $b \in \{0,1\}$, 
    and modifies the keyword trapdoor  
    $\mathbf{kt}_1,\mathbf{kt}_2,\mathbf{kt}_3$ as follows. 
    Recall that $(\mathbf{u}_i,v_i)$ are entries from the LWE instance, 
    $\mathbf{u}^* = \left[v_1,\ldots v_{m}\right] \in \mathbb{Z}_q^{m}$ and 
    $\mathbf{v}^* = \left[v_{m+1},\ldots v_{m+h_1}\right] \in \mathbb{Z}_q^{h_1}$. 
    $\mathcal{B}$ runs $\mathsf{SampleBasisRight}(\mathbf{A},\mathbf{B}_{\widetilde{ID^*_u}},\mathbf{T}_{\mathbf{G}},\rho)
    \to \mathbf{T}_{\widetilde{ID^*_u}}$, $\mathsf{SampleLeft}(\mathbf{A}_{\widetilde{ID^*_u}}, \left[ \mathbf{A}_{\omega_b^*} | \mathbf{A}_{t^*} \right],
    \mathbf{T}_{\widetilde{ID^*_u}},\mathbf{v},\rho) \to  kt_{ID^*_u,\omega_b^*,t^*}$, 
    selects $kx \gets \mathbb{Z}_q $, 
    computes $\mathbf{kt}_1 = \left[ \mathbf{I}_m | \mathbf{R}_1^*  \right]^T \cdot \mathbf{u}^*$, 
    $\mathbf{kt}_2 = \mathbf{v}^* + \mathsf{bin}(kx) \cdot \left\lfloor \frac{q}{2} \right\rfloor$,  
    $\mathbf{kt}_3 = H_2(kx) \oplus bin(kt_{ID_u,\omega,t})$, 
    and returns $dt^* = (\mathbf{kt}_1,\mathbf{kt}_2,\mathbf{kt}_3)$ 
    to $\mathcal{A}$.

    $\mathbf{Phase-2.}$
    $\mathcal{A}$ can continue to make SerKG, UserKG, Token, UpdKG and Revoke queries, except for the SerKG
 query for $ID_s^*$, 
    and $\mathcal{B}$ answers as $\mathbf{Phase-1}$.  

    $\mathbf{Guess.}$ $\mathcal{A}$ outputs a guess $b' \in {0,1}$ for the bit $b$.
If $b' = b$, then $\mathcal{B}$ concludes that $\mathcal{O} = \mathcal{O}_s$;
otherwise, $\mathcal{B}$ outputs that the oracle is $\mathcal{O} =\mathcal{O}'_{s}$. 
    
    The $\mathcal{B}$'s advantage in breaking LWE assumption is 
    $\epsilon(\lambda) \geq 
    Pr\left[ \mathcal{B}\ \mathrm{wins} \right]
     = 
    \frac{1}{2} \times Pr\left[ \mathcal{B}\ \mathrm{wins}|\mathcal{O} =\mathcal{O}_{s}\right] + 
    \frac{1}{2} \times Pr\left[\mathcal{B}\ \mathrm{wins}|\mathcal{O} =\mathcal{O}'_{s}\right]- 
    \frac{1}{2}
    = \frac{1}{2}\times (\frac{1}{2} + \epsilon'(\lambda)) 
    + \frac{1}{2} \times \frac{1}{2} -\frac{1}{2}
    = \frac{\epsilon'(\lambda)}{2}$. 
\end{proof}

\section{Conclusion} \label{section:Con}
In this paper, an EQDDA-RKS scheme was proposed, where data users can securely perform keyword searches over ciphertexts and compute the inner product values of encrypted data. 
In EQDDA-RKS, when a data user’s function key is compromised, it can be revoked. 
Meanwhile, data users can temporarily delegate their keyword search and function computation rights to others. 
We formalised the definition and security model of our EQDDA-RKS scheme and presented a concrete construction.  
Moreover, we conducted the theoretical and experimental analysis of our EQDDA-RKS scheme and 
proved its security.
Notably, in our scheme, data users only need to interact with the CA once, making it suitable for MCC scenarios.
For future work, we plan to design an EQDDA-RKS scheme under the Ring Learning with Errors (RLWE) assumption to further improve efficiency. 

\section*{Acknowledgments}
This work was supported by the National Natural Science Foundation of China (Grant No. 62372103), the Natural Science Foundation of Jiangsu Province (Grant No. BK20231149), and the Jiangsu Provincial Scientific Research Center of Applied Mathematics (Grant No.BK202330020).

\bibliographystyle{IEEEtran}
\bibliography{paper}

\end{document}